\documentclass[journal,onecolumn]{IEEEtran}
\usepackage{cite}
\usepackage{amsmath,amssymb,amsthm,bm,geometry,verbatim}
\usepackage{booktabs,tabularx}
\usepackage{array}
\usepackage[final,pagebackref,colorlinks=true,linkcolor=blue,%
urlcolor=blue,citecolor=red,pdfstartview=FitH]{hyperref}

\ifCLASSOPTIONcompsoc
  \usepackage[caption=false,font=normalsize,labelfont=sf,textfont=sf]
             {subfig}
\else
  \usepackage[caption=false,font=footnotesize]{subfig}
\fi

\usepackage{url}

\allowdisplaybreaks

\newtheorem{theorem}{Theorem}
\newtheorem{lemma}{Lemma}
\newtheorem{corollary}{Corollary}
\newtheorem{proposition}{Proposition}
\newtheorem{note}{Note}
\newtheorem{definition}{Definition}

\newtheorem{result}{Result}

\newtheorem{fact}{Fact}
\newtheorem{remark}{Remark}

\newcommand{\ket}[1]{| #1 \rangle}
\newcommand{\bra}[1]{\langle #1 |}
\newcommand{\cB}{\mathcal{B}}
\newcommand{\C}{\mathbb{C}}
\newcommand{\E}{\mathbb{E}}
\newcommand{\I}{\mathbb{I}}

\newcommand{\prob}{\mathbb{P}}

\newcommand{\U}{\mathbb{U}}

\newcommand{\cD}{\mathcal{D}}
\newcommand{\cG}{\mathcal{G}}
\newcommand{\cH}{\mathcal{H}}
\newcommand{\cI}{\mathcal{I}}
\newcommand{\cL}{\mathcal{L}}

\newcommand{\cT}{\mathcal{T}}

\newcommand{\hcT}{\hat{\mathcal{T}}}

\newcommand{\tcT}{\tilde{\mathcal{T}}}
\newcommand{\trho}{\tilde{\rho}}
\newcommand{\tomega}{\tilde{\omega}}

\newcommand{\vectorise}{\mathrm{vec}}
\newcommand{\Tr}{\mathrm{Tr}\,}

\newcommand{\Haar}{\mathrm{Haar}}
\newcommand{\Hmax}{H_\mathrm{max}}
\newcommand{\TPE}{\mathrm{TPE}}
\newcommand{\poly}{\mathrm{poly}}
\newcommand{\polylog}{\mathrm{polylog}}

\newcommand{\one}{\leavevmode\hbox{\small1\kern-3.8pt\normalsize1}}

\begin{document}
\title{High probability decoupling via approximate unitary designs
and efficient relative thermalisation}

\author{Aditya~Nema and~Pranab~Sen% <-this % stops a space
\thanks{A. Nema is currently with the Institute of Quantum Information, RWTH Aachen University, Germany. The first draft of this work was submitted while pursuing PhD at Schol of Technology and Computer Science, Tata Institute of Fundamental Research , Mumbai and first two revisions were submitted during postdoc at the Department of Mathematical Informatics,
Nagoya University, Japan.
e-mail: aditya.nema.30@gmail.com.}% <-this % stops a space
\thanks{P. Sen is currently with Centre for Quantum Technologies, National University of Singapore and School of Technology and was at Computer Science,
Tata Institute of Fundamental Research, Mumbai India during earlier submissions.
e-mail: pranab.sen.73@gmail.com.}% <-this % stops a space
\thanks{
Manuscript received June 13, 2021.}%
}

% The paper headers
\markboth{IEEE Transactions on Information Theory}%
{High probability decoupling via approximate unitary designs}

\maketitle

\begin{abstract}
We prove a new concentration result for non-catalytic
decoupling by showing that, for suitably large $t$, applying a 
unitary chosen uniformly at random from an approximate $t$-design on
a quantum system
followed by a fixed quantum operation almost decouples, with high
probability, the given
system from another reference system to which it may initially have been
correlated. Earlier works either did not obtain high decoupling 
probability, or used provably inefficient unitaries,
or required catalytic entanglement for decoupling. In contrast,
our approximate unitary designs always guarantee decoupling with 
exponentially high probability and, under certain conditions, 
lead to computationally efficient unitaries.
As a result we conclude that, under suitable conditions, efficiently
implementable approximate unitary designs
achieve relative thermalisation in quantum thermodynamics with 
exponentially high probability.  We also show the scrambling property of black hole, when the black hole evolution is according to pseudorandom approximate unitary $t$-design, as opposed to the Haar random evolution considered earlier by Hayden-Preskill.
\end{abstract}

\IEEEpeerreviewmaketitle

%\keywords{Decoupling theorems, Unitary designs, Fully Quantum Slepian 
%Wolf theorem, Relative Thermalisation, Concentration of 
%measure}%Use showkeys class option if keyword
                              %display desired

\section{Introduction}
\IEEEPARstart{A}{}
peculiar characteristic of quantum information theory is that many 
information transmission protocols, be it compression of quantum
messages or 
sending quantum information through unassisted quantum channels,
can be constructed by first removing 
correlations of a particular system from some other system around it. This
behooves us to prove general theorems
that take a bipartite quantum state shared
between a system $A$ (e.g. the ``particular system'' above) and a 
reference $R$
(e.g. the ``some other system'' above), apply a local operation on $A$,
and then, if suitable conditions are met, prove that the resulting
state is close to a product state between the output system $B$ and the
untouched reference $R$. 
This process of removing quantum correlations, i.e. obtaining a state
close to a product state, is 
referred to as decoupling.
Decoupling theorems play a vital role in proving 
achievability bounds for several quantum information theory protocols
as well as thermalization results in quantum thermodynamics.
In particular, the so-called Fully Quantum Slepian Wolf (FQSW) protocol
\cite{mother_protocol},
which has been hailed as the mother protocol of quantum information
theory, is constructed via a decoupling argument. In the FQSW problem,
the system $A$ is thought of as a bipartite system $A = A_1 \otimes A_2$
and the fixed superoperator is nothing but tracing out $A_2$. 
The FQSW protocol
is used as a building block for many other 
important protocols in quantum information theory in the asymptotic
iid setting e.g. noisy teleportation, noisy super dense coding, 
distributed compression,
entanglement unassisted and assisted quantum channel coding, one way
entanglement distillation, reverse Shannon theorem etc. Asymptotic
iid setting means that the given messages / channels are of the
tensor power form $(\cdot)^{\otimes n}$ for large $n$.
However the basic decoupling and FQSW results are actually one-shot
results where the given message / channel is to be used only once.
The one shot FQSW result
can be extended to obtain a one-shot relative thermalization 
result in quantum
thermodynamics \cite{RelativeThermalization}, where a system
$\Omega \leq S \otimes E$, with $S$ being the subsystem of physical
interest and $E$ being the so-called `environment' or `bath' subsystem, 
initially starts out in a correlated state together with a reference 
system $R$ but very soon evolves into something close to a so-called 
`relative thermal state' on $S$ tensored with the reduced 
state on the reference $R$.

In this paper, we build on the following important decoupling theorem 
proved by Dupuis in his doctoral thesis~\cite{decoupling}.
\begin{fact}
\label{fact:dupuisexpectation}
Consider a quantum state 
$\rho^{AR}$ shared between a system $A$ and a reference $R$. Let 
$\cT^{A \to B}$ be a completely positive trace preserving superoperator
(aka CPTP map aka quantum operation) with input system $A$ and output
system $B$.
Let $U$ be a unitary on the system $A$. Define the function
\[
f(U) := 
\lVert
(\cT^{A \to B} \otimes \I^R)(
(U^A \otimes I^R) \rho^{AR} (U^{A \dagger} \otimes I^R)
)
- \omega^B \otimes \rho^R
\rVert_1,
\]
where $\I^R$ is the identity superoperator on $R$ and 
$I^R$ is the identity operator on $R$. 
Let $A'$ be a new system having the same dimension as $A$.
Define the Choi-Jamio{\l}kowski state
$
\omega^{A'B} := 
(\cT^{A \to B} \otimes \I^{A'})(
\ket{\Phi}\bra{\Phi}^{AA'}
),
$
of $\cT^{A \to B}$ where
$\ket{\Phi}^{AA'} := |A|^{-1/2} \sum_a \ket{a}^A \otimes \ket{a}^{A'}$ is
the standard normalized EPR state on system $A A'$. 
Then,
$
\E_{U^A \sim \Haar}[f(U)] \leq
2^{
-\frac{1}{2} H_2(A|R)_\rho - \frac{1}{2}H_2(A'|B)_\omega
},
$ 
where the expectation is taken over the Haar measure on unitary
operators on $A$. 
The quantity $H_2(\cdot | \cdot)$ is the
conditional R\'{e}nyi $2$-entropy 
defined in Definition~\ref{def:Renyi2} below. 
We remark that 
$H_2(A|R)_\rho = -2\log \lVert \tilde{\rho}^{AR} \rVert_2$ and 
$H_2(A^\prime|B)_\omega = -2\log \lVert \tilde{\omega}^{A'B} \rVert_2$, 
where 
$\tilde{\rho}^{AR}$ and $\tilde{\omega}^{A'B} $ are certain 
positive semidefinite matrices defined in Definition~\ref{def:Renyi2}.
\end{fact}
Informally speaking, the above theorem
states that if some entropic 
conditions are met then, in {\em expectation}, the state 
$\sigma^{BR}$ obtained by first applying a Haar random unitary $U^A$ on 
the initial state $\rho^{AR}$ followed by a CPTP map $\cT^{A \to B}$ is 
close to the decoupled state $\omega^B \otimes \rho^R$. Here,
$\omega^B = \cT^{A \to B}(\frac{I^A}{|A|})$ is the state obtained
by applying $\cT$ to the completely mixed state on $A$.
In fact, $\omega^{A'B}$ defined in Fact~\ref{fact:dupuisexpectation} 
above is nothing
but the Choi-Jamio{\l}kowski state corresponding to the CPTP map 
$\cT^{A \to B}$.
Intuitively, a Haar random unitary $U^A$ `randomizes' or `scrambles' 
the state on
$A$ to give the completely mixed state which is then sent to $\omega^B$
by $\cT^{A \to B}$. So it is reasonable to believe that the local
state on $B$ should be $\omega^B$.
Notice that the local state on $R$ after applying $U^A$ and 
$\cT^{A \to B}$ is always $\rho^R$. The punch of the decoupling theorem 
is that the global state is close to the desired tensor product
state. 

The distance of the actual global state from the desired
tensor product state is upper bounded by two quantities. The
first quantity $H_2(A|R)_\rho$
is usually negative, which signifies that $A$ and $R$ are entangled
in the initial state $\rho$. To decouple $A$ from $R$ we start by
applying a Haar random unitary $U$ to the system $A$.
A single unitary cannot decouple $A$ from $R$, and that is why the
decoupling theorem above also has the CPTP map $\cT$. Now in an
intuitive sense, the EPR state $\Phi^{AA'}$ is the `most entangled
state'. So if a Haar random unitary $U$ on the system $A$ of
$\Phi^{AA'}$ followed by the CPTP 
map $\cT$ can decouple the output system $B$ from $R$, then it must
be able to decouple $B$ from $R$ when the input is any entangled
state $\rho^{AR}$, provided that the `amount of entanglement'
$H_2(A|R)_\rho$ between
$A$ and $R$ in $\rho$ is less than the `amount of entanglement' 
$H_2(A'|B)_\omega$ between
$A'$ and $B$ in $\omega$. This explains the expression 
$H_2(A'|B)_\omega + H_2(A|R)_\rho$
in the above upper bound. To counteract a negative
$H_2(A|R)_\rho$, the quantity $H_2(A'|B)_\omega$ had better be
positive which signifies that $A'$ is mostly decoupled from $B$
in the state $\omega$.

Dupuis showed in his doctoral thesis how the decoupling 
theorem above can be used to recover in a unified fashion several 
previously known results as well as obtain some totally new results
in quantum information theory.
Szehr et al.~\cite{Smooth_decoupling} extended the decoupling theorem
by showing that the expectation can be taken over approximate unitary
$2$-designs (defined formally in Definition~\ref{def:design} below) 
instead of over Haar random unitaries. 
The advantage of unitary $2$-designs
is that efficient constructions for them exist unlike the case with
Haar random unitaries.
Szehr et al. also upper bounded the expected trace distance
in terms of smooth entropic quantities which have better mathematical
properties compared to  the non-smooth ones. In particular, in the 
asymptotic iid limit, the smooth entropic quantities are suitably
bounded by $n$ times the corresponding Shannon entropies which is not
the case with the non-smooth quantities. Their result (adapted to our
notations) is stated below.
\begin{fact}
\label{fact:szehrexpectation}
Under the setting of Fact~\ref{fact:dupuisexpectation} above,
$$
\E_{U^A \sim \Haar}[f(U)] \leq
2^{
-\frac{1}{2} H^\epsilon_{2}(A|R)_\rho 
- \frac{1}{2}H^\epsilon_{2}(A'|B)_\omega
}+12\epsilon,
$$
where the expectation is taken over the Haar measure on unitary
operators on $A$. The same result holds if the expectation is taken
over the uniform choice of a unitary from an exact 2-design. 
If the expectation is taken over the uniform choice of a unitary from a 
$\delta$-approximate 2-design,
the upper bound
gets multiplied by a multiplicative factor dependent on the dimension
of $A$ and $\delta$.
The smooth conditional R\'{e}nyi $2$-entropy terms appearing in the 
bound are defined in \cite{decoupling}.
\end{fact}

In a different vein Anshu and Jain~\cite{decoupling_convexsplit} showed,
extending earlier work by Ambainis and Smith~\cite{AmbainisSmith},
that it is possible to add a small ancilla $C$ in tensor product with
$A$, apply an efficient
unitary to $A \otimes C$ and then trace out $C$ so that $A$ is now
decoupled from $R$ even before applying the
CPTP map $\cT$. The difference between Ambainis and Smith's or Anshu
and Jain's works on one hand, and Dupuis', Szehr et
al.'s or our works on the other hand is that we want a single unitary 
on the system $A$ to achieve decoupling
and not the average of a number of unitaries on $A$ or, more generally, 
a unitary on a
larger system $A \otimes C$. A single
unitary cannot decouple $A$ from $R$. That is why the
decoupling theorem above also has the CPTP map $\cT$. 
The single-unitary-followed-by-CPTP-map form of the decoupling
theorem is required for quantum Shannon theory applications where there 
is no entanglement
assistance e.g. sending quantum information over an unassisted quantum
channel.

After obtaining the decoupling result in expectation above, it is
natural to ask whether such a theorem also holds with high probability 
over the choice of the random unitary $U^A$. Dupuis~\cite{decoupling}
answered this question in the affirmative for the Haar measure. That
result, adapted to our notation, is as follows:
\begin{fact}
\label{fact:dupuisconcentration}
Under the setting of Fact~\ref{fact:dupuisexpectation} above, we have
\[
\prob_{U^A \sim \Haar}[
f(U) >
2^{
-\frac{1}{2} H_2^\epsilon(A|R)_\rho - \frac{1}{2}H_2^\epsilon(A'|B)_\omega
} + 16 \epsilon + \delta
] \leq
2 \exp\left(
-\frac{|A| \delta^2}{2^{H_{\mathrm{min}}^\epsilon(A)_\rho + 4}}
\right),
\]
where the smooth min-entropy $H_{\mathrm{min}}^\epsilon(\cdot)$ is
defined in Definition~\ref{def:Hmin} below,
and the probability is taken over the Haar
measure on $\U^A$.
\end{fact}
The concentration of measure result for the decoupling theorem above
immediately implies an exponential concentration result for the 
FQSW problem, which further implies
that relative thermalization occurs for a system in contact with a
heat bath for all but an exponentially small fraction of unitary
evolutions of the system as long as the system is assumed to evolve
according to a Haar random unitary. However this is not a very 
satisfactory explanation from a physical and computational point of view 
as Haar random
unitaries are provably impossible to implement by quantum
circuits with size polylogarithmic in the dimension of the system. Also,
Haar random unitaries on a system $A$ require $\Omega(|A|^2 \log |A|)$ 
number of random bits for a precise description. This leads us to wonder
if relative thermalization can be achieved with high probability by
simpler unitary evolutions of the system $A$. 
Nakata et al.~\cite{Winter_decoupling} gave an affirmative answer
by showing that decoupling can be achieved by choosing 
products of random unitaries
diagonal in the Pauli $X$ and $Z$ bases, but even they are not 
efficiently implementable and in addition, require
$\Omega(|A| \log |A|)$ random bits for a precise description.
Moreover the fraction of such
unitaries which achieve decoupling is not strongly concentrated near
one.

\subsection{Our results}
In this paper we prove for the first time that, for suitable values of
$t$, {\em approximate unitary 
$t$-designs achieve decoupling with probability
exponentially close to one}.  An exact $t$-design of $n \times n$
unitaries can be
described using $O(t \log n)$ random bits \cite{Kuperberg} as opposed to 
$\Omega(n^2 \log n)$ random bits required to describe a Haar random
unitary to reasonable precision. Thus for many applications our 
result implies a substantial saving in the number of random bits 
compared to Dupuis' result. However, the concentration guaranteed by
our result is less than that guaranteed by Dupuis even though it is
exponential. 
Our concentration bound
for decoupling via unitary designs is expressed in terms of smooth
entropic quantities. An informal version is stated below:
\begin{result}
\label{result:main}
Consider a quantum state 
$\rho^{AR}$ shared between a system $A$ and a reference $R$. Let 
$\cT^{A \to B}$ be a completely positive trace preserving superoperator
with input system $A$ and output
system $B$.
Let $U$ be a unitary on the system $A$. Define the function
\[
f(U) := 
\lVert
(\cT^{A \to B} \otimes \I^R)(
(U^A \otimes I^R) \rho^{AR} (U^{A \dagger} \otimes I^R)
)
- \omega^B \otimes \rho^R
\rVert_1,
\]
where $\I^R$ is the identity superoperator on $R$ and 
$I^R$ is the identity operator on $R$. 
Let $A'$ be a new system having the same dimension as $A$.
Define the Choi-Jamio{\l}kowski state
$
\omega^{A'B} := 
(\cT^{A \to B} \otimes \I^{A'})(
\ket{\Phi}\bra{\Phi}^{AA'}
),
$
of $\cT^{A \to B}$ where
$\ket{\Phi}^{AA'} := |A|^{-1/2} \sum_a \ket{a}^A \otimes \ket{a}^{A'}$ is
the standard EPR state on system $A A'$. 
Let $0 < \epsilon, \delta < 1/3$. Let $\kappa > 0$.
Then,
\begin{equation}
\label{eq:main_decoupling}
\prob_{U^A \sim \mathrm{design}}[
f(U) >
2^{
-\frac{1}{2} H_2^\epsilon(A|R)_\rho 
-\frac{1}{2} \mathbb{H}_2^{\epsilon,\delta}(A'|B)_\omega + 1
} + 14 \sqrt{\epsilon} + 2 \kappa
] \leq 
7 \cdot 2^{-a \kappa^2}.
\end{equation}
where the matrix $U^A$ is chosen uniformly at random from an 
approximate $t$-design of unitaries,
$
a := 
|A| \cdot
2^{
-(1+\delta)\mathbb{H}_{\max}^\epsilon(B)_\omega +
H_2^{\epsilon}(A|R)_\rho - 9
},
$
$
t := 8 a \kappa^2.
$
The smooth entropies $\mathbb{H}_{\max}^\epsilon(\cdot)$,
$\mathbb{H}_2^{\epsilon,\delta}(\cdot|\cdot)$, 
$H_2^{\epsilon}(\cdot|\cdot)$ are defined in
Definitions \ref{def:Hmaxprime}, \ref{def:Renyi2prime},
\ref{def:Renyi2} respectively below.
Since the result holds for all $\kappa>0$, which is the parameter for 
deviation from the expectation value of $f(U)$, one can replace 
$\kappa$ in terms of the parameter $t$ of our approximate unitary 
$t$-design by $\sqrt{\frac{t}{8a}}$. This can be done once the analysis 
is done for a particular value of $\kappa>0$. Further, we chose 
parameter  $\beta:=\sqrt{\frac{1}{a}}$. We thus get probability of 
the decoupling Equation~\ref{eq:main_decoupling} exponential in `$t$', as:
\begin{equation}
\label{eq:decoupling_parameterized}
\prob_{U^A \sim \mathrm{design}}[
f(U) >
2^{
-\frac{1}{2} H_2^\epsilon(A|R)_\rho 
-\frac{1}{2} \mathbb{H}_2^{\epsilon,\delta}(A'|B)_\omega + 1
}  + 14 \sqrt{\epsilon} + \sqrt{t/2} \, \beta
] \leq 
7 \cdot 2^{-\frac{t}{8}}.
\end{equation}
\end{result}  

The three smooth one-shot
entropic terms used in Result~\ref{result:main} approach the standard
Shannon entropic terms in the asymptotic iid limit. 
\begin{remark} 
\label{rem:parameterization}
Note that Equations~\ref{eq:main_decoupling} and 
\ref{eq:decoupling_parameterized} are exactly the same. Moreover, 
Equation~\ref{eq:main_decoupling} represents a measure concentration 
result or the so-called tail probability (commonly used for large 
deviation analysis) for the decoupling function $f$. 
Equation~\ref{eq:main_decoupling} provides a more intuitive relation of 
the decoupling phenomenon with approximate unitary $t$-designs. 
This way of writing also indicates that we have exponentially high 
concentration in the decoupling theorem, exponential in $t$.
\end{remark}
We can thus
infer the following corollary
of our main result in the asymptotic iid setting. We state an informal
version of the corollary below.
\begin{result}
\label{result:corollary}
Consider the setting of Result~\ref{result:main} above.
Let $n$ be a large enough positive integer.
Consider the $n$-fold tensor powers 
$\omega^{(A')^n B^n} := (\omega^{A'B})^{\otimes n}$,
$\rho^{A^n R^n} := (\rho^{AR})^{\otimes n}$.
Let $\epsilon' := 8(n + |A||B|)^{|A||B|} \epsilon^{1/4}$.
Let $\kappa > 0$. Then,
\begin{align*}
\underset{U}{\prob}\bigg[
f(U) >
2^{
-\frac{n}{2} (H(A|R)_\rho - \delta(3 H(AR)_\rho + 7 H(R)_\rho))
-\frac{n}{2} (H(A'|B)_\omega - \delta(3 H(A'B)_\omega + 7 H(B)_\omega))
}\\ 
+ 28 (\epsilon')^{1/4} + 2 \kappa
\bigg] 
\leq 7 \cdot 2^{-a \kappa^2},~~~~~~~&
\\\hspace{0.9\textwidth}
\end{align*}
which can also be expressed as (following the explanation from 
Remark~\ref{rem:parameterization}):
\begin{align*}
\underset{U}{\prob}\bigg[
f(U) >
2^{
-\frac{n}{2} (H(A|R)_\rho - \delta(3 H(AR)_\rho + 7 H(R)_\rho))
-\frac{n}{2} (H(A'|B)_\omega - \delta(3 H(A'B)_\omega + 7 H(B)_\omega))
}\\ 
+ 28 (\epsilon')^{1/4} + \sqrt{t/2} \, \beta
\bigg] 
\leq 7 \cdot 2^{-\frac{t}{8}},~~~~~~~&
\\\hspace{0.9\textwidth}
\end{align*}
where the unitary $U^{A^n}$ is chosen uniformly at 
random from an approximate $t$-design,
\begin{eqnarray*}
a 
& := & 
|A|^n \cdot
2^{
n (H(A|R)_\rho - \delta(3 H(AR)_\rho + 7 H(R)_\rho))
- n H(B)_\omega (1+7\delta) - 9
}, \\
t 
&:= &
|A|^n \kappa^2
2^{
n (H(A|R)_\rho + 32 \sqrt{\epsilon'}) + \log (\epsilon')^{-1}
- n H(B)_\omega (1-5\delta) - 6}, \\
\beta
&:= &
\sqrt{\frac{1}{a}}.
\end{eqnarray*}
\end{result} 
The proof of our main result and the analysis of its iid limit 
requires us to define two novel one-shot entropic quantities that we call
{\em smooth modified conditional R\'{e}nyi 2-entropy} 
$\mathbb{H}_2^{\epsilon,\delta}(\cdot|\cdot)$ and
{\em smooth modified max-entropy} $\mathbb{H}_{\max}^\epsilon(\cdot)$.
Their definitions and techniques used in our proofs should be of
independent interest.

Our concentration result for decoupling immediately implies that
approximate unitary $|A_1|$-designs
decouple a quantum system in the Fully Quantum Slepian Wolf (FQSW)
theorem with probability $1 - \exp(-\Theta(|A_1|))$, 
where the system $A$ is expressed as a tensor product 
$A_1 \otimes A_2$ and the superoperator simply traces out $A_2$.
\begin{result}[Partial trace ( or FQSW) concentration under design] 
\label{result:fqsw}
Consider the setting of Result~\ref{result:main}. Consider the FQSW
decoupling function
\[
f(U) = f_{FQSW}(U^{A_1 A_2}) := 
\lVert 
\Tr_{A_2} [(U^{A_1 A_2} \otimes I^R) \circ \rho^{A_1 A_2 R})] - 
\pi^{A_1} \otimes \rho^R 
\rVert_1.
\]
Let $H_2^{\epsilon}(A|R)=-\log ||{(\trho')}^{AR}||_2^2$. Suppose 
we are promised that 
$
\lVert (\trho')^R \rVert_2^2 <
0.9 |A_1| |A_2| \lVert (\trho')^{AR} \rVert_2^2,
$
$|A_1| \geq 2$, $|A_2| > |A_1|$ and
$
|A_2| 2^{H_2^\epsilon(A_1 A_2|R)_\rho - 8} - 4 > 
2 \log |A_1| + 3 \log |A_2|.
$
Let $\kappa > 0$.
The following concentration inequality holds:
\begin{align}
\label{eq:FQSW}
\prob_{U \sim \mathrm{design}}[
f(U) > 
\sqrt{\frac{|A_1|}{|A_2|}} \cdot
2^{-\frac{1}{2} H_2^\epsilon(A|R)_\rho + 1 } 
+ 14 \sqrt{\epsilon} + 2 \kappa
] \leq 
7 \cdot 2^{-a \kappa^2},
\end{align}
where the unitary $U^A$ is chosen uniformly at random from an approximate 
$t$-design, 
$a := |A_2| 2^{H_2^\epsilon(A|R)_\rho - 9}$ and
$t := 8 a \kappa^2$.
Moreover, if $|A_1| \leq \polylog(|A_2|)$ and 
$
\kappa = \sqrt{\frac{|A_1|}{|A_2|}} \cdot
	2^{-\frac{1}{2} H_2^\epsilon(A|R)_\rho + O(1)}
$
which further implies that $t = O(|A_1|)$,
then efficient
constructions for such approximate $t$-designs exist and the 
concentration Equation~\ref{eq:FQSW} can also be expressed simply as:
\begin{align}
\prob_{U \sim \mathrm{design}}[
f(U) > 
\sqrt{\frac{|A_1|}{|A_2|}} \cdot
2^{-\frac{1}{2} H_2^\epsilon(A|R)_\rho + 1 }  
+ 14 \sqrt{\epsilon} + \beta \, \sqrt{|A_1|/2} 
] \leq 
7 \cdot 2^{-\frac{|A_1|}{2}},
\end{align}
where $\beta := \sqrt{1/a}$.
\end{result}
The statement just above Result~\ref{result:fqsw} can be obtained
by setting 
$
\kappa = \sqrt{\frac{|A_1|}{|A_2|}} \cdot
	2^{-\frac{1}{2} H_2^\epsilon(A|R)_\rho + O(1)}.
$
This immediately leads to
a saving in the number of random bits to $O(|A_1| \log (|A_1| |A_2|))$
for approximate $t$-design
from $\Omega(|A_1|^2 |A_2|^2 \log (|A_1| |A_2|))$ required by Haar random
unitaries. 
If $|A_1| = \polylog |A_2|$, then efficient algorithms exist
for implementing approximate unitary $|A_1|$-designs 
\cite{brandao2012local, sen:zigzag}. Thus, for small values of 
$|A_1|$ our result shows that FQSW decoupling can indeed be
achieved by efficiently implementable unitaries with probability
exponentially close to one. 
This result can be extended to show that
{\em for small systems $S$, relative thermalisation can be
achieved by efficiently implementable unitaries with probability
exponentially close to one for a wide range of parameters, the 
first result of this kind.}

We remark that the task of replacing Haar random unitary operator 
via randomly chosen unitary operator from an approximate unitary design 
is fairly non-trivial in the case of the decoupling theorem that we 
consider in this work. This is because of the following reasons:
\begin{enumerate}

\item 
The function $f(U)$ in Result~\ref{result:main} above is not 
a polynomial in the entries of the unitary operator $U$. So, even though
Dupuis proved a concentration result for $f(U)$ 
under Haar measure in Fact~\ref{fact:dupuisconcentration}, a similar 
statement for 
$U$ chosen uniformly from a unitary design is not straightforward. 
Hence we upper bound $f(U)$ by a function $g^2(U)$, defined in 
Section~\ref{subsec:prooftechnique} below using Fact~\ref{fact:cs_tilde}, 
which is a polynomial in 
entries of $U$. The 
methodology of replacing Haar measure with unitary design can be 
applied to $g^2(U)$. However we first have to prove a 
concentration result for $g(U)$ under the Haar measure,
which calls for the evaluation of a `good' Lipschitz constant 
of $g(U)$, which is another challenging task that we carry out here. 
We then have to prove a concentration result for 
$g^2(U)$ from the
concentration result for $g(U)$, which is also a non-trivial task.\\
To have an analogous `tight' Lipschitz constant in asymptotic iid limit 
so that we can recover Dupuis result, we need to define new smooth 
max entropy! Known definition of smoothed max entropy does not suffice 
for this purpose due to weighting operators in Cauchy-Schwarz.

\item 
In order to define $g(U)$ appropriately, we have to perturb the
CPTP map $\cT^{A \to B}$ to a CP map $\hcT^{A \to B}$ (see 
Equation~\ref{eq:perturbT} for precise definition)  in the diamond 
norm in order to obtain tail bounds involving smooth conditional entropies.
This ensures that, for any input state $\rho^{AR}$, the operator
$(\cT^{A \to B} \otimes \I^R)(\rho^{AR})$ is close to the operator
$(\hcT^{A \to B} \otimes \I^R)(\rho^{AR})$ in the trace distance.
The smooth conditional entropies defined 
in earlier works like \cite{decoupling} and \cite{Smooth_decoupling}
do not quite suffice for this purpose; they only manage to show that the
positive semidefinite matrices obtained by applying
CP maps $\cT^{A \to B}$, $\hcT^{A \to B}$ to a certain `averaged state'
are close. 
Additionally, in order to obtain a good Lipschitz constant for
$g(U)$, we have to cleverly design
the weighting operator arising from the weighted Cauchy-Schwarz 
inequality required to upper bound Schatten $1$-norm of an operator with 
its Schatten $2$-norm. We also want the smooth one-shot entropic 
quantities to approach their natural Shannon entropic analogues
in the asymptotic iid regime. It is challenging
to meet all three requirements simultaneously, and for this we need
to define a {\em novel one-shot smooth conditional modified 
R\'{e}nyi $2$-entropy}. 

\end{enumerate}
Addressing the above two issues forms the new 
technical advancement towards the decoupling literature. 

\subsection{Proof technique}
\label{subsec:prooftechnique}
We now give a high level description of the proof of our main result.
For a unitary $U$ on the system $A$, we define the value taken by the
decoupling function at $U$ as follows:
\[
f(U) := 
\lVert
(\cT^{A \to B} \otimes \I^R)(
(U^A \otimes I^R) \rho^{AR} (U^{A \dagger} \otimes I^R)
)
- \omega^B \otimes \rho^R
\rVert_1.
\]
We wish to prove a tail bound for $f(U)$ where $U$ is chosen uniformly
from a unitary design. For this, it is easier to first prove a
tail bound for a related function $g(U)$:
\[
g(U) := 
\lVert
((\tcT')^{A \to B} \otimes \I^R)(
(U^A \otimes I^R) (\trho')^{AR} (U^{A \dagger} \otimes I^R)
)
- (\tomega')^B \otimes (\trho')^R
\rVert_2,
\]
where $(\tcT')^{A \to B}$, $(\trho')^{AR}$,
$(\tomega')^{A'B}$ will be defined formally, later in 
Section~\ref{sec:main}. For the discussion in this section, consider 
$\tcT'$ to be a perturbed version of $\cT$ within a distance of 
$
O(\sqrt{\epsilon})$, $(\trho')^{AR}=
2^{-\frac{1}{2}\mathbb{H}_2^{(\epsilon, \delta)}(A|R)_\rho}
$ 
and $(\tomega')^{A'B}=2^{-\frac{1}{2}\mathbb{H}_2^\epsilon(A|B)_\omega}$. 
We will have, for all probability distributions on $U^A$,
\begin{eqnarray*}
\mathbb{P}_{U^A}[
f(U) >
2^{
-\frac{1}{2} H_2^\epsilon(A|R)_\rho 
-\frac{1}{2} \mathbb{H}_2^{\epsilon,\delta}(A'|B)_\omega + 1
} + 14\sqrt{\epsilon} + 2 \theta
]\\
\leq
\mathbb{P}_{U^A}[
g(U) >
\lVert 
(\trho')^{AR}
\rVert_2 \cdot
\lVert 
(\tomega')^{A'B}
\rVert_2 
+ \theta
].&
\end{eqnarray*}

We then bound 
$
\mathbb{P}_{U^A}[
g(U) >
\lVert 
(\trho')^{AR}
\rVert_2 \cdot
\lVert 
(\tomega')^{A'B}
\rVert_2 
+ \theta
]
$
where $U^A$ is chosen according to the Haar measure. For this we
need to upper bound the Lipschitz constant of $g(U)$, which we do
in Lemma~\ref{lem:lipschitz}. Then Levy's lemma 
(Fact~\ref{fact:Levy}) gives an exponential
concentration result for $g(U)$ under the Haar measure. Using techniques 
from 
\cite{low_2009}, \cite{jl}, we obtain upper bounds on the centralised
moments of $(g(U))^2$ under the Haar measure.
Observe now that $(g(U))^2$ is a balanced degree two
polynomial (for the precise meaning see 
Definition~\ref{def:monomial}) in the 
matrix entries of $U$. We then use Low's~\cite{low_2009} derandomisation
technique in order to obtain an exponential
concentration result for $(g(U))^2$ when the unitary $U^A$ is chosen
uniformly from $t$-designs with the value of $t$ stated above. This
then leads to a similar exponential concentration result for 
$f(U)$ when  $U^A$ is chosen uniformly from a $t$-design, completing
the proof of Theorem~\ref{thm:main}. \footnote{A part of this work was presented in the online workshop Beyond IID in Information Theory-8.}

\subsection{Organisation of the paper}
Section~\ref{sec:prelim}  describes some notations, definitions and 
basic facts required for the paper.
Section~\ref{sec:main} proves the main result on one-shot decoupling with 
exponentially high concentration using unitary $t$-designs. The bounds
obtained are described using smooth versions of variants of one-shot
R\'{e}nyi $2$-entropies and max entropies.
Section~\ref{sec:iid} considers the main decoupling 
result in the
iid limit and obtains bounds in terms of the more familiar Shannon
entropic quantities.
Section~\ref{sec:fqsw} shows how to apply the main result in order to 
obtain an exponential concentration for FQSW theorem for unitary
designs. Section~\ref{sec:Applications} discuss the implications of 
FQSW concentration to 
relative thermalisation in quantum thermodynamics and to the 
Hayden-Preskill model for the black hole information paradox.
Section~\ref{sec:conclusion} concludes the paper and discusses
directions for further research.

\section{Preliminaries}  
\label{sec:prelim}
\subsection{Notation}
All vector spaces considered in the paper are finite dimensional 
inner product
spaces, aka finite dimensional Hilbert spaces, over the complex field. 
We use $|V|$ to denote the dimension of a Hilbert space $V$. 
Letters $c_1, c_2, c'_1, c'_2, \ldots$ 
denote positive universal constants. Logarithms are all taken in base
two.
We tacitly assume that the ceiling is taken of any formula that provides 
dimension or value of $t$ in unitary $t$-design.
The symbols $\mathbb{E}$, $\mathbb{P}$ denote expectation and probability 
respectively. The abbreviation "iid" is used to mean identically 
and independently distributed, which just means taking the tensor 
power of the identical copies of the underlying state. 
The notation ":=" is used to denote the definitions 
of the underlying mathematical quantities. 

The notation
$\mathcal{L}(A_1,A_2)$ denotes the Hilbert space of 
all linear operators from Hilbert space $A_1$ to Hilbert space
$A_2$ with the inner product being the Hilbert-Schmidt inner product
$\langle M, N \rangle := \Tr [M^\dag N]$.
For the special case when $A_1 = A_2$ we use
the phrase operator on $A_1$ and the symbol 
$\mathcal{L}(A_1)$.  Further, when $A_1 = A_2 = \C^m$,
$\mathcal{M}_m$ denotes vector space of all $m \times m$ matrices.
The symbol $I^{A}$ denotes the 
identity operator on vector space $A$.
The matrix $\pi^A$ denotes the so-called completely mixed 
state on system $A$, i.e., $\pi^A := \frac{I^A}{|A|}$.
We use the notation $U \circ A$ 
as a short hand to denote the conjugation of the operator $U$ on the 
operator $A$, that is, $U \circ A := U A U^\dagger$.

The symbol $\rho$ usually denotes a quantum state aka density matrix 
which is nothing but a Hermitian positive semidefinite matrix with
unit trace, and
$\mathcal{D}(\mathbb{C}^d)$ denotes the set of all $d \times d$ density 
matrices. The symbol $\mathrm{Pos}(\mathbb{C}^d)$ denotes the set of all
$d \times d$ positive semidefinite matrices, and the symbol
$\U(d)$ denotes the set of all $d \times d$ unitary matrices
with complex entries.
For a positive semidefinite matrix $\sigma$, we use $\sigma^{-1}$ to
denote the operator which is the orthogonal direct sum of
the inverse of $\sigma$ on its support and the zero operator on the
orthogonal complement of the support. This definition of 
$\sigma^{-1}$ is also known as the {\em Moore-Penrose pseudoinverse}.
The symbol $\ket{v}$ denotes a vector $v$ of unit $\ell_2$-norm, 
and $\bra{v}$ denotes the corresponding linear functional.
A rank one density matrix is called a pure quantum state.
Often, in what is a loose notation, a pure quantum state 
$\ket{v}\bra{v}$ is denoted by just the vector $\ket{v}$ or, if we
want to emphasise the density matrix formalism, by the notation $v$. For 
two Hermitian matrices $A$, $B$ of the same dimension, we use 
$A \geq B$ as a shorthand for the statement that $A - B$ is positive
semidefinite. \\
We use the notation $H(\cdot)$ to denote the usual Shannon or von 
Neumann entropy of the underlying state and the notation 
$\mathbb{H}(\cdot)$ to denote the modified or the new defined entropic 
quantities in this work.

Let $M \in \mathcal{L}(A)$.
The symbol $\Tr M$ denotes the 
trace of operator $M$. Trace is a linear map from $\mathcal{L}(A)$
to $\C$.
Let $A$, $B$ be two vector spaces.
The partial trace $\Tr_B [\cdot]$ obtained by tracing out $B$ 
is defined to be the unique linear map from 
$\mathcal{L}(A \otimes B)$ to $\mathcal{L}(A)$ satisfying 
$\Tr_B [M \otimes N] = (\Tr N) M$ for all operators 
$M \in \mathcal{L}(A)$, $N \in \mathcal{L}(B)$.

A linear map $\cT: \mathcal{M}_m \to \mathcal{M}_d $ is called a 
superoperator. A superoperator $\cT$ is said to be {\em positive} if 
it maps 
positive semidefinite matrices to positive semidefinite matrices,
and {\em completely positive} if $\cT \otimes \I$ is a 
positive superoperator for all identity superoperators $\I$.
A superoperator $\cT$ is said to be {\em trace preserving} if 
$\Tr [\cT(M)] = \Tr M$ for all $M \in \mathcal{M}_m$.
Completely positive and trace preserving (abbreviated as CPTP) 
superoperators are called 
{\em quantum operations} or {\em quantum channels}.
In this paper we only consider completely positive and trace 
non-increasing superoperators. Note that both trace and partial trace
defined in the previous paragraph are quantum channels.

The adjoint of a superoperator is defined with respect to the 
Hilbert-Schmidt inner product on matrices. In other words, if
$\cT: \mathcal{M}_m \to \mathcal{M}_d $ is a superoperator, then 
its adjoint
$\cT^\dag: \mathcal{M}_d \to \mathcal{M}_m $ is a superoperator uniquely
defined by the property that
$
\langle \cT^\dag(A), B \rangle =
\langle A, \cT(B) \rangle 
$
for all $A \in \mathcal{M}_d$, $B \in \mathcal{M}_m$.

We will be using the Stinespring representation of a superoperator, which 
we state as the following fact:
\begin{fact}
\label{fact:stinespring}
Any superoperator $\cT^{A \to B}$ can be represented as:
$$
\cT^{A \to B}(M^A) = 
\Tr_Z \{ 
V_{\cT}^{AC \to BZ} 
(M^A \otimes (\ket{0}\bra{0})^C) 
(W_{\cT}^{AC \to BZ})^{\dagger} 
\} 
$$
where $V_{\cT}$, $W_{\cT}$ are operators that map vectors from  
$A \otimes C$ to vectors in $B \otimes Z$. Systems 
$C$ and $Z$ are considered as the input and output ancillary 
systems respectively, such that $|A| |C| = |B| |Z|$. Without loss of
generality, $|C| \leq |B|$ and $|Z| \leq |A|$.
Furthermore, in the following special cases $V_{\cT}$, $W_{\cT}$
have additional properties.
\begin{enumerate}

\item 
$\cT$ is completely positive if and only if $V_{\cT} = W_{\cT}$.

\item 
$\cT$ is trace preserving if and only if $V_{\cT}^{-1} = W_{\cT}^\dag$.
Thus, $\cT$ is completely positive and trace preserving if and only if 
$V_{\cT} = W_{\cT}$ and are unitary operators.

\item 
$\cT$ is completely positive and trace non-decreasing if and only if 
$V_{\cT} = W_{\cT}$ and $\lVert V_{\cT} \rVert_\infty \leq 1$. 

\end{enumerate}  
\end{fact}

For $p \geq 1$, Schatten $p$-norm for any operator 
$M \in \mathcal{L}(A_1,A_2)$ is defined as 
$\lVert M \rVert_p \triangleq {[\Tr({(M^{\dagger}M)}^{p/2})]}^{1/p}$.
In other words, $\lVert M \rVert_p$ is nothing but the $\ell_p$-norm
of the tuple of singular values of $M$.
The Schatten $\infty$-norm is defined by taking the limit
$p \to \infty$, and turns out to be the largest singular value of $M$.
The Schatten $2$-norm, aka the Hilbert Schmidt norm, is nothing but
the $\ell_2$-norm of the tuple obtained by stretching out the entries of
the matrix into a vector. 
The Schatten $\infty$-norm is nothing but the 
operator norm 
$
\lVert M \rVert_\infty = \max_{\lVert v \rVert_2 = 1} \lVert M v \rVert_2.
$
The Schatten $1$-norm is also known as the trace norm.
We have the norm properties
$
|\Tr M| \leq \lVert M \rVert_1,
$
$
\lVert M \rVert_1 \leq \sqrt{\Tr I} \lVert M \rVert_2,
$
$
\lVert M \rVert_p \leq (\Tr I)^{1/p} \lVert M \rVert_\infty,
$
$\lVert M \otimes N \rVert_p = \lVert M \rVert_p \cdot \lVert N \rVert_p$,
$ \lVert M \rVert_p \leq \lVert M \rVert_q$ if $p \geq q$ and 
$
\lVert M N \rVert_p \leq 
\min\{
\lVert M \rVert_p \lVert N \rVert_\infty,
\lVert M \rVert_\infty \lVert N \rVert_p
\}.
$

The distance between two CP maps $\cT_1^{A \to B}$ and $\cT_2^{A \to B}$ 
can be measured in terms of the
{\em diamond norm} \cite{KitaevWatrous} defined as follows:
\[
\lVert \cT_1 - \cT_2 \rVert_\Diamond :=
\max_{\rho^{AA'}} \; 
\lVert 
(\cT_1 \otimes \I^{A'})(\rho^{AA'}) -
(\cT_2 \otimes \I^{A'})(\rho^{AA'})
\rVert_1,
\]
where $A'$ is a new Hilbert space of the same dimension as $A$ and
the maximisation is over all quantum states $\rho^{AA'}$.

\subsection{Matrix manipulation}
\begin{fact}
\label{fact:unitary_equi_purification}
For Hilbert spaces $\cH_X,\; \cH_Y$ suppose that vectors 
$\ket{\psi}, \ket{\phi} \in \cH_X \otimes \cH_Y$ satisfy
$
\Tr_Y (\ket{\psi}\bra{\psi})= \Tr_Y (\ket{\phi}\bra{\phi}).
$
Then there exists a unitary operator $U$ on $\cH_Y$ such that 
$\ket{\psi}=(I^X \otimes U^Y)\ket{\phi}$.
\end{fact}
Fix an orthonormal basis $\{\ket{a}^{A}\}_a$ of $A$ and 
$\{\ket{z}^{Z}\}_z$ 
of $Z$. Consider the tensor basis 
$
\{\ket{a}^{A} \otimes \ket{z}^{Z}\}_{a,z}
$
of the Hilbert space $A \otimes Z$. 
The isometric linear map 
$\vectorise^{A, Z}: \cL(Z, A) \to A \otimes Z$ is defined as 
the unique linear map satisfying 
$
\vectorise^{A, Z}(\ket{a}^{A} \bra{z}^{Z}) :=
\ket{a}^{A} \otimes \ket{z}^{Z}
$
\cite{watrous2004notes}.
The inverse linear map is denoted by 
$(\vectorise^{A, Z})^{-1}$. It is also an isometry.
We will be using the following property of the 
$\vectorise^{-1}$ map which we state as a fact here. A simpler version
of this fact was used in \cite{aubrun_szarek_werner_2010}.
\begin{fact}
\label{fact:vec_partial_trace}
For any two vectors $\ket{x}^{AZ}$, $\ket{y}^{AZ}$ on a 
bipartite Hilbert space $A \otimes Z$,
\[
\left(\Tr_Z(\ket{x}\bra{y}^{AZ})\right)^{A \times A} = 
\left(\vectorise^{-1}(\ket{x})\right)^{A \times Z}
\left[\left(\vectorise^{-1}(\ket{y})\right)^{A \times Z}\right]^\dagger
\]
where $\vectorise^{-1}:A \otimes Z \to A \times Z := \cL(Z,A)$.
\end{fact}
\begin{proof}
Fix orthonormal bases $\{\ket{a}^A\}_a$, $\{\ket{z}^Z\}_z$ for $A$, $Z$.
We can write 
\[
\ket{x}^{AZ} = \sum_{az} x_{az} \ket{a}^A \ket{z}^Z,
~~~
\ket{y}^{AZ} = \sum_{az} y_{az} \ket{a}^A \ket{z}^Z.
\]
This gives
\[
\vectorise^{-1}(\ket{x}) = \sum_{az} x_{az} \ket{a}^A \bra{z}^Z,
~~~
\vectorise^{-1}(\ket{y}) = \sum_{az} y_{az} \ket{a}^A \bra{z}^Z,
\]
\[
\Rightarrow 
\left(\vectorise^{-1}(\ket{x})\right)
\left(\vectorise^{-1}(\ket{y})\right)^\dagger = 
\sum_{a a'} \sum_z x_{az} y_{a'z}^*  \ket{a}\bra{a'}^A .
\]
On the other hand
\[
\Tr_Z\left(\ket{x}\bra{y}^{AZ}\right) = 
\Tr_Z \bigg(
\sum_{a a'} \sum_{z z'} x_{az} y_{a'z'}^*
\ket{a}\bra{a'}^A  \otimes \ket{z}\bra{z'}^Z
\bigg) = 
\sum_{a a'} x_{az} y_{a'z}^* \ket{a}\bra{a'}^A .
\]
This completes the proof.
\end{proof}
We now state the so called polar decomposition of any linear operator.
\begin{fact}
\label{fact:polar}
Any operator $M$ can be expressed as $M= V Q$, known as the left 
polar decomposition, where $V$ is a unitary matrix and $Q$ is a 
positive semidefinite matrix. Also, $M$ can be expressed as $M= P U$, 
where $P$ is a positive semidefinite matrix and $U$ is a unitary 
matrix. This is known as the right polar decomposition.
\end{fact} 

Next, we state four useful facts from Dupuis' thesis \cite{decoupling}.
\begin{fact}[\mbox{\cite[Lemma~I.1]{decoupling}}]
\label{fact:Dupuis_I.1}
Let $\rho$, $\rho'$ and $\sigma$ be positive semidefinite operators 
on $\cH$ such that 
$\Tr[\rho'] \leq \Tr[\sigma]$ and $\rho' \geq \rho$. Then, 
$\lVert \rho' - \sigma \rVert_1 \leq 2 \lVert \rho - \sigma \rVert_1$.
\end{fact}
\begin{fact}[\mbox{\cite[Lemma~I.2]{decoupling}}]
\label{fact:Dupuis_I.2}
Let $\rho^{AB}$ be a positive semidefinite operator, and let 
$0 \leq P^{B} \leq I^{B}$. Then,
$
\Tr_B[(P^B \otimes I^A) \rho^{AB} (P^B \otimes I^A)] \leq \rho^A.
$
\end{fact}
\begin{proof}
We give a more direct and elementary proof of this fact than what 
was given in \cite{decoupling}. The proof is a simple application of 
the definition of the partial trace and the fact that 
$P^B \leq I^B \Rightarrow {(P^B)}^2 \leq I^B$. By spectral theorem for 
positive semidefinite matrices, we express $P^B$ in its eigenbasis 
as $P^B = \Sigma_{i=1}^{|B|} p_i \ket{b_i}\bra{b_i}^B $. Since 
$P^B \leq I^B$, therefore $p_i \leq 1,\; \forall \;i$. Now we 
express $\rho^{AB}$ in block diagonal form with 
$\{ b_j \}_{j=1}^{|B|}$ as the orthonormal basis for $\cH_B$:
$$
\rho^{AB}=\Sigma_{j,j'=1}^{|B|} A^A_{j,j'} \otimes \ket{b_j}\bra{b_{j'}}^B 
~~~ \Rightarrow ~~~
\rho^A = \Sigma_{j}^{|B|} A^A_{j,j}.
$$
The block matrices $A_{j,j}$ are positive semidefinite.
Now evaluating $\Tr_B [P^B \cdot \rho^{AB}]$:
\[
\Tr_B \left[P^B \cdot \rho^{AB}\right] = 
\sum_{j,j',k,l=1}^{|B|} p_k p_l A^A_{j,j'}  
\langle{b_k}|b_j \rangle \langle{b_{j'}}|b_l \rangle 
\langle{b_l}|b_k \rangle = 
\sum_{j=1}^{|B|} p_j^2 A^A_{j,j} \overset{a}{\leq} 
\sum_{j=1}^{|B|} A^A_{j,j} =
\rho^A,
\]
where (a) holds since $p_j \leq 1$ and $A_{j,j}$ are positive 
semidefinite matrices for all $j$. This completes the proof.
\end{proof}
\begin{fact}[\mbox{\cite[Lemma~I.3]{decoupling}}]
\label{fact:Dupuis_I.3}
Let $\ket{\psi}^{AB} \in A \otimes B$, $\rho^A \in \mathrm{Pos}(A)$ 
such that $\rho^A \leq \psi^A$. Then, there exists an operator 
$P^B$ on $B$ such that $0^B \leq P^B \leq I^B$ and 
$
\Tr_B[
(P^B \otimes I^A) 
\ket{\psi}^{AB}\bra{\psi}
(P^B \otimes I^A)
] = \rho^A.
$
\end{fact}
\begin{proof}
We give a simpler and more direct proof of this fact than given in 
\cite{decoupling}. Since $\psi^A \geq \rho^A$, there exists 
a positive semidefinite matrix $\sigma^A$ such that 
$\psi^A=\rho^A + \sigma^A$. Let the vector $\ket{\rho}^{AB}$ be 
a purification of $\rho^A$ and the vector $\ket{\sigma}^{AB}$ be a 
purification of $\sigma^A$. The squares of the lengths of the purifying 
vectors are equal to the traces of the respective matrices.
Now let $Q = \C^2$ be the system representing a qubit. 
We define the unit length pure state $\ket{\theta}^{ABQ}$ as:
$$
\ket{\theta}^{ABQ} \triangleq 
\ket{\rho}^{AB} \otimes \ket{0}^Q + \ket{\sigma}^{AB} \otimes \ket{1}^Q
$$ 
It follows that $\ket{\theta}^{ABQ}$ is a purification of the state 
$\psi^A$ and so is the state $\ket{\psi}^{AB} \otimes \ket{0}^Q$. 
Thus by Fact~\ref{fact:unitary_equi_purification} there exists a unitary 
matrix $U^{BQ}$ on the composite system $BQ$ satisfying:
$$
\ket{\theta}^{ABQ}=(I^A \otimes U^{BQ})(\ket{\psi}^{AB} \ket{0}^Q)
$$ 
Now we define a POVM measurement that first appends the ancilla $Q$ 
initialized to state $\ket{0}^Q$ to the state $\ket{\psi}^{AB}$, 
followed by applying the unitary $I^A \otimes U^{BQ}$ on the state 
$\ket{\psi}^{AB}\ket{0}^Q$ and finally measuring the ancilla system 
$Q$ of the resultant state in computational basis 
$\{ \ket{0}\bra{0}^Q, \ket{1}\bra{1}^Q\}$. The measurement succeeds if
we get the outcome $0$ in the ancilla register. 
Formally, the outcome `success' is described by an operator 
$M^B := (I^B \otimes \bra{0}^Q) U^{BQ} (I^B \otimes \ket{0}^Q)$.
Clearly, $\lVert M^B \rVert_\infty \leq 1$. We thus have:
\begin{align*}
\lefteqn{
\Tr_B [(I^A \otimes M^B) 
\ket{\psi}\bra{\psi}^{AB}(I^A \otimes M^{\dagger \; B})] 
} \\
&= 
\Tr_B [(I^{AB} \otimes \bra{0}^Q)(I^A \otimes U^{BQ}) 
(\ket{\psi}^{AB}\ket{0}^Q)
(\bra{\psi}^{AB}\bra{0}^Q) 
(I^A \otimes U^{\dagger \; BQ}) (I^{AB} \otimes \ket{0}^Q)] \\
&=
\Tr_B [(I^{AB} \otimes \bra{0}^Q) \ket{\theta}\bra{\theta}^{AB}  
(I^{AB} \otimes \ket{0}^Q)]
\;=\;
\Tr_B [\ket{\rho}^{AB} \bra{\rho}] 
\;=\;
\rho^A.
\end{align*}
Now, to come up with $P^B$ as mentioned in the statement of the 
fact we express 
$M^B = U_M^B P^B$, using the polar decomposition from 
Fact~\ref{fact:polar}, with $P_B \geq 0$. Since 
$\lVert M \rVert_\infty \leq 1$, therefore $P^B \leq I^B$.
Thus we get,
\begin{align*}
\rho^A 
&= 
\Tr_B [(I^A \otimes M^B) \ket{\psi}\bra{\psi}^{AB}
(I^A \otimes M^{\dagger \;B})]\\
&=
\Tr_B [(I^A \otimes U_M^B) (I^A \otimes P^B) \ket{\psi}\bra{\psi}^{AB}  
(I^A \otimes P^B) (I^A \otimes U_M^{\dagger \; B})] \\
&\overset{a}{=} 
\Sigma_i (I^A \otimes \bra{i}^B) (I^A \otimes U_M^B) (I^A \otimes P^B) 
\ket{\psi}\bra{\psi}^{AB}  (I^A \otimes P^B) 
(I^A \otimes U_M^{\dagger \; B}) (I^A \otimes \ket{i}^B)\\
&= 
\Sigma_i  (I^A \otimes {(U_M^\dagger \ket{i})^B}^\dagger) 
(I^A \otimes P^B) \ket{\psi}\bra{\psi}^{AB}   
(I^A \otimes P^B) (I^A \otimes (U_M^\dagger \ket{i})^B)\\
&\overset{b}{=} 
\Sigma_i (I^A \otimes \bra{u_i}^B) (I^A \otimes P^B) 
\ket{\psi}\bra{\psi}^{AB}   (I^A \otimes P^B) (I^A \otimes \ket{u_i}^B)\\ 
&= 
\Tr_B [(I^A \otimes P^B) \ket{\psi}\bra{\psi}^{AB}   (I^A \otimes P^B)]
\end{align*}
where (a) follows by the basis dependent definition of $\Tr_B$ by 
fixing $\{ \ket{i} \}_{i=1}^{|B|}$ as an orthonormal basis for system $B$;
(b) holds since $U_M^B$ is a unitary matrix that maps orthonormal basis 
$\{\ket{i}\}_{\{i \in B\}}$ to orthonormal basis 
$\{\ket{u_i}\}_{\{i \in B\}}$.
We thus have a $0 \leq P^B \leq I^B$ satisfying the fact. This 
completes the proof.
\end{proof}
\begin{fact}[\mbox{\cite[Lemma~3.5]{decoupling}}]
\label{fact:Dupuis_3.5}
Let $\rho^{AB} \in \mathrm{Pos}(A \otimes B)$ and let
$\rho^B := \Tr_A \rho^{AB}$. Then,
\[
|A|^{-1} \leq
\frac{\lVert \rho^{AB} \rVert_2^2}{\lVert \rho^{B} \rVert_2^2} \leq
|A|.
\]
\end{fact}
 
In order to upper bound Schatten $1$-norm of an operator, sometimes it 
is more convenient to upper bound Schatten $2$-norm of a slightly 
modified operator. The following fact, which is nothing but an
application of the Cauchy-Schwarz inequality, allows us to do so.
\begin{fact} 
\label{fact:cs_tilde}
Let $M \in \cL(\cH)$ and $\sigma \in \cD(\cH)$, $\sigma > 0$. Then
$ 
\lVert M \rVert_1 \leq 
\lVert \sigma^{-1/4} M \sigma^{-1/4} \rVert_2.
$
\end{fact}

We will also need Winter's gentle measurement 
lemma~\cite{winter:gentle}.
\begin{fact}
\label{fact:gentle}
Let $P$ be a positive operator such that $P \leq I$. For any density 
matrix $\rho$, satisfying $\Tr [P \rho P] \geq 1 - \epsilon$, it holds 
that $\lVert \rho - P \rho P \rVert_1 \leq 2 \sqrt{\epsilon}.$
\end{fact}

We now state an important geometric fact about how a pair of subspaces
of a Hilbert space interact. This fact, first discovered by Jordan
a hundred and fifty years ago but which has since been independently
rediscovered many times, defines {\em canonical angles} between a pair of
subspaces. These angles are sometimes called as {\em chordal angles}.
\begin{fact} \mbox{\cite[Fact~6]{sen:interference}}
\label{fact:chordal}
Let $A$, $B$ be subspaces of a Hilbert space $\cH$. Then there is a
decomposition of $\cH$ as an orthogonal direct sum of the following
types of subspaces:
\begin{enumerate}

\item
One dimensional spaces orthogonal to both $A$ and $B$;

\item
One dimensional spaces contained in both $A$ and $B$;

\item
One dimensional spaces contained in $A$ and orthogonal to $B$;

\item
One dimensional spaces contained in $B$ and orthogonal to $A$;

\item
Two dimensional spaces intersecting $A$, $B$ each in one dimensional
spaces.

\end{enumerate}
Moreover, the one dimensional spaces in (2) and (3) above together with
the one dimensional intersections of the spaces in (5) with $A$ form
an orthonormal basis of $A$. A similar statement holds for $B$.
\end{fact}

We end this section by stating two properties of the 
so-called {\em swap trick} that will be useful later on.
\begin{fact}[\mbox{\cite[Lemma~3.3]{decoupling}}]
\label{fact:swap}
For two operators $M^A, N^A \in \cL(A)$, we have
$
\Tr[(MN)^A] = \Tr[(M^{A_1} \otimes N^{A_2}) F^{A_1 A_2}], 
$
where $A_1$, $A_2$ are two Hilbert spaces of the same dimension as
$A$ and $F^{A_1 A_2}$ swaps the tensor multiplicand systems 
$A_1$ and $A_2$.
\end{fact} 
\begin{fact}
\label{fact:swappartialtrace}
For an operator $M^{AR} \in \cL(A \otimes R)$, we have
\[
\bigg\lVert 
\Tr_{R_1 R_2} \left[
(I^{A_1 A_2} \otimes F^{R_1 R_2})
(M^{A_1 R_1} \otimes (M^\dag)^{A_2 R_2})
\right]
\bigg\rVert_1 \leq
|A| \bigg\lVert M^{AR} \bigg\rVert_2^2.
\]
\end{fact} 
\begin{proof}
Fix an orthonormal basis $\{\ket{r}\}_{\{r \in R\}}$ for the system 
$R$. Let
$
M^{AR} =
\sum_{r r'} M_{r r'}^A \otimes \ket{r}\bra{r'}^R ,
$
where $M_{r r'}^A$ is an operator in $A$ for every $r$, $r'$.
Then,
\begin{eqnarray*}
\lefteqn{
\Tr_{R_1 R_2} [
(I^{A_1 A_2} \otimes F^{R_1 R_2})
(M^{A_1 R_1} \otimes (M^\dag)^{A_2 R_2})
]
} \\
& = &
\sum_{r r' r'' r'''}
\left(M_{r r'}^{A_1} \otimes (M_{r''' r''}^\dag)^{A_2}\right) 
\Tr_{R_1 R_2} \left[
F^{R_1 R_2} 
\left(\ket{r}\bra{r'}^{R_1}  \otimes \ket{r''}\bra{r'''}^{R_2}  \right)
\right] \\
& = &
\sum_{r r' r'' r'''}
\left(M_{r r'}^{A_1} \otimes (M_{r''' r''}^\dag)^{A_2}\right)
\Tr_{R_1 R_2} \left[
F^{R_1 R_2} 
\ket{r}^{R_1} \ket{r''}^{R_2} \bra{r'}^{R_1} \bra{r'''}^{R_2}
\right] \\
& = &
\sum_{r r' r'' r'''}
\left(M_{r r'}^{A_1} \otimes (M_{r''' r''}^\dag)^{A_2}\right)
\Tr_{R_1 R_2} \left[
\ket{r''}^{R_1} \ket{r}^{R_2} \bra{r'}^{R_1} \bra{r'''}^{R_2}
\right] 
\;=\;
\sum_{r r'}
M_{r r'}^{A_1} \otimes (M_{r r'}^\dag)^{A_2}.
\end{eqnarray*}
So,
\begin{align*}
\bigg\lVert
\sum_{r r'}
M_{r r'}^{A_1} \otimes (M_{r r'}^\dag)^{A_2}
\bigg\rVert_1
  &\leq  
\sum_{r r'}
\bigg\lVert
M_{r r'}^{A_1} \otimes (M_{r r'}^\dag)^{A_2}
\bigg\rVert_1 
=
\sum_{r r'} \bigg\lVert M_{r r'}^{A} \bigg\rVert_1^2 \\
&\leq
\sum_{r r'} |A| \bigg\lVert M_{r r'}^{A} \bigg\rVert_2^2 \\
&= 
|A| \bigg\lVert M^{AR} \bigg\rVert_2^2.
\end{align*}
This completes the proof.
\end{proof}

\subsection{Entropic quantities}
\label{subsec:entropies}
The Shannon entropy of a random variable $X$ with probability 
distribution $(p_x)_x$ is defined as $H(X)_p := -\sum_x p_x \log p_x$.
For a quantum system $B$ in a state $\omega^B$, the Shannon entropy 
is defined analogously as $H(B)_\omega := -\Tr [\omega \log \omega]$.
For a bipartite quantum system $AB$ in a state $\omega^{AB}$, the 
conditional Shannon entropy is defined as 
$
H(A|B)_\omega := H(AB)_\omega - H(B)_\omega.
$

We recall the definition of the smooth conditional R\'{e}nyi 2-entropy 
from \cite{decoupling}.
\begin{definition}
\label{def:Renyi2}
Let $0 \leq \epsilon < 1$. 
The $\epsilon$-smooth conditional R\'{e}nyi 2-entropy for a 
bipartite positive 
semidefinite operator $\rho^{AR}$ on systems $A$ and $R$ is defined as:
$$
H_2^\epsilon(A|R)_{\rho} := 
-2 \log 
\min_{
\substack{ 
{\sigma^{AR} \in \mathrm{Pos}(AR):}\\
{\lVert \rho^{AR}-\sigma^{AR}\rVert_1 \leq \epsilon} \\
{\omega^R \in \cD(R): \omega^R > 0^R}
}
}
\{ 
\lVert 
(\omega^R \otimes \I^A)^{-1/4} \sigma^{AR} (\omega^R \otimes \I^A)^{-1/4} 
\rVert_2
\}.
$$ 
When $\epsilon = 0$, 
we simply refer to the above quantity as conditional 
R\'{e}nyi $2$-entropy and denote it by $H_2(A|R)_{\rho}$ and
define 
$
\trho^{AR} :=
(\omega^R \otimes \I^A)^{-1/4} \rho^{AR} (\omega^R \otimes \I^A)^{-1/4}.
$
\end{definition}

We also recall the definition of $\epsilon$-smooth max-entropy 
defined in \cite{QuantumAEP}.
\[
\Hmax^\epsilon(B)_\rho :=
2 \log \min_{
\substack{
{\sigma^{B} \in \mathrm{Pos}(B):}\\
{\lVert \rho^{B}-\sigma^{B}\rVert_1 \leq \epsilon}
}
} 
\{\Tr \sqrt{\sigma}\}
\]

Next, we recall the definition of $\epsilon$-smooth min-entropy
defined in \cite{decoupling}.
\begin{definition}
\label{def:Hmin}
Let $0 \leq \epsilon < 1$. The  $\epsilon$-smooth min-entropy of 
$\rho^B$ is defined as:
\[
H_{\mathrm{min}}^\epsilon(B)_\rho := -\log 
\min_{
\substack{ 
\sigma^B \in \mathrm{Pos}(B): \\
\lVert \sigma^B - \rho^B \rVert_1 \leq \epsilon
}} \,
\lVert \sigma^B \rVert_\infty.
\]
\end{definition}

We now define a new quantity that we call the 
{\em smooth modified max-entropy}.
\begin{definition}
\label{def:Hmaxprime}
The $\epsilon$-smooth modified max-entropy of system $B$ 
under a quantum state $\omega^B$ is defined as:
\[
\mathbb{H}_{\max}^\epsilon(B)_\omega := 
\log \lVert ((\omega''_\epsilon)^B)^{-1} \rVert_\infty,
\]
where $(\omega''_\epsilon)^B$ is the positive semidefinite 
matrix obtained by
zeroing out those smallest eigenvalues of $\omega^B$ that sum to
less than or equal to $\epsilon$.
The $\epsilon$-smooth modified max-entropy of a probability distribution
can be defined similarly.
\end{definition}
It is easy to see that
$
\Hmax^\epsilon(B)_\omega \leq 
\mathbb{H}_{\max}^\epsilon(B)_\omega \leq 
\log (|B| / \epsilon)
$
for any state $\omega^B$.\\
A slightly finer linear algebraic analysis would give us that 
$
\mathbb{H}_{\max}^\epsilon(B)_\omega \leq \log 
\left( \frac{|B|-\textrm{supp}((\omega''_\epsilon)^B)+1}{\epsilon} 
\right)
$.
One of the applications of this smoothed modified max-entropy is found 
in a recent work on one-shot purity distillation in \cite{One_shot_purity}.

We next define a novel entropic quantity called {\em smooth modified 
conditional R\'{e}nyi 2-entropy}.
\begin{definition}
\label{def:Renyi2prime}
Let $0 \leq \epsilon, \delta < 1$. 
The $(\epsilon,\delta)$-smooth modified conditional R\'{e}nyi 2-entropy 
for a bipartite positive semidefinite operator $\omega^{AB}$ on 
systems $A$ and $B$ is defined as:
\begin{eqnarray*}
\lefteqn{\mathbb{H}_2^{\epsilon,\delta}(A|B)_{\omega}} \\
& := &
-2 \log 
\min_{
\substack{ 
{
\eta^{AB}: 0^{AB} \leq \eta^{AB} \leq \omega^{AB},
\lVert \omega^{AB}-\eta^{AB} \rVert_1 \leq \epsilon
}\\
{
\forall \ket{v} \in \mathrm{supp}(\eta^{AB}):
\lVert 
(I^{A} \otimes \Pi^B_{\mathrm{supp}(\omega'''_{\epsilon,\delta})}) \ket{v} 
\rVert_2^2 
\geq 1- \epsilon
}
}
}
\lVert 
(I^A \otimes (\omega'''_{\epsilon,\delta})^{B})^{-1/4} 
\eta^{AB} 
(I^A \otimes (\omega'''_{\epsilon,\delta})^{B})^{-1/4} 
\rVert_2,
\end{eqnarray*}
where $(\omega'''_{\epsilon,\delta})^B$ is the positive 
semidefinite operator obtained
by zeroing out those eigenvalues of $\omega^B$ that are smaller than
$2^{-(1+\delta)\mathbb{H}_{\max}^\epsilon(B)_\omega}$.
\end{definition}
Observe that $(\omega'''_{\epsilon, \delta})^B \geq (\omega''_\epsilon)^B$,
where $(\omega''_\epsilon)^B$ is defined in 
Definition~\ref{def:Hmaxprime} above.\\
We remark that the above modified definitions of the smoothed max and 
the conditional 2-entropy is inspired (not directly inferred though) from 
the so-called information spectrum methods used in the analysis of 
asymptotic independent but not identical classical and quantum 
information theory. We refer the reader to 
\cite{infor_spectrum_Hayashi, information_spectrum_Han, 
Classical_info_spectrum} for a detailed exposition.

It is easy to see, via Fact~\ref{fact:gentle}, that for any 
state $\omega^{AB}$,
$
H_2^{4\sqrt{\epsilon}}(A|B)_{\omega} \geq 
\mathbb{H}_2^{\epsilon,\delta}(A|B)_{\omega} 
$
for any $\delta > 0$.
Also the following fact holds.
\begin{fact}
\label{fact:Renyi2upperbound}
Let $\epsilon > 0$. Then,
$
H_2^\epsilon(A|B)_\omega  \leq 
H(A|B)_\omega + 8 \epsilon \log |A| + 
2 + 2 \log \epsilon^{-1}.
$
\end{fact}
\begin{proof}
The proof follows by combining Equation~8 of \cite{tomamichel:Renyi}
with Theorem~7, Lemma~2 and Equation~33 of \cite{QuantumAEP} and then
applying the Alicki-Fannes inequality \cite{AlickiFannes}.
\end{proof}

\subsection{Types and typicality}
The smooth entropic quantities defined in the previous section are
suitably bounded by the standard Shannon entropic quantities in the
iid limit, as will be shown in Section~\ref{sec:iid}. 
In order to lay the groundwork for the proofs in Section~\ref{sec:iid},
we recall the definitions of types, typical sequences and subspaces.

\begin{definition}
Let $X$ be a finite set. Fix a probability distribution $p$ on $X$. 
Let $n$ be a positive integer. Let 
$X^n$ denote the random variable
corresponding to $n$ independent copies of $X$. 
The notation $x^n$ shall represent a sequence of length $n$ over the 
alphabet $X$ . Let $N(a|x^n)$ denote the number of occurrences of the 
symbol $a \in X$ in the sequence $x^n$. The multiset 
$
\{N(a|x^n)\}_{a \in X}
$
is called the {\em type} of $x^n$. The set of all possible types is
nothing but the set of all possible $|X|$-tuples of non-negative
integers summing up to $n$.
\end{definition}
\begin{definition}
Let $0 < \delta < 1$. The set 
of strongly $\delta$-typical types of
length $n$ over the alphabet $X$ pertaining to the distribution $p$ 
is defined as \cite{book:elgamalkim}
$
\bigg\{
(m_a)_{a \in X}: 
\forall a \in X, \;
m_a \in n p(a) (1 \pm \delta)
\bigg\}.
$
A sequence $x^n$ is said to be strongly $\delta$-typical if its type
is strongly $\delta$-typical. The set of strongly $\delta$-typical 
sequences is denoted by $T^{X^n}_{p,\delta}$.
\end{definition}
Let $p^n$ denote the $n$-fold tensor power of probability distribution
$p$. The strongly typical sequences satisfy 
the following property which is 
called as Asymptotic Equipartition Property (AEP) in classical
Shannon theory.
\begin{fact}[\cite{sen:interference}]
\label{fact:classical_AEP}
The number of types is ${n+|X|-1 \choose |X|-1}$. The set of all possible
sequences $X^n$ is partitioned into a disjoint union, over all possible 
types, of sequences having a given type.
Let $0 < \epsilon, \delta < 1/2$. 
Define $p_{\mathrm{min}} := 2^{-\mathbb{H}_{\max}^{\epsilon/2}(X)_p}$.
Let 
$
n \geq 4 p^{-1}_{\mathrm{min}} \delta^{-2} 
\log(|X|/\epsilon).
$
Then,
\begin{eqnarray*}
\sum_{x^n \in T^{X^n}_{p,\delta}} p^{n}(x^n) 
& \geq &
1-\epsilon, \\
\forall x^n \in T^{X^n}_{p,\delta}: 
2^{-nH(X)(1+\delta)} 
& \leq & 
p^{n}(x^n) 
\;\leq\;
2^{-nH(X)(1-\delta)},  \\
2^{nH(X)(1-\delta)} (1-\epsilon)
& \leq & 
|T^{X^n}_{p,\delta}| 
\;\leq\;
2^{nH(X)(1+\delta)}.
\end{eqnarray*}
\end{fact}

In the quantum setting,
we extend the notion of types and typical sequences with respect to a 
particular 
distribution to the notion of type subspaces and typical subspaces with 
respect to the 
$n$-fold tensor product of a quantum state.
\begin{definition}
Let $\rho$ be a density matrix over a Hilbert space $B$. Consider a 
canonical eigenbasis 
$\cB = \{ \ket{\chi_1},\ldots, \ket{\chi_{|B|}} \}$
of $\rho$. Consider the diagonalisation 
$\rho = \sum_{\chi \in \cB} q(\chi) \ket{\chi}\bra{\chi}$, where
the set of eigenvalues
$\{q(\chi)\}_\chi$ can be treated as a probability distribution 
over $\cB$. 
Given a type $(m(\chi))_{\chi \in \cB}$, which is nothing but a $|B|$-tuple
of non-negative integers summing to $n$, we define the corresponding
type subspace to be the span of all $n$-fold tensor products of vectors 
from $\cB$ having the given type.
\end{definition}
\begin{definition}
Let $0 < \epsilon, \delta < 1/2$.
The strongly 
$\delta$-typical subspace of $B^{\otimes n}$ corresponding to the $n$-fold 
tensor power operator $\rho^{\otimes n}$, $T^{B^n}_{\rho, \delta}$,
is defined as the orthogonal direct sum of all type subspaces with
strongly $\delta$-typical types with respect to the probability
distribution $q$ on $\cB$.
\end{definition}
Let $\Pi^{B^n}_{\rho, \delta}$ denote the orthogonal projection onto 
$T^{B^n}_{\rho, \delta}$. The typical projector satisfies the following so 
called quantum AEP analogous to that of Fact~\ref{fact:classical_AEP}:
\begin{fact}[\cite{sen:interference}]
\label{fact:quantum_AEP}
The number of types is ${n+|B|-1 \choose |B|-1}$. The Hilbert space
$B^{\otimes n}$ can be decomposed into an orthogonal direct sum,
over all possible types, of type subspaces.
Let $0 < \epsilon, \delta < 1/2$. Let $\rho$ be a quantum state. 
Define
$
q_{\mathrm{min}} \cong q_{min}(\rho) :=
2^{-\mathbb{H}_{\max}^{\epsilon/2}(B)_\rho}.
$
Suppose that 
$n \geq 4 q^{-1}_{\mathrm{min}} \delta^{-2} \log (|B| / \epsilon)$. 
Then,
\begin{eqnarray*}
\Tr [\rho^{\otimes n} \Pi^{B^n}_{\rho,\delta}] 
& \geq &
1-\epsilon, \\
2^{-nH(X)(1+\delta)} \Pi^{B^n}_{\rho,\delta}
& \leq & 
\Pi^{B^n}_{\rho,\delta} \rho^{\otimes n}
\;=\;
\rho^{\otimes n} \Pi^{B^n}_{\rho,\delta}
\;=\;
\Pi^{B^n}_{\rho,\delta} \rho^{\otimes n} \Pi^{B^n}_{\rho,\delta}
\;\leq\;
2^{-nH(X)(1-\delta)}  \Pi^{B^n}_{\rho,\delta}, \\
2^{nH(X)(1-\delta)} (1-\epsilon)
& \leq & 
\Tr \Pi^{B^n}_{\rho,\delta}
\;\leq\;
2^{nH(X)(1+\delta)}.
\end{eqnarray*}
\end{fact}

\subsection{Concentration of measure}
We state the main tool for concentration of measure of
Lipschitz functions defined on the sphere or on the unitary group in
high dimensions.
\begin{definition}
A complex valued function $f$ defined on a subset 
$S \subseteq \mathbb{C}^n$ is said to be 
$L$-Lipschitz, or with Lipschitz constant $L$, if 
$\forall x, y \in S$ it satisfies the following inequality:
\begin{equation}
\lvert f(x)-f(y) \rvert \leq L \lVert x-y \rVert_2.
\end{equation}
\end{definition}
\begin{remark}
In order to find Lipschitz constant of a complex valued Lipschitz 
function $f$ with the domain as the the set of unitary operators, we find 
an upper bound on the ratio $\frac{|f(U)-f(V)|}{\lVert U-V \rVert_2}$. 
We recall that the norm $\lVert U-V \rVert_2$ is the Schatten-2 norm of 
the operator $U-V$ defined as $\sqrt{\Tr[{(U-V)}^\dagger (U-V)]}$. 
\end{remark}
\begin{fact} 
\label{fact:Levy}
{\bf (Levy's Lemma~\cite{Levy, AGZ})} Let $f$ be an $L$-Lipschitz 
function on $\U(n)$ where the metric on $\U(n)$ is
induced by the embedding of $\U(n)$ into $\C^{n^2}$. In 
other words, the metric on $\U(n)$ is
taken to be the Schatten 2-norm. Consider the Haar probability measure on
$\U(n)$. Let the mean of $f$ be $\mu$. Then:
\[
\prob_{U \sim \Haar}\left(\lvert f(U) - \mu \rvert \geq \lambda\right) 
\leq 2 \exp\left(-\frac{n \lambda^2}{4L^2}\right).
\]
\end{fact}

\begin{remark}
L\'evy's lemma is widely used tool in high dimensional probability and 
asymptotic geometric functional analysis. While the original form can 
be found in \cite{Levy}, several simplifications have been developed in
the literature and one such classic reference is \cite{Ledoux}. Here 
we use a much simplified form derived for sub-Gaussian distribution 
described in \cite[Theorem~5.1.4]{vershynin}.  
\end{remark}

The following fact can be used to compute upper bounds on
the centralised moments of Lipschitz functions. 
The proof follows Bellare and Rompel's
seminal work on concentration for sums of $t$-wise independent random
variables \cite[Lemma~A.1]{BellareRompel} and its quantum adaptation
by Low\cite[Lemma~3.3]{low_2009}. However, inspired by a technique from
\cite{jl}, we extend the earlier results
in an important and essential way by computing upper bounds on the
centralised moments of squares of Lipschitz functions also, which will
be required in Section~\ref{subsec:gmoment}.
\begin{fact}
\label{fact:moment}
Let $X$ be a non-negative random variable. Suppose there is a number
$\mu$ satisfying, for any $\kappa > 0$, the tail bound
$
\mathbb{P}[|X - \mu| > \kappa] \leq C \exp(-a \kappa^2)
$
for some positive constants $C$, $a$. 
Let $m$ be a positive integer. 
Then
\[
\E[(X-\mu)^{2m}] \leq C \left(\frac{m}{a}\right)^m, ~~~
\E[(X^2-\mu^2)^{2m}] \leq 
\left\{
\begin{array}{l l}
3 C \left(\frac{9 m \mu^2 }{a}\right)^{m} 
& 1 \leq m \leq \frac{9}{64} a \mu^2 \\
3 C \left(\frac{64 m^2}{a^2}\right)^{m} 
& \mbox{otherwise}
\end{array}
\right..
\]

\end{fact}
\begin{proof}
Let $\Omega$ with a probability measure $d\omega$ be the sample space 
serving as the domain of the measurable function $X$.
Then,
\begin{eqnarray*}
\lefteqn{
\E[(X-\mu)^{2m}]
} \\
& = &
\int_\Omega (X(\omega) - \mu)^{2m} \, d\omega
\;=\;
\int_\Omega \int_{0 \leq x \leq (X(\omega) - \mu)^{2m}}  \, dx \, d\omega 
\;=\;
\int_0^\infty \int_{(X(\omega) - \mu)^{2m} \geq x}  \, d\omega \, dx \\
& = &
\int_0^\infty \prob[(X-\mu)^{2m} \geq x]\,dx
\;=\;
\int_0^\infty \prob[|X-\mu| \geq x^{\frac{1}{2m}}]\,dx 
\;\leq\;
C \int_0^\infty \exp(-a x^{1/m}) \,dx \\
& = &
C m a^{-m} \int_0^\infty e^{-y} y^{m-1}  \,dy 
\;=\;
C a^{-m} \Gamma(m+1)
\;\leq\;
C \left(\frac{m}{a}\right)^m.
\end{eqnarray*}

Let $A := \{\omega: 0 \leq X(\omega) < 2 \mu\}$ and $\bar{A}$ denote its 
complement in $\Omega$.We thus have 
\begin{equation}
    \label{eq:P(bar_A)}
\prob(\bar{A}) \leq C \exp{(-a \mu^2)}
\end{equation}
and,
\[
\E\left[(X^2-\mu^2)^{2m}\right] = 
\int_\Omega (X(\omega)^2 - \mu^2)^{2m} \, d\omega =
\int_A (X(\omega)^2 - \mu^2)^{2m} \, d\omega +
\int_{\bar{A}} (X(\omega)^2 - \mu^2)^{2m} \, d\omega,
\]
We now separately upper bound the two integrals on the right hand 
side above.  For the first integral,
\begin{eqnarray*}
\int_A (X(\omega)^2 - \mu^2)^{2m} \, d\omega 
& = &
\int_A (X(\omega) - \mu)^{2m} (X(\omega) + \mu)^{2m} \, d\omega 
\;\leq\;
(3\mu)^{2m}
\int_\Omega (X(\omega) - \mu)^{2m} \, d\omega \\ 
& \leq &
C (3 \mu)^{2m} \left(\frac{m}{a}\right)^m
\; =    \;
C \left(\frac{9 m \mu^2}{a}\right)^m. 
\end{eqnarray*}
Now if $m > \frac{9}{64} a \mu^2$, the last term above can be further
upper bounded by $C \left(\frac{64 m^2}{a^2}\right)^m.$

For the second integral,
\begin{eqnarray*}
\lefteqn{
\int_{\bar{A}} (X(\omega)^2 - \mu^2)^{2m} \, d\omega 
} \\
& = &
2m 
\int_{\bar{A}} \int_{0 \leq x \leq (X(\omega)^2 - \mu^2)}
x^{2m-1} \, dx \, d\omega \\
& = &
2m
\int_0^{3 \mu^2} \int_{\bar{A}}
x^{2m-1} \, d\omega \, dx+\,
2m 
\int_{3 \mu^2}^\infty \int_{X(\omega)^2 \geq \mu^2 +x}
x^{2m-1} \, d\omega \, dx \\
& = &
{(3 \mu^2)}^{2m} \prob(\bar{A}) +\,
2m 
\int_{3 \mu^2}^\infty \int_{X(\omega)^2 \geq \mu^2 +x}
x^{2m-1} \, d\omega \, dx \\
& \overset{a}{\leq} &
C {(3 \mu^2)}^{2m} \exp{(-a \mu^2)} +\,
2m 
\int_{3 \mu^2}^\infty \int_{X(\omega) \geq \mu + (\sqrt{x}/2)}
x^{2m-1} \, d\omega \, dx \\
& \overset{b}{\leq} &
C \left(\frac{9 m \mu^2}{a}\right)^m +\,
2m 
\int_{3 \mu^2}^\infty \prob[X - \mu \geq \sqrt{x}/2]
x^{2m-1} \, dx \\
& \leq &
C \left(\frac{9 m \mu^2}{a}\right)^m +\,
2C m 
\int_{3 \mu^2}^\infty \exp(-a x / 4) x^{2m-1} \, dx \\
& \leq &
C \left(\frac{9 m \mu^2}{a}\right)^m +\,
2C m \left(\frac{4}{a}\right)^{2m}
\int_0^\infty e^{-y} y^{2m-1} \, dy \\
&  =   &
C \left(\frac{9 m \mu^2}{a}\right)^m +\,
C \left(\frac{4}{a}\right)^{2m} \Gamma(2m+1) \\
& \leq &
C \left(\frac{9 m \mu^2}{a}\right)^m +\,
C \left(\frac{8m}{a}\right)^{2m}.
\end{eqnarray*}
where (a) follows from Equation~\ref{eq:P(bar_A)} and the observation 
that for any $x \geq 3 \mu^2$, $\sqrt{\mu^2 + x}>\mu + \sqrt{x}/2$; and 
(b) follows from the observation that 
\begin{equation} \label{eq:Exp_barA}
\exp{(-a \mu^2)} 
\leq 
\left(\frac{m}{a \mu^2}\right)^m
\end{equation},
which can be seen from the following cases:
\begin{itemize}

\item[]{Case 1: }For all $m \geq a\mu^2$, it holds that 
$\left(\frac{m}{a \mu^2}\right)^m \geq 1$ and $\exp(-a\mu^2) \leq 1$. 
Thus Equation~\ref{eq:Exp_barA} holds trivially.

\item[]{Case 2: }For $m < a \mu^2$, choose a $\delta <1 $ and 
let $m= \delta a \mu^2$.
Then 
$
\left(\frac{m}{a \mu^2}\right)^m = 
\exp{(-\delta \ln{(\frac{1}{\delta})}a \mu^2)} > 
\exp(-a \mu^2)
$ 
holds since $\delta < 1$. 
As $\delta$ was chosen arbitrarily, Equation~\ref{eq:Exp_barA} holds.
\end{itemize}
Thus, Equation~\ref{eq:Exp_barA} implies:
$$
C {(3 \mu^2)}^{2m} \exp{(-a \mu^2)} 
 \leq 
C \left(\frac{9 m \mu^2}{a}\right)^m 
$$
and (b) holds.

Now if $m \leq \frac{9}{64} a \mu^2$, the last term above can be further
upper bounded by $C \left(\frac{9 m \mu^2}{a}\right)^m.$
Thus
$
\E[(X^2-\mu^2)^{2m}] \leq 
3 C \left(\frac{9 m \mu^2 }{a}\right)^{m} 
$
if $m \leq \frac{9}{64} a \mu^2$, and\\
$
\E[(X^2-\mu^2)^{2m}] \leq 
3 C \left(\frac{64 m^2}{a^2}\right)^{m} 
$
if $m > \frac{9}{64} a \mu^2$. This
completes the proof of the fact.
\end{proof}

\subsection{Unitary t-designs}
\begin{definition}
\label{def:monomial}
 A monomial in elements of a matrix $U$ is of degree $(r,s)$ if it 
contains $r$ conjugated elements and $s$ unconjugated elements of $U$. 
We call 
it balanced if $r = s$ and will simply say a balanced monomial has 
degree $t$ if it is of degree $(t,t)$. A polynomial is of degree $t$ 
if it is a sum of balanced monomials of degree at most $t$.
\end{definition}
\begin{definition}
\label{def:design}
A probability distribution $\nu$ on a finite set of $d \times d$
unitary matrices is said to be a 
an $\epsilon$-approximate unitary $t$-design if for all 
balanced monomials $M$ of degree at most $t$, the following holds 
\cite{low_2009}:
\begin{align*}
\bigg\lvert \mathbb{E}_{\nu}\left[M(U)\right] - 
\mathbb{E}_{\Haar}\left[M(U)\right] \bigg\rvert 
\leq \frac{\epsilon}{d^t}
\end{align*}
If $\epsilon = 0$, we say that $\nu$ is an exact unitary $t$-design,
or just unitary $t$-design.
\end{definition}

For technical ease, we use quantum tensor product expanders (qTPEs) in
place of unitary designs in our actual proofs. 
The formal definition of a qTPE follows.
\begin{definition}
\label{def:qTPE}
A quantum $t$-tensor product expander ($t$-qTPE) in $\cH$, $|\cH| = d$, 
of degree $s$ can be defined as a quantum operation
$
\cG: L(\cH^{\otimes t}) \rightarrow L(\cH^{\otimes t})
$
that can be expressed as
$
\cG(M) = 
\frac{1}{s} \sum_{i=1}^s (U_i)^{\otimes t} M (U_i^{-1})^{\otimes t},
$
for any matrix $M \in L(\cH^{\otimes t})$, where
$\{U_i\}_{i=1}^s$ are $d \times d$ unitary matrices. The 
qTPE is said to have second singular value $\lambda$ if
$
\lVert \cG - \cI \rVert_\infty \leq \lambda,
$
where $\cI$ is the `ideal' quantum operation defined by its action on
a matrix $M$ by
$
\cI(M) := 
\int_{U \in \U(D)} U^{\otimes t} M (U^\dag)^{\otimes t} \,
d\,\Haar(U).
$
In other words, if $M \in L(\cH^{\otimes t})$, then
$
\lVert \cG(M) - \cI(M) \rVert_2 \leq \lambda \lVert M \rVert_2.
$
We use  the notation $(d, s, \lambda, t)$-qTPE to denote such a
quantum tensor product expander.
\end{definition}
A $(d, s, \lambda, t)$-qTPE can be sequentially iterated 
$O(\frac{t \log d + \log \epsilon^{-1}}{\log \lambda^{-1}})$
times to obtain an $\epsilon$-approximate unitary 
$t$-design \cite[Lemma~2.7]{low_2009}.
For $t = \polylog(d)$, efficient construction of $t$-qTPEs 
are known~\cite{brandao2012local,sen:zigzag}.

\section{Proof of the main theorem}
\label{sec:main}
We now state our main theorem in full detail.
\begin{theorem}
\label{thm:main}
Consider a quantum state 
$\rho^{AR}$ shared between a system $A$ and a reference $R$. Let 
$\cT^{A \to B}$ be a completely positive trace preserving superoperator
with input system $A$ and output
system $B$.
Let $U$ be a unitary on the system $A$. Define the function
\[
f(U) := 
\bigg\lVert
\left[\cT^{A \to B} \otimes \I^R\right]\left(
(U^A \otimes I^R) \rho^{AR} (U^{A \dagger} \otimes I^R)
\right)
- \omega^B \otimes \rho^R
\bigg\rVert_1,
\]
where $\I^R$ is the identity superoperator on $R$ and 
$I^R$ is the identity operator on $R$. 
Let $A'$ be a new system having the same dimension as $A$.
Define the Choi-Jamio{\l}kowski state
$
\omega^{A'B} := 
(\cT^{A \to B} \otimes \I^{A'})(
\ket{\Phi}\bra{\Phi}^{AA'}
),
$
of $\cT^{A \to B}$ where
$\ket{\Phi}^{AA'} := |A|^{-1/2} \sum_a \ket{a}^A \otimes \ket{a}^{A'}$ is
the standard EPR state on system $A A'$. 
Let $0 < \epsilon, \delta < 1/3$. Let $\kappa > 0$. Then,
\[
\prob_{U^A \sim \TPE}\left[
f(U) >
2^{
-\frac{1}{2} H_2^\epsilon(A|R)_\rho 
-\frac{1}{2} \mathbb{H}_2^{\epsilon,\delta}(A'|B)_\omega + 1
} + 14 \sqrt{\epsilon} + 2 \kappa
\right] \leq 
7 \cdot 2^{-a \kappa^2}.
\]
where $U^A$ is a unitary matrix chosen uniformly at random from a 
$(|A|, s, \lambda, t)$-quantum tensor product expander, 
$
a := 
|A| \cdot
2^{
-(1+\delta)\mathbb{H}_{\max}^\epsilon(B)_\omega +
H_2^{\epsilon}(A|R)_\rho - 9
},
$
$t := 8 a \kappa^2,$
$
\lambda :=
(|A|^{-8} |B|^{-6} \cdot \mu^2)^t.
$
As mentioned in Remark~\ref{rem:parameterization}, by choosing 
parameters 
$
\alpha:=2^{
-\frac{1}{2} H_2^\epsilon(A|R)_\rho 
-\frac{1}{2} \mathbb{H}_2^{\epsilon,\delta}(A'|B)_\omega+1
}
$ 
and 
$
\beta:=\sqrt{\frac{1}{a}}
$ 
the above concentration inequality can also be expressed as:
\begin{equation*}
\label{eq:main_decoupling_parameterized}
\prob_{U^A \sim \mathrm{design}}[
f(U) >
\alpha + 14 \sqrt{\epsilon} + \sqrt{t/2}\beta
] \leq 
7 \cdot 2^{-\frac{t}{8}}.
\end{equation*}
The quantity $\mu$ is defined as 
$\mu := \E_{\Haar}[g(U)]$ for a related function
$g(U)$ defined in Equation~\ref{eq:gU} below.
We require that$\mu < 1$. 
%It satisfies
% \begin{eqnarray*}
% \mu^2 \leq 
% \E_{\Haar}[(g(U))^2]
% & = &
% \delta_1 \lVert (\trho')^R \rVert_2^2 + 
% \delta_2 \lVert (\trho')^{AR} \rVert_2^2 -
% \lVert (\tomega')^B \rVert_2^2 \lVert (\trho')^R \rVert_2^2 \\
% & < &
% \lVert (\tilde{\omega}')^{A'B}\rVert_2^2 \cdot
% \lVert (\tilde{\rho}')^{AR}\rVert_2^2 
% \;=\;
% 2^{-\mathbb{H}_2^{\epsilon,\delta}(A'|B)_\omega} \cdot
% 2^{-H_2^{\epsilon}(A|R)_\rho},
% \end{eqnarray*}
% where $\delta_1$, $\delta_2$ are defined in 
% Proposition~\ref{prop:decoupling_general_Haar} below, and
% the positive semidefinite matrices
% $(\tilde{\omega}')^{A^\prime B}$, 
% $(\tilde{\rho}')^{AR}$ are defined in Equations~(\ref{eq:tilde_rho}, 
% \ref{eq:tilde_pomega})
% below. Moreover if
% \[
% a \cdot \E_{U \sim \Haar}[(g(U))^2] +
% \log \E_{U \sim \Haar}[(g(U))^2] >
% \log |A| +
% \mathbb{H}_{\max}^\epsilon(B)_\omega -
% H_2^\epsilon(A|R)_\rho,
% \]
% then $\E_{U \sim \Haar}[(g(U))^2] \leq 8\mu^2$.

\end{theorem}  
\begin{remark}
    Such qTPEs can be obtained by sequentially iterating 
$O\left(t (\log |A| + \log |B| + \log \mu^{-1})\right)$ times an 
$(|A|, s, O(1), t)$-qTPE, which is
polynomial in $t$, $\log \mu^{-1}$ and the number of input and output 
qubits of the CPTP map $\cT$ \cite{brandao2012local,sen:zigzag}.
\end{remark}
The proof is broken
into three subsections. In the first subsection, we show that, for
any probability distribution on $U^A$, instead
of proving a tail bound for the given random variable $f(U)$,
it suffices to prove a tail bound for a related random variable $g(U)$,
where $f(U)$, $g(U)$ were informally defined in 
Section~\ref{subsec:prooftechnique}.
In the second subsection, we first obtain an upper bound on the Lipschitz
constant of $g(U)$ which by Levy's lemma (Fact~\ref{fact:Levy}) leads to 
a tail bound for $g(U)$ where
$U$ is chosen from the Haar measure. We then
obtain upper bounds on the centralised moments of $(g(U))^2$ 
under the Haar measure.
Now $(g(U))^2$ is a balanced 
degree two polynomial in the matrix entries of $U$. In the final
subsection, we apply Low's method to finally obtain a tail bound
for $(g(U))^2$ for a uniformly random $U$ chosen from a unitary design.
This finishes the proof of Theorem~\ref{thm:main}.

\subsection{From $f(U)$ to $g(U)$}
\label{sec:f(U)_to_g(U)}
We summarize various operators superoperators and their 
Choi-Jamio\l{}kowski states that we will be working with, in order to 
establish the proof of Theorem~\ref{thm:main}, respectively in the 
following Table~\ref{table:superoperators} and \ref{table:Choi_states}
\begin{center}
\begin{table}[ht]
\caption{Various superoperators used in going from $f(U)$ to $g(U)$.} 
\begin{tabularx}{\linewidth}{ l X }
\toprule
$\cT^{A \to B}: $
& unperturbed decoupling CPTP superoperator with  Stinespring 
unitary: 
$V_{\cT}^{AC \to BZ}$ \\
\hline \\
$\hat{\cT}^{A \to B}(\cdot):=$ 
& $\Tr_Z \left[ (P^Z \otimes I^B)V_{\cT} \circ (\cdot)\right]$: 
CP Trace non-increasing map; $P$ comes from Lemma~\ref{fact:Dupuis_I.3}  \\
\hline \\
$(\cT')^{A \to B}(\cdot):=$
& $\Pi^B_{\omega'''_{\epsilon, \delta}} \circ \hat{\cT}$: CP Trace 
non-increasing map \\
\bottomrule
\end{tabularx}
\label{table:superoperators}
\end{table}
%\end{center} 
%\begin{center}
\begin{table}[ht]
\caption{Choi-Jamio\l{}kowski states of superoperators in 
Table~\ref{table:superoperators}.}
\begin{tabularx}{\linewidth}{l X}
\toprule
$ \omega^{A'B}:=$ 
&$(\cT \otimes \cI^{A'})(\Phi^{AA'})$;  
$\Tr(\omega^{A'B})=1$; Desired smoothing: 
$(\omega'''_{\epsilon, \delta})^B:=$ 
positive semidefinite operator obtained by zeroing out those eigenvalues 
of $\omega^B$ that are smaller than 
$2^{-(1+\delta)\mathbb{H}_{\max}^\epsilon(B)_\omega}$\\
\hline\\
$\eta^{A'B}:=$
& $(\hat{\cT} \otimes \cI^{A'})(\Phi^{AA'})$; $\Tr(\eta) \leq 1$\\
\hline\\
$(\omega')^{A'B}:=$
& $(\cT \otimes \cI^{A'})(\Phi^{AA'})$;
 $=\Pi^B_{\omega'''_{\epsilon, \delta}} \circ \eta^{A'B}$ and 
$\Tr[\omega'] \leq 1$\\
\bottomrule
\end{tabularx}
\label{table:Choi_states}
\end{table}
\end{center} 

Recall that for a unitary $U$ on the system $A$, we define the value 
taken by the decoupling function at $U$ as follows:
\[
f(U) := 
\bigg\lVert
\left[\cT^{A \to B} \otimes \I^R\right]\left(
(U \otimes I^R) \rho^{AR} (U^{ \dagger} \otimes I^R)
\right)
- \omega^B \otimes \rho^R
\bigg\rVert_1.
\]
Let $\eta^{AB} \leq \omega^{AB}$ be the positive semidefinite operator
achieving the optimum in Definition~\ref{def:Renyi2prime}.
Let $V_\cT^{AC \rightarrow BZ}$ be a unitary Stinespring dilation of
the CPTP map $\cT^{A \rightarrow B}$ provided by 
Fact~\ref{fact:stinespring}. Thus 
\[
(\cT^{A \to B} \otimes \I^{A'})(M^{A A'}) = 
\Tr_Z \left[
(V_{\cT}^{AC \to BZ} \otimes I^{A'})
(M^{AA'} \otimes \ket{0}\bra{0}^C) 
(V_{\cT}^{AC \to BZ} \otimes I^{A'})^{\dagger} 
\right]
\]
for any $M^{A A'} \in \cL(A \otimes A')$, where $A'$ is a new Hilbert
space of the same dimension as $A$.
Recall that
$
\omega^{A'B} := 
(\cT^{A \to B} \otimes \I^{A'})(\Phi^{A A'}),
$
where $\Phi^{A A'}$ is the standard EPR pure state on $A A'$.
By Fact~\ref{fact:Dupuis_I.3}, there exists a POVM element on $Z$,
$0^Z \leq P^Z \leq I^Z$, such that
\begin{eqnarray}
\label{eq:perturbT}
\lefteqn{(\hcT^{A \to B} \otimes \I^{A'})(\Phi^{A A'})} \\ \nonumber
& := &
\Tr_Z \left[
(P^Z \otimes I^{B A'})
(V_{\cT}^{AC \to BZ} \otimes I^{A'})
(\Phi^{AA'} \otimes \ket{0}\bra{0}^C) 
(V_{\cT}^{AC \to BZ} \otimes I^{A'})^{\dagger} 
(P^Z \otimes I^{B A'})
\right] \\ 
& = &
\eta^{B A'}. \nonumber
\end{eqnarray}
The superoperator $\hcT^{A \to B}$ is completely positive and trace
non-increasing.
Define the function
\[
\hat{f}(U) := 
\bigg\lVert
\left[\hcT^{A \to B} \otimes \I^R\right]\left(
(U^A \otimes I^R) \rho^{AR} (U^{A \dagger} \otimes I^R)
\right)
- \omega^B \otimes \rho^R
\bigg\rVert_1.
\]
By Fact~\ref{fact:Dupuis_I.2}, 
\begin{equation}
(\hcT^{A \to B} \otimes \I^R)(
(U \otimes I^R) \rho^{AR} (U^{ \dagger} \otimes I^R)
) \leq
(\cT^{A \to B} \otimes \I^R)(
(U \otimes I^R) \rho^{AR} (U^{ \dagger} \otimes I^R)
).
\end{equation}
Hence by Fact~\ref{fact:Dupuis_I.1}, $f(U) \leq 2 \hat{f}(U)$.

Let $\Pi^B_{\omega'''_{\epsilon,\delta}}$ be the projector onto the
support of $(\omega'''_{\epsilon,\delta})^B$, where we recall from Definition~\ref{def:Renyi2prime} that
${(\omega'''_{\epsilon, \delta})}^B$ is the positive semidefinite 
operator obtained by zeroing out those eigenvalues of $\omega^B$ that 
are smaller than
$2^{-(1+\delta)\mathbb{H}_{\max}^\epsilon(B)_\omega}$.
Define the completely positive trace non-increasing superoperator
$
(\cT')^{A \to B} := 
\Pi^B_{\omega'''_{\epsilon,\delta}} \circ \hcT^{A \to B}.
$
From Definition~\ref{def:Renyi2prime} and
Fact~\ref{fact:gentle}, we can conclude 
\begin{align*}
\bigg\lVert
\left[\hcT^{A \to B} \otimes \I^R\right]\left(
(U \otimes I^R) \rho^{AR} (U^{ \dagger} \otimes I^R)
\right) -
\left[(\cT')^{A \to B} \otimes \I^R\right]\left(
(U \otimes I^R) \rho^{AR} (U^{ \dagger} \otimes I^R)
\right)
\bigg\rVert_1\\
\leq 2\sqrt{\epsilon} ~~~~~~~~~~~~~~~&
\end{align*}
in fact, we can conclude that
$
\lVert \hcT^{A \to B} - (\cT')^{A \to B} \rVert_\Diamond 
\leq 2\sqrt{\epsilon}.
$
Define the function
\[
f'(U) := 
\bigg\lVert
\left[(\cT')^{A \to B} \otimes \I^R\right]\left(
(U \otimes I^R) \rho^{AR} (U^{ \dagger} \otimes I^R)
\right)
- \omega^B \otimes \rho^R
\bigg\rVert_1.
\]
By triangle inequality, $|f'(U) - \hat{f}(U)| \leq 2\sqrt{\epsilon}$
which further implies that
$f(U) \leq 2 f'(U) + 4\sqrt{\epsilon}$.

Define the states $(\rho')^{AR}$, $\xi^R$ to be the ones achieving the 
optimum
in Definition~\ref{def:Renyi2} of $H_2^\epsilon(A|R)_\rho$ i.e.
\begin{equation} \label{eq:vanilla_H2}
\lVert 
(I^A \otimes \xi^R)^{-1/4}
(\rho')^{AR}
(I^A \otimes \xi^R)^{-1/4}
\rVert_2 =
2^{-\frac{1}{2} H_2^\epsilon(A|R)_\rho}.
\end{equation}
Define the function
\[
f''(U) := 
\bigg\lVert
\left[(\cT')^{A \to B} \otimes \I^R\right]\left(
(U^A \otimes I^R) (\rho')^{AR} (U^{A \dagger} \otimes I^R)
\right)
- \omega^B \otimes (\rho')^R
\bigg\rVert_1.
\]
By triangle inequality,
$|f''(U) - f'(U)| \leq 2\epsilon$ which implies that\\
$f(U) \leq 2 f''(U) + 8\sqrt{\epsilon}$.

Now define the Hermitian positive semidefinite matrix
$
(\omega')^{A'B} :=
\left[(\cT')^{A \to B} \otimes \I^{A'}\right](\Phi^{A A'}).
$
Observe that 
$
(\omega')^{A'B} = 
\Pi^B_{\omega'''_{\epsilon,\delta}} \circ \eta^{A'B}.
$
From Definition~\ref{def:Renyi2prime} and
Fact~\ref{fact:gentle}, we have 
$
\lVert
(\omega')^{A'B} -
\eta^{A'B}
\rVert_1 \leq 2\sqrt{\epsilon}
$
which further implies that
$
\lVert
(\omega')^{A'B} -
\omega^{A'B}
\rVert_1 \leq 3\sqrt{\epsilon}.
$
Define the function
\[
f'''(U) := 
\bigg\lVert
\left[(\cT')^{A \to B} \otimes \I^R\right]\left(
(U \otimes I^R) (\rho')^{AR} (U^{ \dagger} \otimes I^R)
\right)
- (\omega')^B \otimes (\rho')^R
\bigg\rVert_1.
\]
Again by triangle inequality,
$|f'''(U) - f''(U)| \leq 3\sqrt{\epsilon}$ which implies that\\
$f(U) \leq 2 f'''(U) + 14\sqrt{\epsilon}$.

Observe now that the range space of $(\cT')^{A \to B}$ is contained
in the support of $(\omega'''_{\epsilon,\delta})^B$. By 
Fact~\ref{fact:cs_tilde}, we can upper 
bound $f'''(U)$ by the function $g(U)$ defined by
\begin{equation}
\label{eq:gU}
g(U) := 
\bigg\lVert
\left[(\tcT')^{A \to B} \otimes \I^R\right]\left(
(U^A \otimes I^R) (\trho')^{AR} (U^{A \dagger} \otimes I^R)
\right)
- (\tomega')^B \otimes (\trho')^R
\bigg\rVert_2,
\end{equation}
where
$
(\tcT')^{A \to B} := 
((\omega'''_{\epsilon,\delta})^{-1/4})^B  \circ (\cT')^{A \to B},
$
\begin{equation}
\label{eq:tilde_rho}
(\trho')^{AR} :=
(I^A \otimes \xi^R)^{-1/4}
(\rho')^{AR}
(I^A \otimes \xi^R)^{-1/4},
\end{equation}
and
\begin{equation}
\label{eq:tilde_pomega}
(\tomega')^{A'B} :=
((\tcT')^{A \to B} \otimes \I^{A'})(\Phi^{A A'}) =
(I^{A'} \otimes (\omega'''_{\epsilon,\delta})^B)^{-1/4}
(\omega')^{A'B}
(I^{A'} \otimes (\omega'''_{\epsilon,\delta})^B)^{-1/4}.
\end{equation}
Thus,
$f(U) \leq 2 g(U) + 14\sqrt{\epsilon}$.
Recall from Definitions~\ref{def:Renyi2} and \ref{def:Renyi2prime} 
respectively, that
\begin{equation}
\label{eq:H2_rho_AR}
2^{-\frac{1}{2} H_2^\epsilon(A|R)_\rho} =
\lVert (\trho')^{AR} \rVert_2, \text{ and }
\end{equation}
\begin{equation}
\label{eq:H2_prime_epsilon}
2^{-\frac{1}{2} \mathbb{H}_2^{\epsilon,\delta}(A'|B)_\rho} =
\lVert (\tomega')^{A'B} \rVert_2.
\end{equation}

We have thus shown the following lemma.
\begin{lemma}
\label{lem:f_to_g}
Let $\mu, \kappa > 0$.
For all probability distributions on $U^A$,
\[
\mathbb{P}_{U^A}[f(U) > 2 \mu + 14\sqrt{\epsilon} + 2 \kappa]
\leq
\mathbb{P}_{U^A}[g(U) > \mu + \kappa].
\]
In particular this holds for
$
\mu =
2^{
-\frac{1}{2} H_2^\epsilon(A|R)_\rho 
-\frac{1}{2} \mathbb{H}_2^{\epsilon,\delta}(A'|B)_\omega 
} =
\lVert (\trho')^{AR} \rVert_2 \cdot
\lVert (\tomega')^{A'B} \rVert_2.
$
\end{lemma}

\subsection{Bounding centralised moments of $(g(U))^2$ under Haar measure}
\label{subsec:gmoment}
We now upper bound the tail of $g(U)$ when $U^A$ is chosen from
the Haar measure. For this we
need to upper bound the Lipschitz constant of $g(U)$ as follows. 
\begin{lemma}
\label{lem:lipschitz}
The Lipschitz constant of $g(U)$, that is for unitary operators 
$U,\,V$, the ratio $\frac{|g(U)-g(V)|}{\lVert U-V \rVert_2}$ is less than
$
2^{
\frac{1+\delta}{2} \mathbb{H}_{\max}^\epsilon(B)_\omega -
\frac{1}{2} H_2^\epsilon(A|R)_\rho + 1
}.
$ 
\end{lemma}
\begin{proof}
From Equations~\ref{eq:tilde_rho}, \ref{eq:vanilla_H2} and Definition~\ref{def:Renyi2} recall 
that $\lVert{(\trho')}^{AR} \rVert_2=2^{-1/2H_2^\epsilon(A|R)}$. 
Write $(\trho')^{AR}$ in any canonical tensor basis for 
$A \otimes R$:
\begin{align*}
(\trho')^{AR} &= 
\sum_{ij} \sum_{kl} \trho'_{ijkl} 
\ket{i}\bra{j}^A \otimes \ket{k}\bra{l}^R = 
\sum_{kl} (\tilde{M}'_{kl})^A \otimes \ket{k}\bra{l}^R\\
&=  \sum_{kl} \sum_{x} s_x^{kl} \ket{a_x^{kl}}\bra{b_x^{kl}}^A  \otimes 
	\ket{k}\bra{l}^R,
\end{align*}
where $\tilde{M}'_{kl} := \sum_{ij} \tilde{\rho}_{ijkl} 
\ket{i}\bra{j}^R$, and 
$\tilde{M}'_{kl} = \sum_{x} s_x^{kl} \ket{a^{kl}_x}\bra{b^{kl}_x}^A$ 
is the singular value decomposition of $\tilde{M}'_{kl}$.

Let $W_{\tcT'}^{AC \rightarrow BZ}$ be a Stinespring dilation of the
completely positive trace non-increasing map $\tcT'^{A \rightarrow B}$
provided by Fact~\ref{fact:stinespring}.\\
Thus 
$
(\tcT')^{A \rightarrow B}(M^A) =
\Tr_Z [W_{\tcT'} \circ (M^A \otimes 0^C)],
$
where $0^C := \ket{0}\bra{0}^C$.
Note that 
\begin{equation}
\label{eq:W_tilde_tau_prime}
W_{\tcT'}^{AC \rightarrow BZ} =
\left((\omega'''_{\epsilon,\delta})^B \otimes I^Z\right)^{-1/4}
(\Pi^B_{\omega'''_{\epsilon,\delta}} \otimes I^Z) 
(I^B \otimes P^Z) 
V_{\cT}^{AC \rightarrow BZ},
\end{equation}
where $V_{\cT}^{AC \rightarrow BZ}$ is the unitary Stinespring dilation
of $\cT^{A \rightarrow B}$ provided by Fact~\ref{fact:stinespring}
and $P^Z$, $\Pi^B_{\omega'''_{\epsilon,\delta}}$ and
$(\omega'''_{\epsilon,\delta})^B$ are defined in 
Section~\ref{sec:f(U)_to_g(U)}. We have
\begin{equation}
\label{eq:W_tilde_tau_prime_infty}
\lVert
W_{\tcT'}^{AC \rightarrow BZ}
\rVert_\infty \leq 
\lVert
((\omega'''_{\epsilon,\delta})^B \otimes I^Z)^{-1/4}
\rVert_\infty =
\lVert
((\omega'''_{\epsilon,\delta})^B)^{-1}
\rVert_\infty^{1/4} \leq
2^{\frac{(1+\delta)}{4} \mathbb{H}_{\max}^\epsilon(B)_\omega},
\end{equation}
where the last inequality follows from the definition of
$(\omega'''_{\epsilon,\delta})^B$ given in 
Definition~\ref{def:Renyi2prime}. Moreover,
\begin{equation}
\label{eq:W_tilde_tau_prime_two}
\lVert
W_{\tcT'}^{AC \rightarrow BZ}
\rVert_2 \leq 
\sqrt{|B| |Z|} 
\lVert
W_{\tcT'}^{AC \rightarrow BZ}
\rVert_\infty \leq 
\sqrt{|B| |A|} \cdot
2^{\frac{(1+\delta)}{4} \mathbb{H}_{\max}^\epsilon(B)_\omega},
\end{equation}
where we used that $|Z| \leq |A|$ guaranteed by 
Fact~\ref{fact:stinespring}.

Let $U^A$, $V^A$ be two unitaries on $A$. Then,
\begin{align} 
\lefteqn{\bigg|g(U)-g(V)\bigg|} \label{eq:lem2_1}\\
& \leq 
\bigg\lVert 
\left[(\tcT')^{A \rightarrow B} \otimes \I^R\right]
\left((U^A \otimes I^R ) \circ (\trho')^{AR}\right) - 
\left[(\tcT')^{A \rightarrow B} \otimes \I^R\right]
\left((V^A \otimes I^R ) \circ (\trho')^{AR}\right)
\bigg\rVert_2  \label{eq:lem2_2} \\
&  =   
\bigg\lVert 
\left[(\tcT')^{A \rightarrow B} \otimes \I^R\right]
\left(
(U^A \otimes I^R ) \circ (\trho')^{AR} - 
(V^A \otimes I^R ) \circ (\trho')^{AR}
\right)
\bigg\rVert_2 \nonumber \\
& \leq    
\bigg\lVert 
\left[(\tcT')^{A \rightarrow B} \otimes \I^R\right]
\left(
(U^A \otimes I^R)(\trho')^{AR}(U^A \otimes I^R)^\dag -
(U^A \otimes I^R)(\trho')^{AR}(V^A \otimes I^R)^\dag
\right)
\bigg\rVert_2 \nonumber\\
&
~~
{} +
\bigg\lVert 
\left[(\tcT')^{A \rightarrow B} \otimes \I^R\right]
\left(
(U^A \otimes I^R)(\trho')^{AR}(V^A \otimes I^R)^\dag -
(V^A \otimes I^R)(\trho')^{AR}(V^A \otimes I^R)^\dag
\right)
\bigg\rVert_2. \nonumber
\end{align}
We now upper bound
\begin{align*}
\lefteqn{
\bigg\lVert 
\left[(\tcT')^{A \rightarrow B} \otimes \I^R\right]
\left(
(U^A \otimes I^R)(\trho')^{AR}(U^A \otimes I^R)^\dag -
(U^A \otimes I^R)(\trho')^{AR}(V^A \otimes I^R)^\dag
\right)
\bigg\rVert_2
} \\
& =
\bigg\lVert 
\sum_{kl} 
\left(
\sum_x s^{kl}_x 
(\tcT')^{A \rightarrow B}
\left(
(U \ket{a_x^{kl}}^A\bra{b_x^{kl}} U^\dag) -
(U \ket{a_x^{kl}}^A\bra{b_x^{kl}} V^\dag)
\right) 
\right) \otimes \ket{k}^R\bra{l}
\bigg\rVert_2 \\
& =
\sqrt{
\sum_{kl} 
\left\lVert 
\sum_x s^{kl}_x 
(\tcT')^{A \rightarrow B}
\left(
(U \ket{a_x^{kl}}\bra{b_x^{kl}}^A U^\dag) -
(U \ket{a_x^{kl}}\bra{b_x^{kl}}^A V^\dag)
\right)
\right\rVert_2^2
}.
\end{align*}
Fix $k$, $l$. For ease of notation drop the superscript $kl$ below. 
We now upper bound
\begin{align*}
\lefteqn{
\left\lVert 
\sum_x s_x 
(\tcT')^{A \rightarrow B}
(
(U \ket{a_x}^A\bra{b_x} U^\dag) -
(U \ket{a_x}^A\bra{b_x} V^\dag)
) 
\right\rVert_2
} \\
& = 
\left\lVert 
\sum_x s_x
\bigg(
\Tr_Z [W_{\tcT'} \circ ((U \ket{a_x}^A\bra{b_x} U^\dag) \otimes 0^C)] -
\Tr_Z [W_{\tcT'} \circ ((U \ket{a_x}^A\bra{b_x} V^\dag) \otimes 0^C)] 
\bigg)
\right\rVert_2 \\
& \overset{a}{=} 
\left\lVert 
\sum_x s_x
\left(
P_{x,U}^{B \times Z} (Q_{x,U}^{B \times Z})^\dag -
P_{x,U}^{B \times Z} (Q_{x,V}^{B \times Z})^\dag 
\right)
\right\rVert_2 \\
& \overset{b}{=} 
\left\lVert 
P_{U}^{BQ \times ZA} (Q_{U}^{BQ \times ZA})^\dag -
P_{U}^{BQ \times ZA} (Q_{V}^{BQ \times ZA})^\dag 
\right\rVert_2 \\
& \leq
\left\lVert P_{U}^{BQ \times ZA} \right\rVert_2
\left\lVert
Q_{U}^{BQ \times ZA} -
Q_{V}^{BQ \times ZA}
\right\rVert_2 \\
& \overset{c} {\leq} 
2^{\frac{1+\delta}{2} \mathbb{H}_{\max}^\epsilon(B)_\omega}
\lVert \tilde{M}'_{kl} \rVert_2
\lVert U-V \rVert_2,
\end{align*}
where
\vspace*{-1mm}
\begin{itemize}
\item[(a)] 
\[
P_{x,U}^{B \times Z} := 
(\vectorise^{B,Z})^{-1}(
W_{\tcT'}^{AC \rightarrow BZ}(U^A \otimes I^C) 
(\ket{a_x}^A \otimes\ket{0}^C)
),
\]
\[
Q_{x,U}^{B \times Z} := 
(\vectorise^{B,Z})^{-1}(
W_{\tcT'}^{AC \rightarrow BZ}(U^A \otimes I^C) 
(\ket{b_x}^A \otimes\ket{0}^C)
),
\]
and $Q_{x,V}^{B \times Z}$ is defined similarly.
The above operators map system $Z$ to system $B$ or are 
$B \times Z$ matrices for fixed bases of $B$ and $Z$. The 
equality holds due to Fact~\ref{fact:vec_partial_trace}.

\item[(b)] 
Let $Q$ be a single qubit register and $x$ range over the computational
basis of $A$.
\[
P_U^{BQ \times ZA} := 
\sum_x s_x (P_{x,U}^{B \times Z} \otimes \ket{0}^Q \bra{x}^A),
~~~
Q_U^{BQ \times ZA} := 
\sum_x (Q_{x,U}^{B \times Z} \otimes \ket{0}^Q \bra{x}^A),
\]
and $Q_V^{BQ \times ZA}$ is defined similarly. The equality follows
by inspection.

\item[(c)]
We have
\begin{eqnarray*}
\lVert P_{U,kl}^{BQ \times ZA} \rVert_2^2 
& = &
\sum_x (s_x^{kl})^2 \lVert P_{x,U,kl}^{B \times Z} \rVert_2^2 
\;=\;
\sum_x (s_x^{kl})^2 
\lVert 
W_{\tcT'}^{AC \rightarrow BZ}(U^A \otimes I^C) 
(\ket{a_x^{kl}}^A \otimes\ket{0}^C)
\rVert_2^2 \\
& \leq &
\lVert 
W_{\tcT'}^{AC \rightarrow BZ}
\rVert_\infty^{2}
\sum_x (s_x^{kl})^2 
\;\leq\;
2^{\frac{1+\delta}{2} \mathbb{H}_{\max}^\epsilon(B)_\omega}
\lVert \tilde{M}'_{kl} \rVert_2^2,
\end{eqnarray*}
where the last inequality follows from 
Equation~\ref{eq:W_tilde_tau_prime_infty}.  
Again using Equation~\ref{eq:W_tilde_tau_prime_infty} in the second
and third inequality below, we get
\begin{eqnarray*}
\lefteqn{
\left\lVert Q_{U}^{BQ \times ZA} - Q_{V}^{BQ \times ZA} \right\rVert_2^2 
} \\
& = &
\sum_x \left\lVert Q_{x,U}^{B \times Z} - 
Q_{x,V}^{B \times Z} \right\rVert_2^2 \\
& = &
\sum_x 
\left\lVert 
W_{\tcT'}^{AC \rightarrow BZ}(U^A \otimes I^C) 
(\ket{b_x}^A \otimes\ket{0}^C) -
W_{\tcT'}^{AC \rightarrow BZ}(V^A \otimes I^C) 
(\ket{b_x}^A \otimes\ket{0}^C)
\right\rVert_2^2 \\
& = &
\sum_x 
\left\lVert 
W_{\tcT'}^{AC \rightarrow BZ}
((U^A \otimes I^C) - (V^A \otimes I^C))
(\ket{b_x}^A \otimes\ket{0}^C)
\right\rVert_2^2 \\
& \leq &
\lVert 
W_{\tcT'}^{AC \rightarrow BZ}
\rVert_\infty^{2}
\sum_x 
\left\lVert 
((U^A \otimes I^C) - (V^A \otimes I^C))
(\ket{b_x}^A \otimes\ket{0}^C)
\right\rVert_2^2 \\
& \leq &
\lVert 
((\omega'''_{\epsilon,\delta})^B)^{-1}
\rVert_\infty^{1/2}
\sum_x 
\left\lVert 
((U^A \otimes I^C) - (V^A \otimes I^C))
(\ket{b_x}^A \otimes\ket{0}^C)
\right\rVert_2^2 \\
& =    &
\lVert ((\omega'''_{\epsilon,\delta})^B)^{-1} \rVert_\infty^{1/2}
\sum_x \lVert (U - V)^A \ket{b_x} \rVert_2^2 \\
& =   &
\lVert ((\omega'''_{\epsilon,\delta})^B)^{-1} \rVert_\infty^{1/2}
\sum_x \bra{b_x} (U-V)^\dag (U - V) \ket{b_x} \\
& =    &
\lVert ((\omega'''_{\epsilon,\delta})^B)^{-1} \rVert_\infty^{1/2}
\Tr [(U-V)^\dag (U - V)] 
\; \leq   \;
2^{\frac{1+\delta}{2} \mathbb{H}_{\max}^\epsilon(B)_\omega}
\lVert U-V \rVert_2^2.
\end{eqnarray*}
$
\Rightarrow \lVert P_{U}^{BQ \times ZA} \rVert_2
\lVert
Q_{U}^{BQ \times ZA} -
Q_{V}^{BQ \times ZA}
\rVert_2 
\leq
2^{\frac{1+\delta}{2} \mathbb{H}_{\max}^\epsilon(B)_\omega}
\lVert \tilde{M}'_{kl} \rVert_2
\lVert U-V \rVert_2,
$
thereby proving the desired inequality.
\end{itemize}
This implies that
\begin{eqnarray*}
\lefteqn{
\left\lVert 
\left[(\tcT')^{A \rightarrow B} \otimes \I^R\right]
\left(
(U \otimes I)(\trho')^{AR}(U \otimes I)^\dag -
(U \otimes I)(\trho')^{AR}(V \otimes I)^\dag
\right)
\right\rVert_2
} \\
& \leq &
\sqrt{
\sum_{kl}
2^{(1+\delta) \mathbb{H}_{\max}^\epsilon(B)_\omega}
\lVert \tilde{M}'_{kl} \rVert_2^2
\lVert U-V \rVert_2^2 
} \\
&  =   &
2^{\frac{1+\delta}{2} \mathbb{H}_{\max}^\epsilon(B)_\omega}
\lVert U-V \rVert_2 
\sqrt{
\sum_{kl}
\lVert \tilde{M}'_{kl} \rVert_2^2
} 
\; =  \;
2^{\frac{1+\delta}{2} \mathbb{H}_{\max}^\epsilon(B)_\omega}
\lVert U-V \rVert_2 
\lVert (\trho')^{AR} \rVert_2.
\end{eqnarray*}
Hence
$
|g(U) - g(V)| \leq
2 \cdot 2^{\frac{1+\delta}{2} \mathbb{H}_{\max}^\epsilon(B)_\omega}
\lVert (\trho')^{AR} \rVert_2
\lVert U-V \rVert_2.
$
This completes the proof of the lemma.
\end{proof}

\begin{note}
\label{note:W_tilde(T)}
The proof technique to find a tighter Lipschitz constant of the 
function $g$ (as opposed to the naive one) entails the evaluation of 
the Schatten-$2$ norm of the partial isometry 
$W_{\tcT'}^{AC \rightarrow BZ}$. This is imperative in finding the 
closeness between the higher order moments of the function $g(U)^2$ under 
Haar measure and quantum tensor product expander in 
Lemma~\ref{lem:MomentTPEHaar}. Hence we re-write 
Equation~\ref{eq:W_tilde_tau_prime_two} below:
\begin{equation*}
\lVert
W_{\tcT'}^{AC \rightarrow BZ}
\rVert_2 \leq 
\sqrt{|B| |Z|} 
\lVert
W_{\tcT'}^{AC \rightarrow BZ}
\rVert_\infty \leq 
\sqrt{|B| |A|} \cdot
2^{\frac{(1+\delta)}{4} \mathbb{H}_{\max}^\epsilon(B)_\omega}.
\end{equation*}
\end{note}

\begin{lemma}
\label{lem:maxg}
For any unitary $U \in \U(A)$, 
$
g(U) \leq
(2|A|)^{1/2} \cdot 
2^{
\frac{1+\delta}{2} \mathbb{H}_{\max}^\epsilon(B)_\omega -
\frac{1}{2} H_2^\epsilon(A|R)_\rho
}. 
$
\end{lemma}
\begin{proof}
Define the Hermitian matrix 
$\gamma^{AR} := (\trho')^{AR} - \pi^A \otimes (\trho')^R$. We have
\begin{equation}
\label{eq:elltwogamma}
\begin{array}{rcl}
\lefteqn{\lVert \gamma^{AR} \rVert_2^2 } \\
& = &
\lVert (\trho')^{AR} \rVert_2^2 +
\lVert \pi^A \otimes (\trho')^R \rVert_2^2 -
\langle (\trho')^{AR}, (\pi^A \otimes (\trho')^R) \rangle -
\langle (\pi^A \otimes (\trho')^R), (\trho')^{AR} \rangle \\
& \leq &
\lVert (\trho')^{AR} \rVert_2^2 +
\lVert \pi^A \otimes (\trho')^R \rVert_2^2
\;\leq\;
2 \lVert (\trho')^{AR} \rVert_2^2, 
\end{array}
\end{equation}
where we used the fact that $(\trho')^{AR}$, $(\pi^A \otimes (\trho')^R)$
are positive semidefinite matrices in the first inequality and
Fact~\ref{fact:Dupuis_3.5} in the second inequality.

Observe that 
$
g(U) = 
\left\lVert 
[(\tcT')^{A \to B} \otimes \I^R]\left((U^A \otimes I^R) \circ 
\gamma^{AR}\right)
\right\rVert_2.
$
Arguing similarly as in the proof of Lemma~\ref{lem:lipschitz} (proceeding along the similar lines as in Equations~\ref{eq:lem2_1}, \ref{eq:lem2_2} without the second term inside the norm and similarly leading to the inequalities (a), (b) and (c) without involving their respective second terms which are subtracted),
we can conclude that
\[
g(U) \leq
2^{\frac{1+\delta}{2} \mathbb{H}_{\max}^\epsilon(B)_\omega} 
\lVert \gamma^{AR} \rVert_2 \cdot \lVert U \rVert_2 \leq
(2|A|)^{1/2} \cdot 
2^{
\frac{1+\delta}{2} \mathbb{H}_{\max}^\epsilon(B)_\omega -
\frac{1}{2} H_2^\epsilon(A|R)_\rho
}, 
\]
where we have used $\lVert U \rVert_2 = |A|^{1/2}$ and 
Equation~\ref{eq:H2_prime_epsilon} in the last inequality. 
This completes the proof.
\end{proof}

We now apply Levy's Lemma (Fact~\ref{fact:Levy}) 
to obtain an exponential upper bound on the deviation of $g(U)$
about its expectation $\mu$ when $U$ is chosen from the Haar measure.
\begin{proposition}
\label{prop:decoupling_general_Haar}
Define $\mu := \mathbb{E}_{U \sim \Haar}[g(U)]$. 
Define 
$
a :=
|A| 
2^{H_2^\epsilon(A|R)_\rho-
(1+\delta)\mathbb{H}_{\max}^\epsilon(B)_\omega - 4}.
$
Let $\kappa > 0$. 
Then
\begin{equation} \label{eq:concHaar}
\prob_{U \sim \Haar}[|g(U) - \mu| > \kappa] \leq 
2 \exp(-a \kappa^2).
\end{equation}
Additionally,
\begin{eqnarray} 
\mu^2 
& \leq &
\E_{U \sim \Haar}[(g(U))^2] 
\; = \;
\alpha \lVert (\trho')^R \rVert_2^2 + 
\beta \lVert (\trho')^{AR} \rVert_2^2 -
\lVert (\tomega')^B \rVert_2^2 \lVert (\trho')^R \rVert_2^2 \label{eq:equality}\\
& < &
\lVert (\tilde{\omega}')^{A'B}\rVert_2^2 \cdot
\lVert (\tilde{\rho}')^{AR}\rVert_2^2 
\;=\;
2^{-\mathbb{H}_2^{\epsilon,\delta}(A'|B)_\omega - H_2^\epsilon(A|R)_\rho}, \label{eq:inequality}
\end{eqnarray}
where
$(\tilde{\omega}')^{A^\prime B}$, 
$(\tilde{\rho}')^{AR}$ are defined in Equations~\ref{eq:tilde_pomega}, 
\ref{eq:tilde_rho} respectively, 
\[
\delta_1 :=
\lVert (\tomega')^B \rVert_2^2 
\frac{|A|^2 - |A| \eta}{|A|^2 - 1},
~~
\delta_2 :=
\lVert (\tilde{\omega}')^{A'B}\rVert_2^2 
\frac{|A|^2 - |A| \eta^{-1}}{|A|^2 - 1},
~~
\eta := 
\frac{
\lVert (\tilde{\omega}')^{A'B}\rVert_2^2
}{
\lVert (\tomega')^B \rVert_2^2 
}.
\]
Moreover if
\[
a \cdot \E_{U \sim \Haar}[(g(U))^2] +
\log \E_{U \sim \Haar}[(g(U))^2] >
\log |A| +
\mathbb{H}_{\max}^\epsilon(B)_\omega -
H_2^\epsilon(A|R)_\rho, \text{ then }
\]
\begin{align} \label{eq:mu2}
&\E_{U \sim \Haar}[(g(U))^2] \leq 8\mu^2.
\end{align}
\end{proposition}       
\begin{proof}
Equations~(3.31, 3.32, 3.33, 3.37 and 3.43) in \cite{decoupling} and Equations~(3.44, 3.45 and 3.46) used in the proof of Fact~\ref{fact:dupuisexpectation} (in \cite{decoupling})
implies Equation~\ref{eq:equality} and Equation~\ref{eq:inequality} rspectively.
Fact~\ref{fact:Levy} applied to the function $g(U)$ 
with upper bound on the Lipschitz constant given by 
Lemma~\ref{lem:lipschitz} gives Equation~\ref{eq:concHaar}. 
Further,
\begin{eqnarray*}
\lefteqn{
\E_{U \sim \Haar}[(g(U))^2] 
} \\
& \overset{a}{=}  &
 \E_{U \sim \Haar}\left[(g(U))^2 \one \left\{ g(U) \leq \left(\mu + 2^{-1} 
\left\{\underset{U \sim \Haar}{\E}\left[(g(U))^2\right]\right\}^{1/2}
\right) \right\} \right] \\
&       &
  {} +
  \E_{U \sim \Haar} \left[(g(U))^2 \one\left\{ g(U) \leq \left(\mu + 2^{-1} 
\left\{\underset{U \sim \Haar}{\E}\left[(g(U))^2\right]\right\}^{1/2}
\right) \right\} \right]\\
& \overset{b}{\leq} &
\left(\mu + 2^{-1} 
\left\{\underset{U \sim \Haar}{\E}\left[(g(U))^2\right]\right\}^{1/2}
\right)^2 \\
&   &
{} + 
\left[(2|A|)^{1/2} \cdot 
2^{
\frac{1+\delta}{2} \mathbb{H}_{\max}^\epsilon(B)_\omega -
\frac{1}{2} H_2^\epsilon(A|R)_\rho
} \right]^2\cdot
\underset{U \sim \Haar}{\prob}\left[
g(U) - \mu > 2^{-1} 
\left\{\underset{U \sim \Haar}{{\E}}[(g(U))^2]\right\}^{1/2}
\right] \\
& \leq &
2 \mu^2 + 2^{-1} \underset{U \sim \Haar}{\E}[(g(U))^2] \\
&   &
{} +
{\left[ (2|A|)^{1/2} \cdot 
2^{
\frac{1+\delta}{2} \mathbb{H}_{\max}^\epsilon(B)_\omega -
\frac{1}{2} H_2^\epsilon(A|R)_\rho
}\right] }^2 \cdot
2 \exp(-2^{-2} a \cdot \underset{U \sim \Haar}{\E}[(g(U))^2]) \\
& \leq &
2 \mu^2 + \frac{3}{4} \underset{U \sim \Haar}{\E}[(g(U))^2].
\end{eqnarray*}
where in (a) the notation $\one\left\{g(U) \leq \left(\mu + 2^{-1} 
\left[\underset{U \sim \Haar}{\E}\left[(g(U))^2\right]\right]^{1/2}
\right)^2\right\}$ denotes the indicator function of the event in its argument, taking a value one if the event occurs and zero otherwise; (b) holds because the first summand uses the fact $\E(\one(event)) = \prob(event) \leq 1$ and we use Lemma~\ref{lem:maxg} to upper bound the second summand.\\  
This implies that 
$
\underset{U \sim \Haar}{{\E}}[(g(U))^2] \leq 8 \mu^2,
$
completing the proof of Equation~\ref{eq:mu2} and hence of the proposition.
\end{proof}

We now evaluate the higher order moments of the functions 
$g(U) - \mu$ and $(g(U))^2 - \mu^2$.
\begin{lemma}
\label{lem:Haar_moments_g}
Define $\mu := \mathbb{E}_{U \sim \Haar}[g(U)]$.
Let $m$ be a positive integer.
Then
the $(2m)$-th moments of the functions $g(U)-\mu$ and $(g(U))^2-\mu^2$ 
are upper bounded by
\[
\E_{U \sim \Haar}(\lvert g(U)-\mu \rvert^{2m}) \leq 
2 (m 
2^{
(1+\delta)\mathbb{H}_{\max}^\epsilon(B)_\omega -
H_2^{\epsilon}(A|R)_\rho + 4
} |A|^{-1}
)^m,
\]
\begin{eqnarray*}
\lefteqn{\E_{U \sim \Haar}(\lvert (g(U))^2-\mu^2 \rvert^{2m})} \\
& \leq &
\left\{
\begin{array}{l l}
6 (9 m \mu^2 
2^{
(1+\delta)\mathbb{H}_{\max}^\epsilon(B)_\omega -
H_2^{\epsilon}(A|R)_\rho + 4
} |A|^{-1}
)^m
& ; m \leq
\frac{9}{64} |A| \mu^2 \cdot
2^{
-(1+\delta)\mathbb{H}_{\max}^\epsilon(B)_\omega +
H_2^{\epsilon}(A|R)_\rho - 4
} \\
6 (64 m^2 
2^{
2 (1+\delta)\mathbb{H}_{\max}^\epsilon(B)_\omega -
2 H_2^{\epsilon}(A|R)_\rho + 8
} |A|^{-2}
)^m
& \mbox{; otherwise}
\end{array}
\right..
\end{eqnarray*}
\end{lemma}
\begin{proof}
Applying Fact~\ref{fact:moment} to the non-negative random variable
$|g(U) - \mu|$ with concentration given by 
Proposition~\ref{prop:decoupling_general_Haar} 
gives the bounds of the lemma.
\end{proof}

\begin{note}
It might might seem apparent at this point that one can directly apply 
\cite[Lemma~3.4]{low_2009} to obtain the expectation of the function 
$\left[g^2(U)-\mu^2 \right]^{2m}$ and then apply Markov's inequality to 
conclude the result. However, there is a technical subtlety here. 
Low's Lemma~3.4 in \cite{low_2009} has an additive parameter `$\alpha$', 
which is the sum of the coefficients in the polynomial expansion of 
$g^2(U)$. Since the function $g$ depends the perturbed superoperator 
$\cT'$ (and not $\cT$), finding the value of $\alpha$ is more complicated 
then our subsequent analysis (especially Lemma~6). Furthermore, applying 
Low's technique at this point with crude approximation (still not so easy)
will give a much larger value of $t$ than what we get by appealing to 
the definition of quantum tensor product expanders. It might not be 
feasible to conclude that under any assumptions computationally 
efficient designs can lead to decoupling. In contrast, our analysis 
results in computationally efficient decoupling under the setting 
Theorem~2, for $|A_1|=\mathrm{poly} (\log |A_2|)$ and $\kappa=O(\mu)$. 
Hence our subsequent method that leads to Lemma~6 and 7 are needed to 
obtain a handle on the parameters $\lambda,\, m$ of the quantum tensor 
product expanders. This provides a way to analyze how large a value of 
the unitary design parameter $t$ is required to obtain an exponentially 
decaying tail for any deviation $\kappa$.
\end{note}

\subsection{Concentration of $(g(U))^2$ under $t$-design}
In this section we finally obtain an exponential concentration for
$(g(U))^2$ when $U$ is chosen uniformly at random from a unitary $t$-design
for suitable $t$. We first prove the following lemma.
\begin{lemma}
\label{lem:TtwoSwap}
Let $\cT^{A \to B}$ be a completely positive superoperator with
Stinespring dilation $W_{\cT}^{AC \to BZ}$, where
$|A| |C| = |B| |Z|$, the input ancillary system is $C$ 
and the output ancillary system is $D$. Let 
$F^{A_1 A_2}$ and $F^{B_1 B_2}$ be the appropriate
swap operators. Then
$$
\left\lVert 
\left((\cT^{\dagger})^{B_1 \to A_1} \otimes 
(\cT^{\dagger})^{B_2 \to A_2}\right)
(F^{B_1 B_2})
\right\rVert_2 = 
\left\lVert 
\left(\cT^{A_1 \to B_1} \otimes \cT^{A_2 \to B_2}\right)(F^{A_1 A_2})
\right\rVert_2 \leq 
\lVert W_{\cT} \rVert_2^4.
$$ 
\end{lemma}
\begin{proof}
By Stinespring representation of $\cT$ as given 
in Fact~\ref{fact:stinespring}, \\
$
\cT^{A \to B}(M^A) = 
\Tr_Z \left[W_{\cT} (M^A \otimes \ket{0}\bra{0}^C) W_{\cT}\right]
$
for any $M^A \in \cL(A)$.
Expressing the swap operator $F^{A_1 A_2}$ in computational basis,
we have
\[
F^{A_1 A_2} =
\sum_{a a'} 
\{\ket{a'}\ket{a}\}\{\bra{a}\bra{a'}\}^{A_1 A_2} =
\sum_{a a'} 
\ket{a'}\bra{a}^{A_1} \otimes \ket{a}\bra{a'}^{A_2}.
\]
Note that the swap operator is Hermitian. Further, observe that
\begin{align*}
\lefteqn{
\left\lVert
\left[(\cT^{\dagger})^{B_1 \to A_1} \otimes 
(\cT^{\dagger})^{B_2 \to A_2}\right]
(F^{B_1 B_2})
\right\rVert_2^2
} \\
&  \overset{a}{=} 
\Tr \left[
\left\{
\left((\cT^{\dagger})^{B_1 \to A_1} \otimes 
(\cT^{\dagger})^{B_2 \to A_2}\right)
(F^{B_1 B_2})
\right\}^2
\right] \\
&  \overset{b}{=} 
\Tr \left[
\left\{
\left(
(\cT^{\dagger} \otimes \cT^{\dagger})
F^{B_1 B_2}
\right)
\otimes
\left(
(\cT^{\dagger}\otimes \cT^{\dagger})
F^{B_1' B_2'}
\right)
\right\} F^{(A_1 A_2) (A_1' A_2')}
\right] \\
& \overset{c}{=} 
\Tr \left[
(F^{B_1 B_2} \otimes F^{B_1' B_2'})
\left\{
\left[(\cT^{A_1 \to B_1} \otimes \cT^{A_2 \to B_2}) \otimes
(\cT^{A_1' \to B_1'} \otimes \cT^{A_2' \to B_2'})\right]
 F^{(A_1 A_2) (A_1' A_2')}
\right\}
\right] \\
& \overset{d}{=} 
\Tr \left[
F^{(B_1 B_1') (B_2 B_2')}
\left\{
[(\cT^{A_1 \to B_1} \otimes \cT^{A_1' \to B_1'})] \otimes
[\cT^{A_2 \to B_2} \otimes \cT^{A_2' \to B_2'}]
(F^{A_1 A_1'} \otimes F^{A_2 A_2'})
\right\}
\right] \\
& = 
\Tr \left[
F^{(B_1 B_1') (B_2 B_2')}
\left\{
[\cT^{A_1 \to B_1} \otimes \cT^{A_1' \to B_1'}](F^{A_1 A_1'}) \otimes
[\cT^{A_2 \to B_2} \otimes \cT^{A_2' \to B_2'}](F^{A_2 A_2'})
\right\} 
\right] \\
& \overset{e}{=} 
\Tr \left[
\left\{
[\cT^{A_1 \to B_1} \otimes \cT^{A_1' \to B_1'}](F^{A_1 A_1'})\right\}^2
\right] \\
&= 
\left\lVert
[\cT^{A_1 \to B_1} \otimes \cT^{A_2 \to B_2}](F^{A_1 A_2})
\right\rVert_2^2,
\end{align*}
where in
\begin{itemize}

\item[(a)] 
we use the fact that $\cT^\dag$ is completely positive
as $\cT$ is completely positive and the fact that the swap operator
is Hermitian, implying that
$
((\cT^{\dagger})^{B_1 \to A_1} \otimes (\cT^{\dagger})^{B_2 \to A_2})
(F^{B_1 B_2})
$ 
is Hermitian,

\item[(b)] 
we take $A_1'$, $A_2'$ to be two new systems 
of the same dimension as $A$, $F^{(A_1 A_2) (A_1' A_2')}$ as the
operator swapping $(A_1 A_2)$ with $(A_1' A_2')$ and 
Fact~\ref{fact:swap},

\item[(c)]
we use the definition of the adjoint of a superoperator under the 
Hilbert-Schmidt inner product,

\item[(d)]
we use a property of the swap operator,

\item[(e)]
we use Fact~\ref{fact:swap}.

\end{itemize}
This proves the first of the equalities asserted above.

Finally,
\begin{align*}
\lefteqn{
\left\lVert
[\cT^{A_1 \to B_1} \otimes \cT^{A_2 \to B_2}](F^{A_1 A_2})
\right\rVert_2
} \\
&= 
\left\lVert 
\sum_{a a'}
\left\{
\left(\Tr_Z \left[
W_{\cT} 
(\ket{a'}\bra{a}^{A_1} \otimes \ket{0}\bra{0}^C) 
W_{\cT}^{\dagger}
\right]
\right)^{B_1} \otimes 
\left(
\Tr_Z \left[
W_{\cT} 
(\ket{a}\bra{a'}^{A_2} \otimes \ket{0}\bra{0}^C ) 
W_{\cT}^{\dagger}
\right] \right)^{B_2}
\right\}
\right\rVert_2 \\
&= 
\left\lVert 
\sum_{a a'}
\left\{
\left( \Tr_Z \left[
W_{\cT} 
(\ket{a'}\ket{0})(\bra{a}\bra{0})^{A_1 C} 
W_{\cT}^{\dagger}
\right]
\right)^{B_1} \otimes 
\left(
\Tr_Z \left[
W_{\cT} 
(\ket{a}\ket{0})(\bra{a'}\bra{0})^{A_2 C} 
W_{\cT}^{\dagger}
\right]
\right)^{B_2} \right\}
\right\rVert_2 \\
& \overset{a}{=} 
\left\lVert 
\sum_{a a'} 
(P_{a'} P_a^\dagger)^{B_1} \otimes (P_a P_{a'}^\dagger)^{B_2}
\right\rVert_2 
\;\leq\; 
\sum_{a a'} 
\left\lVert 
(P_{a'} P_a^\dagger)^{B_1} \otimes (P_a P_{a'}^\dagger)^{B_2}
\right\rVert_2 \\
& =
\sum_{a a'} 
\left\lVert (P_{a'} P_a^\dagger)^{B_1} \right\rVert_2 \cdot
\left\lVert (P_a P_{a'}^\dagger)^{B_2} \right\rVert_2 
\; \leq \; 
\sum_{a a'} 
\lVert P_{a'}^{B \times Z} \rVert_2^2 \cdot 
\lVert P_a^{B \times Z} \rVert_2^2 
\; = \;
\left(
\sum_{a} 
\lVert P_a^{B \times Z} \rVert_2^2 
\right)^2 \\
& =
\left(
\sum_{a} 
\left\lVert \left(W_\cT (\ket{a}^{A} \otimes \ket{0}^C)\right)^{BZ} 
\right\rVert_2^2 
\right)^2 
\; \leq
\left(
\sum_{a c} 
(\bra{a}^{A}\bra{c}^C) W_\cT^\dag W_\cT (\ket{a}^{A}\ket{c}^C)
\right)^2 \\
& =
\left(
\Tr [W_\cT^\dag W_\cT]
\right)^2 
\; = \;
\lVert W_{\cT} \rVert_2^4,
\end{align*}
where in (a) we define 
$
P_{a}^{B \times Z} := 
(\vectorise^{B,Z})^{-1}((W_{\cT}(\ket{a}^A \ket{0}^C))^{BZ})
$ 
and use Fact~\ref{fact:vec_partial_trace}.
This completes the proof of the lemma.
\end{proof}

Note that $(g(U))^2$ is a balanced degree two polynomial in the matrix
entries of $U$. We now find out how close the moments of
$(g(U))^2$ under Haar measure are to their counterparts under
$t$-design. 

We do this by using the following Proposition~\ref{prop:g_exp_2i}.
\begin{proposition}
\label{prop:g_exp_2i}
The $(2i)^{th}$ power of the function $g(U)$ which is equivalent to the 
simplification of the balanced degree two polynomial polynomial 
$[g(U)^2]^i$ is given by:
\begin{eqnarray*}
    \lefteqn{
            (g(U)^{2})^{i}
            } \\
    & = &
     \Tr \bigg[
    \{
    (\gamma')^{A_1 A_2}
    \}^{\otimes i} \cdot\\
    &  &
    ~~~~~~~
    \left\{
    \left(U^\dag)^{A_1} \otimes (U^\dag)^{A_2}\right)^{\otimes i} \circ
    \left[((\tcT')^\dag)^{B_1 \to A_1} \otimes 
    ((\tcT')^\dag)^{B_2 \to A_2}\right]
    (F^{B_1 B_2})
    \right\}^{\otimes i}
    \bigg],
\end{eqnarray*}
where we define the Hermitian matrices 
$\gamma^{AR} := (\trho')^{AR} - \pi^A \otimes (\trho')^R$ and \\
$
(\gamma')^{A_1 A_2} :=
\Tr_{R_1 R_2}
\left\{
(I^{A_1 A_2} \otimes F^{R_1 R_2})
(\gamma^{A_1 R_1} \otimes \gamma^{A_2 R_2})
\right\}.
$   
\end{proposition}
\begin{proof} 
Let $A_1$, $A_2$ denote two Hilbert spaces of the same dimension as
$A$; similarly for $B_1$, $B_2$ and $R_1$, $R_2$.
Further, $A_1(j)$, $1 \leq j \leq i$ denote Hilbert spaces of the same 
dimension as $A_1$, likewise for $A_2(j)$, $B_1(j)$, $B_2(j)$,
$R_1(j)$ and $R_2(j)$.
Observe that
\begin{eqnarray*}
\lefteqn{
(g(U)^{2})^{i}
} \\
& = &
\left(
\left\lVert 
[(\tcT')^{A \to B} \otimes \I^R]\left((U^A \otimes I^R) \circ 
(\trho')^{AR}\right) - 
(\tomega'^B) \otimes (\trho')^R
\right\rVert_2^2
\right)^i \\
& = &
\left(
\Tr \left[
\left\{
[(\tcT')^{A \to B} \otimes \I^R]((U^A \otimes I^R) \circ \gamma^{AR})
\right\}^2
\right]
\right)^i \\
& = &
\left(
\Tr \left[
\left\{
[(\tcT')^{A_1 \to B_1} \otimes \I^{R_1}]
((U^{A_1} \otimes I^{R_1}) \circ \gamma^{A_1 R_1})
\right\} 
\right. \right.\\
&   &   
~~~~ \left.\left.
\otimes
\left\{
[(\tcT')^{A_2 \to B_2} \otimes \I^{R_2}]
((U^{A_2} \otimes I^{R_2}) \circ \gamma^{A_2 R_2})
\right\} 
(F^{B_1 B_2} \otimes F^{R_1 R_2})
\right]
\right)^i \\
& = &
\bigg(
\Tr \bigg[
(I^{A_1 A_2} \otimes F^{R_1 R_2})
(\gamma^{A_1 R_1} \otimes \gamma^{A_2 R_2}) \\
&  &
~~~~~~~~
\left\{
\left(
((U^\dag)^{A_1} \otimes (U^\dag)^{A_2}) \circ
\left(
[((\tcT')^\dag)^{B_1 \to A_1} \otimes ((\tcT')^\dag)^{B_2 \to A_2}]
(F^{B_1 B_2})
\right)
\right) \otimes I^{R_1 R_2}
\right\}
\bigg]
\bigg)^i \\
& = &
\Tr \bigg[
\bigotimes\limits_{j=1}^i
\bigg(
(I^{A_1(j) A_2(j)} \otimes F^{R_1(j) R_2(j)})
(\gamma^{A_1(j) R_1(j)} \otimes \gamma^{A_2(j) R_2(j)}) \\
&  &
~~~~~~~~~~~~
\left\{
((U^\dag)^{A_1(j)} \otimes (U^\dag)^{A_2(j)}) \circ
\left(
\left[((\tcT')^\dag)^{B_1(j) \to A_1(j)} \otimes 
 ((\tcT')^\dag)^{B_2(j) \to A_2(j)}\right]
(F^{B_1(j) B_2(j)})
\right) \right.\\
&   &
~~~~~~~~~~~~~~~~~~~~~ \left.
{} \otimes I^{R_1(j) R_2(j)}
\right\}
\bigg)
\bigg] \\
& = &
\Tr \bigg[
\left\{
\bigotimes\limits_{j=1}^i
\left(
(I^{A_1(j) A_2(j)} \otimes F^{R_1(j) R_2(j)})
(\gamma^{A_1(j) R_1(j)} \otimes \gamma^{A_2(j) R_2(j)})
\right)
\right\} \\
&  &
~~~~~~~
\left\{
\bigotimes\limits_{j=1}^i
\left(
((U^\dag)^{A_1(j)} \otimes (U^\dag)^{A_2(j)}) \circ
\left[((\tcT')^\dag)^{B_1(j) \to A_1(j)} \otimes 
 ((\tcT')^\dag)^{B_2(j) \to A_2(j)}\right]
(F^{B_1(j) B_2(j)})
\right) \right.\\
&   &
~~~~~~~~~~~~~~~~~~~~~\left.
{} \otimes 
\bigotimes\limits_{j=1}^i
I^{R_1(j) R_2(j)}
\right\}
\bigg] \\
& = &
\Tr \bigg[
\Tr_{\bigotimes\limits_{j=1}^i (R_1(j) R_2(j))}
\left\{
\bigotimes\limits_{j=1}^i
\left(
(I^{A_1(j) A_2(j)} \otimes F^{R_1(j) R_2(j)})
(\gamma^{A_1(j) R_1(j)} \otimes \gamma^{A_2(j) R_2(j)})
\right)
\right\}
\cdot \\
&  &
~~~
\bigotimes\limits_{j=1}^i
\left\{
\left((U^\dag)^{A_1(j)} \otimes (U^\dag)^{A_2(j)}\right) \circ
\left(
\left[((\tcT')^\dag)^{B_1(j) \to A_1(j)} \otimes 
 ((\tcT')^\dag)^{B_2(j) \to A_2(j)}\right]
(F^{B_1(j) B_2(j)})
\right)
\right\}
\bigg] \\
& = &
\Tr \bigg[
\left\{
\Tr_{R_1 R_2}
\left(
(I^{A_1 A_2} \otimes F^{R_1 R_2})
(\gamma^{A_1 R_1} \otimes \gamma^{A_2 R_2})
\right)
\right\}^{\otimes i} \\
&  &
~~~~~~~
\left\{
\left((U^\dag)^{A_1} \otimes (U^\dag)^{A_2}\right) \circ
\left(
\left[((\tcT')^\dag)^{B_1 \to A_1} \otimes 
 ((\tcT')^\dag)^{B_2 \to A_2}\right]
(F^{B_1 B_2})
\right)
\right\}^{\otimes i}
\bigg] \\
& = &
\Tr \bigg[
\{
(\gamma')^{A_1 A_2}
\}^{\otimes i} \cdot\\
&  &
~~~~~~~
\left\{
\left(U^\dag)^{A_1} \otimes (U^\dag)^{A_2}\right)^{\otimes i} \circ
\left[((\tcT')^\dag)^{B_1 \to A_1} \otimes 
 ((\tcT')^\dag)^{B_2 \to A_2}\right]
(F^{B_1 B_2})
\right\}^{\otimes i}
\bigg]
\end{eqnarray*}
This completes the proof of the proposition. 
\end{proof}

\begin{lemma}
\label{lem:MomentTPEHaar}
Let $i$ be a positive integer.
Consider a $(|A|, s, \lambda, 4i)$-qTPE for some positive integer $s$
and $\lambda \geq 0$. Then,
\[
\left|
\underset{U \sim \TPE}{\E}[(g(U))^{2i}] - 
\underset{U \sim \Haar}{\E}[(g(U))^{2i}] 
\right| \leq
(2 |A|^3 |B|^2)^i \cdot \lambda \cdot 
2^{i (1+\delta) 
\mathbb{H}_{\max}^\epsilon(B)_\omega -i H_2^\epsilon(A|R)_\rho}.
\]
\end{lemma}

 \begin{proof}
We evaluate $[g(U)]^{2i}$ using Proposition~\ref{prop:g_exp_2i} as:
\begin{eqnarray*}
\lefteqn{
 (g(U)^{2})^{i}
 } \\
 & = &
    \Tr \bigg[
 \{
 (\gamma')^{A_1 A_2}
 \}^{\otimes i} \cdot\\
 &  &
 ~~~~~~~
 \left\{
 \left(U^\dag)^{A_1} \otimes (U^\dag)^{A_2}\right)^{\otimes i} \circ
 \left[((\tcT')^\dag)^{B_1 \to A_1} \otimes 
  ((\tcT')^\dag)^{B_2 \to A_2}\right]
 (F^{B_1 B_2})
 \right\}^{\otimes i}
 \bigg].
 \end{eqnarray*}
Hence by taking expectation with respect to Haar and TPE and using the 
linearity of expectation we get,
\begin{eqnarray*}
\lefteqn{
\left|
\underset{U \sim \TPE}{\E}[(g(U))^{2i}] - 
\underset{U \sim \Haar}{\E}[(g(U))^{2i}] 
\right|
} \\
& = &
\bigg|
\underset{U \sim \TPE}{\E}\bigg[
\Tr \left[
\left\{
(\gamma')^{A_1 A_2}
\right\}^{\otimes i} \cdot \right.\\
&  &
~~~~~~~~~~~~~~~~\left.
\left\{(U^\dag)^{A_1} \otimes (U^\dag)^{A_2}\right\}^{\otimes i} \circ
\left\{
\left[((\tcT')^\dag)^{B_1 \to A_1} \otimes 
 ((\tcT')^\dag)^{B_2 \to A_2}\right]
(F^{B_1 B_2})
\right\}^{\otimes i}
\right]
\bigg] \\
&   &
~~~~~
{}-
\underset{U \sim \Haar}{\E}\bigg[
\Tr \left[
\left\{
(\gamma')^{A_1 A_2}
\right\}^{\otimes i} \cdot \right.\\
&  &
~~~~~~~~~~~~~~~~~~~~\left.
\left\{(U^\dag)^{A_1} \otimes (U^\dag)^{A_2}\right\}^{\otimes i} \circ
\left\{
\left[((\tcT')^\dag)^{B_1 \to A_1} \otimes 
 ((\tcT')^\dag)^{B_2 \to A_2}\right]
(F^{B_1 B_2})
\right\}^{\otimes i}
\right]
\bigg] 
\bigg| \\
& = &
\bigg|
\Tr \bigg[
\left\{(\gamma')^{A_1 A_2}\right\}^{\otimes i} \cdot \\
&  &
~~~~
\left\{
\underset{U \sim \TPE}{\E}\left[
\left((U^\dag)^{A_1} \otimes (U^\dag)^{A_2}\right)^{\otimes i} \circ
\left(
\left[((\tcT')^\dag)^{B_1 \to A_1} \otimes 
 ((\tcT')^\dag)^{B_2 \to A_2}\right]
(F^{B_1 B_2})
\right)^{\otimes i}
\right] \right.\\
&   &
~~~~~~~\left.
{}-
\underset{U \sim \Haar}{\E}\left[
\left((U^\dag)^{A_1} \otimes (U^\dag)^{A_2}\right)^{\otimes i} \circ
\left(
\left[((\tcT')^\dag)^{B_1 \to A_1} \otimes 
 ((\tcT')^\dag)^{B_2 \to A_2}\right]
(F^{B_1 B_2})
\right)^{\otimes i}
\right]\right\}
\bigg] 
\bigg| \\
& \leq &
\left\lVert ((\gamma')^{A_1 A_2})^{\otimes i} \right\rVert_1 \cdot\\
&   &
~~~~~
{} 
\left\lVert
\underset{U \sim \TPE}{\E}\left[
\left\{(U^\dag)^{A_1} \otimes (U^\dag)^{A_2}\right\}^{\otimes i} \circ
\left\{
\left[((\tcT')^\dag)^{B_1 \to A_1} \otimes 
 ((\tcT')^\dag)^{B_2 \to A_2}\right]
(F^{B_1 B_2})
\right\}^{\otimes i}
\right] \right.\\
&   &
~~~~~~~\left.
{}-
\underset{U \sim \Haar}{\E}\left[
\left\{(U^\dag)^{A_1} \otimes (U^\dag)^{A_2}\right\}^{\otimes i} \circ
\left\{
\left[((\tcT')^\dag)^{B_1 \to A_1} \otimes 
 ((\tcT')^\dag)^{B_2 \to A_2}\right]
(F^{B_1 B_2})
\right\}^{\otimes i}
\right]
\right\rVert_\infty \\
& \leq &
\left\lVert (\gamma')^{A_1 A_2} \right\rVert_1^i \cdot\\
&   &
~~~~~
{}
\left\lVert
\underset{U \sim \TPE}{\E}\left[
\left\{(U^\dag)^{A_1} \otimes (U^\dag)^{A_2}\right\}^{\otimes i} \circ
\left\{
\left[((\tcT')^\dag)^{B_1 \to A_1} \otimes 
 ((\tcT')^\dag)^{B_2 \to A_2}\right]
(F^{B_1 B_2})
\right\}^{\otimes i}
\right] \right.\\
&   &
~~~~~~\left.
{}-
\underset{U \sim \Haar}{\E}\left[
\left\{(U^\dag)^{A_1} \otimes (U^\dag)^{A_2}\right\}^{\otimes i} \circ
\left\{
\left[((\tcT')^\dag)^{B_1 \to A_1} \otimes 
 ((\tcT')^\dag)^{B_2 \to A_2}\right]
(F^{B_1 B_2})
\right\}^{\otimes i}
\right]
\right\rVert_2 \\
& \overset{a}{\leq} &
\left(|A| \lVert \gamma^{A R} \rVert_2^2\right)^i \cdot \lambda \cdot
\left\lVert
\left\{
\left[((\tcT')^\dag)^{B_1 \to A_1} \otimes 
 ((\tcT')^\dag)^{B_2 \to A_2}\right]
(F^{B_1 B_2})
\right\}^{\otimes i}
\right\rVert_2 \\
& \overset{b}{\leq} &
\left(2 |A| \lVert (\trho')^{A R} \rVert_2^2\right)^i \cdot \lambda \cdot
\left\lVert
\left[((\tcT')^\dag)^{B_1 \to A_1} \otimes 
 ((\tcT')^\dag)^{B_2 \to A_2}\right]
(F^{B_1 B_2})
\right\rVert_2^i \\
& \overset{c}{\leq} &
\left(2 |A| \lVert (\trho')^{A R} \rVert_2^2\right)^i \cdot \lambda \cdot
\lVert W_{\tcT'} \rVert_2^{4i} 
\;\overset{d}{\leq}\;
(2 |A|)^i 2^{-i H_2^\epsilon(A|R)_\rho} \cdot \lambda \cdot
2^{i (1+\delta) \mathbb{H}_{\max}^\epsilon(B)_\omega}
(|A|^2 |B|^2)^i \\
& = &
(2 |A|^3 |B|^2)^i \cdot \lambda \cdot 
2^{i (1+\delta) 
\mathbb{H}_{\max}^\epsilon(B)_\omega -i H_2^\epsilon(A|R)_\rho},
\end{eqnarray*}
where
\begin{itemize}
\item[(a)]
follows from Fact~\ref{fact:swappartialtrace} and 
Definition~\ref{def:qTPE},

\item[(b)]
follows from Equation~\ref{eq:elltwogamma}.

\item[(c)]
follows from Lemma~\ref{lem:TtwoSwap},

\item[(d)]
follows from Note~\ref{note:W_tilde(T)} and 
Equations~\ref{eq:W_tilde_tau_prime_two} and \ref{eq:tilde_rho}.
\end{itemize}
This completes the proof of the lemma.
\end{proof}

Now we upper bound the centralised $(2m)$-th moment of 
${(g(U))}^2$ under the approximate unitary design.
\begin{lemma}
\label{lem:CentralMomentTPE}
Let $m$ be a positive integer.
Suppose $\mu := \E_{U \sim \Haar}[g(U)]$ and
$H_2^\epsilon(A|R)_\rho < 0$.
Consider a $(|A|, s, \lambda, 4m)$-qTPE for some positive integer $s$
and 
\[
\lambda^{1/m} \leq
\left\{
\hspace*{-2mm}
\begin{array}{l l}
2^4 m \cdot |A|^{-7} |B|^{-4} \mu^2 \cdot
2^{
-(1+\delta) \mathbb{H}_{\max}^\epsilon(B)_\omega + 
H_2^{\epsilon}(A|R)_{\rho}, 
} 
&
\hspace*{-2mm}
m <
\frac{9}{64} |A| \mu^2 \cdot
2^{
-(1+\delta)\mathbb{H}_{\max}^\epsilon(B)_\omega +
H_2^{\epsilon}(A|R)_\rho - 4
} \\
2^{10} m^2 \cdot |A|^{-8} |B|^{-4} 
&
\hspace*{-2mm}
\mbox{otherwise}
\end{array}
\right..
\]
Then,
\begin{eqnarray*}
\lefteqn{\underset{U \sim \TPE}{\E}[((g(U))^2 - \mu^2)^{2m}]} \\
& \leq &
\left\{
\begin{array}{l l}
7 (9 m \mu^2
2^{
(1+\delta)\mathbb{H}_{\max}^\epsilon(B)_\omega -
H_2^{\epsilon}(A|R)_\rho + 4
} |A|^{-1}
)^m 
&
m <
\frac{9}{64} |A| \mu^2 \cdot
2^{
-(1+\delta)\mathbb{H}_{\max}^\epsilon(B)_\omega +
H_2^{\epsilon}(A|R)_\rho - 4
} \\
7 (64 m^2 
2^{
2 (1+\delta)\mathbb{H}_{\max}^\epsilon(B)_\omega -
2 H_2^{\epsilon}(A|R)_\rho + 8
} |A|^{-2}
)^m 
&
\mbox{otherwise}
\end{array}
\right..
\end{eqnarray*}
\end{lemma}
\begin{proof}
From Lemma~\ref{lem:MomentTPEHaar}, we get
\begin{align*}
\lefteqn{
\left|
\underset{U \sim \TPE}{\E}[((g(U))^2 - \mu^2)^{2m}] - 
\underset{U \sim \Haar}{\E}[((g(U))^2 - \mu^2)^{2m}]
\right|
} \\
& = 
\left|
\sum_{i=0}^{2m} {{2m}\choose{i}} (-\mu^2)^{2m-i}
(
\underset{U \sim \TPE}{\E}[(g(U))^{2i}] - 
\underset{U \sim \Haar}{\E}[(g(U))^{2i}] 
)
\right| \\
& \leq 
\sum_{i=0}^{2m} {{2m}\choose{i}} (\mu^2)^{2m-i}
\left|
\underset{U \sim \TPE}{\E}[(g(U))^{2i}] - 
\underset{U \sim \Haar}{\E}[(g(U))^{2i}] 
\right| \\
& \leq 
\lambda
\sum_{i=0}^{2m} {{2m}\choose{i}} (\mu^2)^{2m-i} 
(2 |A|^3 |B|^2)^i \cdot
2^{i (1+\delta) 
\mathbb{H}_{\max}^\epsilon(B)_\omega -i H_2^\epsilon(A|R)_\rho} \\
& =
\lambda
(
\mu^2 +
(2 |A|^3 |B|^2) \cdot
2^{(1+\delta) 
\mathbb{H}_{\max}^\epsilon(B)_\omega - H_2^\epsilon(A|R)_\rho} 
)^{2m} \\
&\overset{a}{\leq}
\lambda
(
(3 |A|^3 |B|^2) \cdot
2^{(1+\delta) 
\mathbb{H}_{\max}^\epsilon(B)_\omega - H_2^\epsilon(A|R)_\rho} 
)^{2m}.
\end{align*}
where (a) follows from the observation that $\mu^2 \leq \E_U[{g(U)}^2]$ (from Jensen's inequality for the function $x^2$) and using the upper bound on $g(U)$ from Lemma~\ref{lem:maxg} and the fact that $|A|, |B| \geq 1$.\\ 
Using Lemma~\ref{lem:Haar_moments_g} and the above inequality, we get
\begin{eqnarray*}
\lefteqn{
\underset{U \sim \TPE}{\E}\left[\left((g(U))^2 - \mu^2\right)^{2m}\right] 
} \\
& \leq &
\lambda
\left(
(3 |A|^3 |B|^2)^2 \cdot
2^{2 (1+\delta) \mathbb{H}_{\max}^\epsilon(B)_\omega - 
	2 H_2^\epsilon(A|R)_\rho
} 
\right)^{m} \\
&  &
{} +
\left\{
\begin{array}{l l}
6 \left(9 m \mu^2
2^{
(1+\delta)\mathbb{H}_{\max}^\epsilon(B)_\omega -
H_2^{\epsilon}(A|R)_\rho + 4
} |A|^{-1}
\right)^m 
&
;m <
\frac{9}{64} |A| \mu^2 \cdot
2^{
-(1+\delta)\mathbb{H}_{\max}^\epsilon(B)_\omega +
H_2^{\epsilon}(A|R)_\rho - 4
} \\
6 \left(64 m^2 
2^{
2 (1+\delta)\mathbb{H}_{\max}^\epsilon(B)_\omega -
2 H_2^{\epsilon}(A|R)_\rho + 8
} |A|^{-2}
\right)^m 
&
\mbox{; otherwise}
\end{array}
\right. \\
& \leq &
\left\{
\begin{array}{l l}
7 \left(9 m \mu^2
2^{
(1+\delta)\mathbb{H}_{\max}^\epsilon(B)_\omega -
H_2^{\epsilon}(A|R)_\rho + 4
} |A|^{-1}
\right)^m 
&
;m <
\frac{9}{64} |A| \mu^2 \cdot
2^{
-(1+\delta)\mathbb{H}_{\max}^\epsilon(B)_\omega +
H_2^{\epsilon}(A|R)_\rho - 4
} \\
7 \left(64 m^2 
2^{
2 (1+\delta)\mathbb{H}_{\max}^\epsilon(B)_\omega -
2 H_2^{\epsilon}(A|R)_\rho + 8
} |A|^{-2}
\right)^m 
&
\mbox{; otherwise}
\end{array}
\right..
\end{eqnarray*}
This completes the proof of the lemma.
\end{proof}

\begin{proof}[Proof of Theorem~\ref{thm:main}]
Define the positive integer
$
m :=
\lfloor
|A| \kappa^2
2^{
-(1+\delta)\mathbb{H}_{\max}^\epsilon(B)_\omega +
H_2^{\epsilon}(A|R)_\rho - 8
}
\rfloor.
$
Using Lemma~\ref{lem:CentralMomentTPE}, we get
\begin{align*}
\lefteqn{
\prob_{U \sim \TPE}[g(U) - \mu  >  \kappa] 
\leq 
\prob_{U \sim \TPE}[(g(U))^2  > (\mu + \kappa)^2] 
} \\
& \leq 
\left\{
\begin{array}{l l}
\prob_{U \sim \TPE}[(g(U))^2 - \mu^2 > 2 \mu \kappa]
& ; \kappa \leq \mu \\
\prob_{U \sim \TPE}[(g(U))^2 - \mu^2 > \kappa^2]
& \mbox{; otherwise} \\
\end{array}
\right. \\
& \leq 
\left\{
\begin{array}{l l}
\prob_{U \sim \TPE}[((g(U))^2 - \mu^2)^{2m} > (2 \mu \kappa)^{2m}]
& ;\kappa \leq \mu \\
\prob_{U \sim \TPE}[((g(U))^2 - \mu^2)^{2m} > \kappa^{4m}]
& \mbox{; otherwise} \\
\end{array}
\right. \\
& \leq 
\left\{
\begin{array}{l l}
\frac{\E_{U \sim \TPE}\left[\left((g(U))^2 - 
\mu^2\right)^{2m}\right]}{(2 \mu \kappa)^{2m}}
& ;\kappa \leq \mu \\
\frac{\E_{U \sim \TPE}\left[\left((g(U))^2 - 
\mu^2\right)^{2m}\right]}{\kappa^{4m}}
& \mbox{; otherwise} \\
\end{array}
\right. \\
& \leq
\left\{
\begin{array}{l l}
\frac{
7 \left(9 m \mu^2
2^{
(1+\delta)\mathbb{H}_{\max}^\epsilon(B)_\omega -
H_2^{\epsilon}(A|R)_\rho + 4
} |A|^{-1}
\right)^m
}{
(2\mu \kappa)^{2m}
} 
& ; \kappa \leq \mu,
m <
\frac{9}{64} |A| \mu^2 \cdot
2^{
-(1+\delta)\mathbb{H}_{\max}^\epsilon(B)_\omega +
H_2^{\epsilon}(A|R)_\rho - 4
} \\
\frac{
7 \left(8 m 
2^{
(1+\delta)\mathbb{H}_{\max}^\epsilon(B)_\omega -
H_2^{\epsilon}(A|R)_\rho + 4
} |A|^{-1}
\right)^{2m}
}{
\kappa^{4m}
} 
& \mbox{; otherwise} \\
\end{array}
\right. \\
& \leq
\left\{
\begin{array}{l l}
7 \left(3 m \kappa^{-2}
2^{
(1+\delta)\mathbb{H}_{\max}^\epsilon(B)_\omega -
H_2^{\epsilon}(A|R)_\rho + 4
} |A|^{-1}
\right)^m
& ;\kappa \leq \mu \\
7 \left(8 m \kappa^{-2} 
2^{
(1+\delta)\mathbb{H}_{\max}^\epsilon(B)_\omega -
H_2^{\epsilon}(A|R)_\rho + 4
} |A|^{-1}
\right)^{2m}
& \mbox{; otherwise} \\
\end{array}
\right. \\
& \leq
7 \cdot 2^{-a \kappa^2},
\end{align*}
where
$
a := 
|A| \cdot
2^{
-(1+\delta)\mathbb{H}_{\max}^\epsilon(B)_\omega +
H_2^{\epsilon}(A|R)_\rho - 9
}.
$\\
Since $0<\delta <1$, ${\mathbb{H}_{max}}^\epsilon(B)_\omega \leq 
\log |B|$ and $H_2^\epsilon(A|R)_\rho \geq -\log |A|$, therefore\\ 
$
(m \cdot |A|^{-8} |B|^{-6} \cdot \mu^2)^m \leq
(
2^4 \cdot m \cdot |A|^{-7} |B|^{-4} \mu^2 \cdot 
2^{
-(1+\delta)\mathbb{H}_{\max}^\epsilon(B)_\omega +
H_2^{\epsilon}(A|R)_\rho - 4
}
)^m
$\\
And thus, the above probability analysis requires us to use a 
$(|A|, s, \lambda', 4m)$-qTPE with
\[
\lambda' \leq 
\left(m \cdot |A|^{-8} |B|^{-6} \cdot \mu^2\right)^m,
\]
if $\kappa < \mu$ and
\begin{equation}
\label{eq:lambda}
\lambda' \leq 
\left(
2^{10} \cdot m^2 \cdot |A|^{-8} |B|^{-4}
\right)^m
\end{equation}
otherwise.
Define
$
t := 
|A| \kappa^2
2^{
H_2^{\epsilon}(A|R)_\rho
-\mathbb{H}_{\max}^\epsilon(B)_\omega - 6
}.
$
Then $t \geq 4m$ and a $(|A|, s, \lambda, t)$-qTPE suffices for the
derandomization where
$
\lambda :=
(|A|^{-8} |B|^{-6} \cdot \mu^2)^t.
$\\
Note that this value of $\lambda$ is smaller than that obtained in 
Equation~\ref{eq:lambda} and hence a unitary chosen at random from 
this qTPE shall also achieve exponential concentration. 
For a qTPE to achieve this smaller $\lambda$ requires slightly 
more number of iterations of a qTPE with constant singular value gap,
as compared to a qTPE with $\lambda'$ in Equation~\ref{eq:lambda}.
Combined with Lemma~\ref{lem:f_to_g}, we finally get
\[
\prob_{U \sim \TPE}\left[
f(U) >
2^{
-\frac{1}{2} H_2^\epsilon(A|R)_\rho 
-\frac{1}{2} \mathbb{H}_2^{\epsilon,\delta}(A'|B)_\omega + 1
} + 14\sqrt{\epsilon} + 2 \kappa
\right] \leq 
7 \cdot 2^{-a \kappa^2}.
\]
Combining Proposition~\ref{prop:decoupling_general_Haar} with the 
above analysis finishes the proof of Theorem~\ref{thm:main}.
\end{proof}

\section{The asymptotic iid case} 
\label{sec:iid}
In this section we first show that the smooth one-shot entropies 
defined in Section~\ref{subsec:entropies} approach their natural Shannon
entropic analogues in the asymptotic iid limit. We then use those
results to prove the asymptotic iid version of the main theorem. The
asymptotic iid version will be stated in terms of Shannon entropies.

\subsection{Asymptotic smoothing of $\mathbb{H}_{\max}^\epsilon$ and
$\mathbb{H}_2^{\epsilon,\delta}$}
In this section, we use the properties of typical sequences and subspaces 
to find an upper bound on $\mathbb{H}_{\max}^\epsilon$ and a lower bound on
$\mathbb{H}_2^{\epsilon,\delta}$ in the asymptotic limit 
of many iid
copies of the underlying quantum states. The bounds obtained will be
the Shannon entropic quantities that one would expect.
We first prove a few essential lemmas.
\begin{lemma}
\label{lem:projection1}
Suppose we have a density matrix $\omega$ on the system $AB$.
Let $\ket{w_j}^{AB}$, $j \in [|A| |B|]$ be the eigenvectors of
$\omega^{AB}$ with eigenvalues $q_j$. 
For $j \in [|A| |B|]$, define 
$\theta_j^B := \Tr_A [\ket{w_j}^{AB}\bra{w_j}]$.
Let $p_j:=\{ p_j(i)\}_{\{ i \in [|B|]\}}$, 
$j \in [|A||B|]$ be the probability
distribution on $[|B|]$ obtained by measuring 
$\theta_j$ in the eigenbasis of $\omega^{B}$. 
Let $0 < \epsilon, \delta < 1/3$. 
Define 
$
q_{\mathrm{min}} := 
2^{-\mathbb{H}_{\max}^{\epsilon/2}(AB)_\omega},
$ 
$
p_{\mathrm{min}} := 
\min_{j \in [|A| |B|]}
2^{-\mathbb{H}_{\max}^{\epsilon/2}([|B|])_{p_j}},
$ 
Let 
$
n \geq 2^5 q^{-1}_{\mathrm{min}} p^{-1}_{\mathrm{min}} \delta^{-2} 
\log(|A| |B|/\epsilon).
$
Consider the $n$-fold tensor power 
$\omega^{A^n B^n} := (\omega^{AB})^{\otimes n}$.
Let $\tau$ be a strongly $\delta$-typical type of
an eigenvector sequence of $\omega^{A^n B^n}$. Let 
$(\ket{v_1} \otimes \cdots \otimes \ket{v_n})^{A^n B^n}$
be an eigenvector sequence of type $\tau$. 
Let 
$
\sigma^{B^n} := 
\Pi^{B^n}_{\omega, 3\delta}
\omega^{B^n}
\Pi^{B^n}_{\omega, 3\delta}.
$
%Let $\Pi^{B^n}_{\sigma}$ be the orthogonal projection onto the support
%of $\sigma^{B^n}$. 
Then,
\[
\Tr [
(I^{A^n} \otimes \Pi^{B^n}_{\omega, 3\delta}) 
(\ket{v_1} \cdots \ket{v_n}
\bra{v_1} \cdots \bra{v_n})^{A^n B^n}
] \geq 1 - \epsilon.
\]
\end{lemma}
\begin{proof}
Since $\tau$ is a strongly $\delta$-typical type, the number of
occurrences $n_j$ of each $\ket{w_j}^{AB}$ in the sequence
$\ket{v_1}^{AB}, \ldots, \ket{v_n}^{AB}$ is $n q_j (1 \pm \delta)$. After a
suitable rearranging, we can write
\[
(\ket{v_1} \otimes \cdots \otimes \ket{v_n})^{A^n B^n} =
(\ket{w_1}^{\otimes n_1} \otimes \cdots \otimes 
\ket{w_{|A||B|}}^{\otimes n_{|A||B|}})^{A^n B^n}.
\]
To prove the lemma, it suffices to show that
\begin{align} \label{eq:reduced_theta}
\Tr [
\Pi^{B^n}_{\omega, 3\delta}
(
(\theta_1^B)^{\otimes n_1} \otimes \cdots \otimes 
(\theta_{|A||B|}^B)^{\otimes n_{|A||B|}}
)
] \geq 1 - \epsilon.
\end{align}
Let $\ket{x_1}^B, \ldots, \ket{x_{|B|}}^B$ be the eigenbasis of 
$\omega^B$ with eigenvalues $r_1, \ldots, r_{|B|}$. Observe that we
have the operator equality
$
\sum_{j=1}^{|A||B|} q_j \theta_j^B = \omega^B.
$
Now consider the matrices $\theta_j^B$ in the basis
$\ket{x_1}^B, \ldots, \ket{x_{|B|}}^B$. Thus for any $i \in [|B|]$,
$\sum_{j=1}^{|A||B|} q_j p_j(i) = r_i$.\\
Let $\Pi^{B^{n_j}}_{j}$ be the projector onto the eigenvectors of
$(\omega^B)^{\otimes n_j}$ that are strongly $\delta$-typical according
to $p_j$ (and not with respect to the eigenvalues $\{r_i\}_{i \in [|B|]}$ of $\omega^B$). Note that this is not the same as projecting onto the eigenvectors of ${(\theta_j^{B})}^{\otimes n_j}$ and also ${(\theta_j^{B})}^{\otimes n_j}$ does not commute with $\omega_B^{\otimes n_j}$.\\ 
Since $\{p_j(i)\}_{i \in [|B|]}$ is the probability distribution on $B$ obtained by measuring $\theta_j^B$ in the eigenbasis of $\omega^B$ for all $j \in [|A||B|]$ and $\Pi_j^{B^{n_j}}$ is the typical projector with respect to the probability distribution $\{p_j\}$ (not on the eigenbasis of $\omega^{B^{n_j}}$), therefore by Fact~\ref{fact:quantum_AEP} we have,
$
\Tr [\Pi^{B^{n_j}}_{j} (\theta_j^B)^{\otimes n_j}] \geq
1 - \frac{\epsilon}{|A||B|}.
$
Thus,
\begin{align} \label{eq:typical_pj}
\Tr [
(
\Pi^{B^{n_1}}_{1} \otimes \cdots \otimes
\Pi^{B^{n_{|A||B|}}}_{|A||B|}
)
(
(\theta_1^B)^{\otimes n_1} \otimes \cdots \otimes 
(\theta_{|A||B|}^B)^{\otimes n_{|A||B|}}
)
] &\geq
\prod_{j=1}^{|A||B|} \left(1 - \frac{\epsilon}{|A||B|} \right) \nonumber\\ 
& \geq 1-\epsilon\;.
\end{align}

Fix an eigenvector $\left( \textrm{ of }\omega^{B^{n_1}} \otimes \ldots \otimes \omega^{B^{n_{|A||B|}}}\right)$ in the support of 
$
\Pi^{B^{n_1}}_{1} \otimes \cdots \otimes
\Pi^{B^{n_{|A||B|}}}_{|A||B|}
;
$
the eigenvector can be viewed as a sequence of length $n_1$ (typical with respect to probability $p_1$), $\ldots, n_{|A||B|}$ (typical with respect to probability distribution $p_{|A||B|}$) such that $\sum_{j=1}^{|A||B|}n_j=n$.
Then the number of occurrences of $\ket{x_i}$ in the sequence is
\[
\sum_{j=1}^{|A||B|} n q_j (1 \pm \delta) p_j(i) (1 \pm \delta) = 
r_i (1 \pm 3\delta).
\]
This shows that the eigenvector is strongly $(3\delta)$-typical
for the state $\omega^{B^n}$. Also, since eigenvectors corresponding to distinct eigenvalues are orthogonal, therefore $\mathop{\bigotimes} \limits_{j=1}^{|A||B|}\Pi_j^{{B}^{n_j}}$ commutes with $\Pi_{\omega, 3\delta}^{B^n}$. In other words,
\begin{align} \label{eq:op_ineq}
\Pi^{B^{n_1}}_{1} \otimes \cdots \otimes
\Pi^{B^{n_{|A||B|}}}_{|A||B|} \leq
\Pi^{B^n}_{\omega, 3\delta}
%= \Pi^{B^n}_\sigma.
\end{align}
Combining Equation~\ref{eq:typical_pj} with Equation~\ref{eq:op_ineq} and then Equation~\ref{eq:reduced_theta} completes the proof of the lemma.
\end{proof}
\textbf{Note: }The non trivial take away message of the above lemma is that a typical vector of $\omega^{AB}$ still retains the typicality property when projected onto the marginal $\omega^B$ (in the sense that any typical eigenvector of ${(\omega^{AB})}^{\otimes n}$ has a large projection onto the subspace spanned by $I^{A^n} \otimes \Pi_{\omega, 3\delta}^{B^n}$; even though the eigenvectors of $\omega^{AB}$ and $\omega^B$ are not related), but the probability distribution is not the one obtained by eigenvalues of $\omega^B$. Although the above statement might look believable in analogy with the classical typical sequences, it is still non-trivial to prove as is suggested by the proof above. The proof centrally uses the property of strongly typical types.
\begin{lemma}
\label{lem:projection2}
Consider the setting of Lemma~\ref{lem:projection1}.
Let $V_\tau \leq A^n B^n$ denote the type subspace corresponding to
type $\tau$. Then there is a subspace $\hat{V}_\tau \leq V_\tau$,
$|\hat{V}_\tau| \geq (1 - \sqrt{\epsilon}) |V_\tau|$ such that
for every vector $\ket{v} \in \hat{V}_\tau$,
$
\lVert (I^{A^n} \otimes \Pi^{B^n}_\sigma) \ket{v} \rVert_2^2 \geq
1 - \sqrt{\epsilon}.
$
\end{lemma}
\begin{proof}
We invoke Fact~\ref{fact:chordal} with $A := V_\tau$ and 
$B := I^{A^n} \otimes \Pi^{B^n}_\sigma$ in order to prove this lemma.
Take the basis for $A$ provided by Fact~\ref{fact:chordal}.
Call it $\{\ket{a}\}_a$.
Observe that the vectors $\Pi_B \ket{a}$ are pairwise orthogonal
(some of them may be the zero vector).
From Lemma~\ref{lem:projection1}, we know that
$\Tr [\Pi_B \frac{\Pi_A}{|A|}] \geq 1 - \epsilon$.
By Markov's inequality, there is a subset $S$ of the basis vectors of
$A$, $|S| \geq (1-\sqrt{\epsilon}) |A|$ such that for all $a \in S$, 
$
\lVert \Pi_B \ket{a} \rVert_2^2 =\Tr [\Pi_B \ket{a}\bra{a}] 
\geq 1 - \sqrt{\epsilon}.
$
Define the subspace $\hat{A} := \mathrm{span}_{a \in S} \ket{a}$.
From the above observation, for any vector $\ket{v} \in \hat{A}$,
$\lVert \Pi_B \ket{v} \rVert_2^2 \geq 1 - \sqrt{\epsilon}$.
This subspace $\hat{A}$ serves as the subspace $\hat{V}_\tau$ required
by the lemma.
\end{proof}
\begin{lemma}
\label{lem:projection3}
Let $0 < \epsilon < 1$.
Let $\ket{v_1}, \ldots, \ket{v_t}$ be orthonormal vectors lying in
a Hilbert space $\cH$. Suppose there is a subspace $B \leq \cH$ with the
property that $\lVert \Pi_B \ket{v_i} \rVert_2^2 \geq 1 - \epsilon$ for
all $i \in [t]$. Let $\ket{v}$ be a unit vector lying in the span of
the vectors $\ket{v_i}$. Then,
$
\lVert \Pi_B \ket{v} \rVert_2^2 \geq 1 - 8 t \sqrt{\epsilon}.
$
\end{lemma}
\begin{proof}
Let $\ket{v} = \sum_{i=1}^t \alpha_i \ket{v_i}$ where 
$\sum_{i=1}^t |\alpha_i|^2 = 1$. 
Define the column $t$-tuple $\alpha := (\alpha_1, \ldots, \alpha_t)^T$,
and the $t \times t$-matrix $M$ with 
$M_{ij} := \bra{v_i} \Pi_B \ket{v_j}$. Note that $M$ is Hermitian.
Then,
\[
\lVert \Pi_B \ket{v} \rVert_2^2 =
\sum_{i,j=1}^t \alpha_i^* \alpha_j \bra{v_i} \Pi_B \ket{v_j} =
\alpha^\dag M \alpha.
\]
We have $M_{ii} \geq 1-\epsilon$. 
For $i \neq j$, we use triangle inequality and 
Fact~\ref{fact:gentle} 
to obtain
\[
\lVert \Pi_B \ket{v_i} - \Pi_B \ket{v_j} \rVert_2 \geq
\lVert \ket{v_i} - \ket{v_j} \rVert_2 -
\lVert \Pi_B \ket{v_i} - \ket{v_i} \rVert_2 -
\lVert \Pi_B \ket{v_j} - \ket{v_j} \rVert_2 \geq
\sqrt{2} - 4 \sqrt{\epsilon},
\]
which implies that
$
2 - 8 \sqrt{2\epsilon} \leq
\lVert \Pi_B \ket{v_i} - \Pi_B \ket{v_j} \rVert_2^2 \leq
2 - 2 \Re(M_{ij}),
$
which further implies that 
$\Re(M_{ij}) \leq 4 \sqrt{2\epsilon}$.
Arguing similarly with 
$
\lVert \Pi_B \ket{v_i} - \sqrt{-1} \cdot \Pi_B \ket{v_j} \rVert_2,
$
we conclude that
$\Im(M_{ij}) \leq 4 \sqrt{2\epsilon}$. Thus,
$|M_{ij}| \leq 8 \sqrt{\epsilon}$.
By Gershgorin's theorem, the smallest eigenvalue of $M$ is
larger than 
$
1 - \epsilon - 8 (t-1) \sqrt{\epsilon} \geq
1 - 8 t \sqrt{\epsilon}.
$
So
$
\lVert \Pi_B \ket{v} \rVert_2^2 =
\alpha^\dag M \alpha \geq
1 - 8 t \sqrt{\epsilon},
$
completing the proof of the lemma.
\end{proof}

\begin{proposition}
\label{prop:H_max_iid}
Suppose we have a density matrix $\omega$ on the system $B$.
Let $0 < \epsilon, \delta < 1/3$ and  
define 
$
q_{\mathrm{min}} := 2^{-\mathbb{H}_{\max}^{\epsilon/2}(B)_\omega}.
$
Let 
$
n \geq 4 q_{\mathrm{min}}^{-1} \delta^{-2} \log(|B|/\epsilon).
$
Consider the $n$-fold tensor power 
$\omega^{B^n} := (\omega^{B})^{\otimes n}$.
Then,
$
n (1-\delta) H(B)_\omega \leq
\mathbb{H}_{\max}^{\epsilon}(B^n)_\omega \leq
n (1+\delta) H(B)_\omega.
$
\end{proposition}
\begin{proof}
Consider the eigenvalues of $\omega^{B^n}$ that are not
strongly $\delta$-typical; call them atypical.
By Fact~\ref{fact:quantum_AEP}, the atypical eigenvalues sum to 
less than or equal to $\epsilon$ and the smallest typical eigenvalue
is at least $2^{-n (1+\delta) H(B)_\omega}$. Hence the eigenvalues
less than $2^{-n (1+\delta) H(B)_\omega}$ add up to less than or equal
to $\epsilon$. Moreover, the eigenvalues 
less than $2^{-n (1-\delta) H(B)_\omega}$ add up to more than
$1 - \epsilon > \epsilon$.
This completes the proof the proposition.
\end{proof}

\begin{proposition}
\label{prop:H_2_iid}
Suppose we have a density matrix $\omega$ on the system $AB$.
Let $\ket{w_j}^{AB}$, $j \in [|A| |B|]$ be the eigenvectors of
$\omega^{AB}$ with eigenvalues $q_j$. 
For $j \in [|A| |B|]$, define 
$\theta_j^B := \Tr_A [\ket{w_j}^{AB}\bra{w_j}]$.
Let $p_j := \{p_j(i)\}_{\{ i \in [|B|]\} }$, 
$j \in [|A||B|]$ be the probability
distribution on $[|B|]$ obtained by measuring 
$\theta_j$ in the eigenbasis of $\omega^{B}$. 
Let $0 < \epsilon, \delta < 1/3$. 
Define 
$
q_{\mathrm{min}} := 
2^{-\mathbb{H}_{\max}^{\epsilon/2}(AB)_\omega},
$ 
$
p_{\mathrm{min}} := 
\min_{j \in [|A| |B|]}
2^{-\mathbb{H}_{\max}^{\epsilon/2}([|B|])_{p_j}},
$ 
Let 
$
n := 2^5 q^{-1}_{\mathrm{min}} p^{-1}_{\mathrm{min}} \delta^{-2} 
\log(|A| |B|/\epsilon).
$
Consider the $n$-fold tensor power 
$\omega^{A^n B^n} := (\omega^{AB})^{\otimes n}$.
Let $\epsilon' := 8 (n+|A||B|)^{|A||B|} \epsilon^{1/4}$.
Then,
\begin{eqnarray*}
n H(A|B)_\omega + 32 n \sqrt{\epsilon'} \log |A| +
\log (\epsilon')^{-1}
 \geq 
H_2^{4\sqrt{\epsilon'}}(A^n|B^n)_\omega \geq
(H_2')^{\epsilon',5\delta}(A^n|B^n)_\omega ~~~&\\
\geq
n H(A|B)_\omega - 
n \delta (3 H(AB)_\omega + 7 H(B)_\omega).~~~~~~~~&
\end{eqnarray*}
\end{proposition}
\begin{proof}
For a type $\tau$ of $\omega^{A^n B^n}$, define
$p_\tau := \Tr [\Pi_{V_\tau} \omega^{A^n B^n}]$.
By Fact~\ref{fact:quantum_AEP},
\[
\omega^{A^n B^n} =
\bigoplus_\tau \Pi_{V_\tau} \omega^{A^n B^n} \Pi_{V_\tau} =
\bigoplus_\tau p_\tau \frac{\Pi_{V_\tau}}{|V_\tau|},
\]
where the direct sum is over all types $\tau$.
Now define
$
\eta^{A^n B^n} :=
\bigoplus\limits_{\tau: \mathrm{typical}} 
p_\tau \frac{\Pi_{\hat{V}_\tau}}{|V_\tau|},
$
where the sum is only over strongly $\delta$-typical types $\tau$.
By Fact~\ref{fact:quantum_AEP} and Lemma~\ref{lem:projection2},
\[
\eta^{A^n B^n} \leq \omega^{A^n B^n},
~~~
\lVert \eta^{A^n B^n} - \omega^{A^n B^n} \rVert_1 \leq 2 \sqrt{\epsilon}.
\]

Let 
$
\sigma^{B^n} := 
\Pi^{B^n}_{\omega, 3\delta}
\omega^{B^n}
\Pi^{B^n}_{\omega, 3\delta}.
$
%Let $\Pi^{B^n}_{\sigma}$ be the orthogonal projection onto the support
%of $\sigma^{B^n}$. 
By Fact~\ref{fact:quantum_AEP}, we have
\[
\sigma^{B^n} \leq \omega^{B^n},
~~
\lVert \omega^{B^n} - \sigma^{B^n} \rVert_1 \leq \epsilon,
~~
\lVert (\sigma^{B^n})^{-1} \rVert_\infty \leq
2^{n H(B)_\omega (1+3\delta)}.
\]

From Lemma~\ref{lem:projection2},
we already know that for any strongly $\delta$-typical type $\tau$,
for any vector 
$\ket{w_\tau} \in \hat{V}_\tau$, 
$
\lVert (I^{A^n} \otimes \Pi^{B^n}_{\omega, 3\delta}) \ket{w_\tau} \rVert_2^2 \geq
1 - \sqrt{\epsilon}.
$
We now have to show a similar result for an arbitrary linear combination
of vectors $\ket{w_\tau}$ over all strongly $\delta$-typical types $\tau$.
For this we invoke Lemma~\ref{lem:projection3} and 
Fact~\ref{fact:quantum_AEP}.
We thus conclude that for any vector
$\ket{v} \in \mathrm{supp}(\eta^{A^n B^n})$,
\[
\lVert (I^{A^n} \otimes \Pi^{B^n}_{\omega, 3\delta}) \ket{v} \rVert_2 \geq
1 - 8 {n + |A||B| - 1 \choose |A||B| - 1} \epsilon^{1/4} \geq
1 - \epsilon'.
\]

By Fact~\ref{fact:quantum_AEP}, the smallest non-zero eigenvalue
of $(\omega'''_\delta)^{B^n}$ (which is defined by setting $\epsilon=0$ in $\omega'''_{\epsilon, \delta}$ from Definition~\ref{def:Renyi2prime}) is smaller than 
$2^{-nH(B)_\omega(1-\delta)}$. Again invoking Fact~\ref{fact:quantum_AEP},
we conclude that 
$
\mathrm{supp}((\omega'''_{\epsilon,5\delta})^{B^n}) \geq
\mathrm{supp}(\sigma^{B^n}).
$
Thus, for any vector
$\ket{v} \in \mathrm{supp}(\eta^{A^n B^n})$,
$
\lVert 
(I^{A^n} \otimes \Pi^{B^n}_{\omega'''_{\epsilon,5\delta}}) \ket{v} 
\rVert_2 \geq
1 - \epsilon'.
$
Moreover, by Proposition~\ref{prop:H_max_iid}
\begin{align*}
\log \lVert (\omega'''_{\epsilon,5\delta})^{B^n})^{-1} \rVert_\infty =
(1+5\delta) \log \lVert ((\omega''_\epsilon)^{B^n})^{-1} \rVert_\infty &=
(1+5\delta) \mathbb{H}_{\max}^\epsilon(B^n)_\omega\\
& \leq
n (1+7\delta) H(B)_\omega.
\end{align*}

Again using Fact~\ref{fact:quantum_AEP}, we get
\begin{align*}
\lefteqn{\mathbb{H}_2^{\epsilon',5\delta}(A^n|B^n)_\omega} \\
& \geq 
-2\log 
\left\lVert
\left(I^{A^n} \otimes (\omega'''_{\epsilon,5\delta})^{B^n}\right)^{-1/4} 
\eta^{A^n B^n} 
\left(I^{A^n} \otimes (\omega'''_{\epsilon,5\delta})^{B^n}\right)^{-1/4} 
\right\rVert_2 \\
& \geq
-2\log (
\left\lVert
\left(I^{A^n} \otimes (\omega'''_{\epsilon,5\delta})^{B^n}\right)^{-1/4} 
\rVert_\infty^2 
\cdot
\lVert
\eta^{A^n B^n} 
\right\rVert_2
) \\
&   =  
-2\log 
\left(
\left\lVert \left(\omega'''_{\epsilon,5\delta})^{B^n}\right)^{-1} 
\right\rVert_\infty^{1/2}
\cdot
\left(
\sum_{\tau: \mathrm{typical}} \frac{p_\tau^2 |\hat{V}_\tau|}{|V_\tau|^2}
\right)^{1/2}
\right) \\
&  \geq 
-\log \left\lVert \left(\omega'''_{\epsilon,5\delta})^{B^n}\right)^{-1} 
\right\rVert_\infty
-2 \log
\left(
\sum_{\tau: \mathrm{typical}} \frac{p_\tau^2}{|V_\tau|}
\right)^{1/2} \\
&   =  
-\log \lVert (\omega'''_{\epsilon,5\delta})^{B^n})^{-1} \rVert_\infty
-2 \log 
\lVert
\Pi^{A^n B^n}_{\omega,\delta}
\omega^{A^n B^n}
\Pi^{A^n B^n}_{\omega,\delta}
\rVert_2 \\
& \geq 
- n H(B)_\omega (1+7\delta)
-2 \log (
2^{-n H(AB)_\omega (1-\delta)} \cdot
\lVert
\Pi^{A^n B^n}_{\omega,\delta}
\rVert_2
) \\
&   =  
- n H(B)_\omega (1+7\delta)
-2 \log (
2^{-n H(AB)_\omega (1-\delta)} \cdot
(\Tr \Pi^{A^n B^n}_{\omega,\delta})^{1/2}
) \\
& \geq 
- n H(B)_\omega (1+7\delta)
-2 \log 2^{-\frac{n}{2} H(AB)_\omega (1-3\delta)} \\
& \geq 
n H(A|B)_\omega -
n \delta (3 H(AB)_\omega + 7H(B)).
\end{align*}
Combining with Fact~\ref{fact:Renyi2upperbound} 
completes the proof of the proposition.
\end{proof}

\begin{remark}
\label{rem:modified_2entropy}
The importance of using smooth one shot entropic  quantities has been
remarked in many earlier works. In particular, it is only the smooth
quantities that approach their natural Shannon entropic counterparts
in asymptotic iid limit, not their unsmoothed version.
$H_2^\epsilon(\cdot|\cdot)$ is also shown to be bounded by its natural
Shannon counterpart in the asymptotic iid limit. However we would like
to emphasise here the importance of the additional smoothing parameter
$\delta$, which is crucial in this work.

We note that the smooth decoupling theorem was first proved by Szehr
et al. \cite{Smooth_decoupling} in terms of the conditional 
$H_{\min}^\epsilon$ quantity. Actually their proof
works without any change for the conditional $H_2^\epsilon$ quantity 
also, which is
a slightly stronger statement. However an important step in their
proof was to replace the original superoperator in the statement of
the decoupling theorem by another superoperator whose Choi matrix is a
smoothed version of the Choi matrix of the original superoperator.
They showed that on average over the Haar random unitary $U$ applied to
the input state, the new superoperator gives an output state close to
the output state of the original superoperator. However this does not
mean that for every unitary $U$, the output of the new superoperator is
close to the output of the original superoperator. However, in order
to prove a concentration result for decoupling, we need to bound a
certain Lipschitz constant. By its very definition the Lipschitz
constant should work for {\em all} unitaries not just for average Haar
random ones. This stringent requirement leads us to the second
smoothing
parameter $\delta$. Existence of this $\delta$ implies that for every
unitary U, the output state of the new superoperator is $\delta$-close
in trace distance to the output state of the original superoperator.
The parameter $\delta$ is the one shot analogue of the (non-trivial)
observation that there is a large subspace of the jointly typical
subspace of $(\rho^{A B})^{\otimes n}$ such that every vector in this
subspace has projection length of at most $1 - \delta$ on the tensor
product of the marginal typical subspaces on $(\rho^A)^{\otimes n}$
and $(\rho^B)^{\otimes n}$.
\end{remark}

\begin{remark}
Consider fixed systems $A$, $B$ and a fixed state 
$\omega^{AB}$. For a  fixed $\delta$, divide the smooth modified 
conditional R\'{e}nyi $2$-entropy, the smooth conditional
R\'{e}nyi $2$-entropy and the smooth modified max-entropy
by $n$ and let $\epsilon \rightarrow 0$. This implies that 
$n \rightarrow \infty$ and
$\epsilon' \rightarrow 0$. Finally, let $\delta \rightarrow 0$. This
shows that in the asymptotic iid limit, the smooth conditional
R\'{e}nyi $2$-entropy divided by $n$ is lower bounded by the 
smooth modified conditional
R\'{e}nyi $2$-entropy divided by $n$ which is further lower bounded by 
the conditional
Shannon entropy, and the smooth  modified
max-entropy divided by $n$ is upper bounded by the 
Shannon entropy.
\end{remark}

\subsection{Proof of the iid extension of Theorem~\ref{thm:main}}
\label{sec:decoupling_iid}
In this section we take our main one-shot concentration result
and apply it in the asymptotic iid
setting.  That is, we take the $n$-fold tensor product copy of the channel
$\cT$ and the state $\rho^{AR}$, apply Theorem~\ref{thm:main} to it,
and obtain bounds in terms of the standard Shannon entropies.
\begin{corollary}
\label{cor:main}
Consider the setting of Theorem~\ref{thm:main} above.
Consider the density matrix $\omega^{A'B}$.
Let $\ket{w_j}^{A'B}$, $j \in [|A| |B|]$ be the eigenvectors of
$\omega^{A'B}$ with eigenvalues $q_j$. 
For $j \in [|A| |B|]$, define 
$\theta_j^B := \Tr_{A'} [\ket{w_j}\bra{w_j}^{AB}]$.
Let $p_j$, $j \in [|A||B|]$ be the probability
distribution on $[|B|]$ obtained by measuring 
$\theta_j$ in the eigenbasis of $\omega^{B}$. 
Define 
$
q_{\mathrm{min}} := 
2^{-\mathbb{H}_{\max}^{\epsilon/2}(A'B)_\omega},
$ 
$
p_{\mathrm{min}} := 
\min_{j \in [|A| |B|]}
2^{-\mathbb{H}_{\max}^{\epsilon/2}([|B|])_{p_j}},
$ 
$
n := 2^5 q^{-1}_{\mathrm{min}} p^{-1}_{\mathrm{min}} \delta^{-2} 
\log\frac{|A| |B|}{\epsilon}.
$
Consider the $n$-fold tensor powers 
$\omega^{(A')^n B^n} := (\omega^{A'B})^{\otimes n}$,
$\rho^{A^n R^n} := (\rho^{AR})^{\otimes n}$.
Let $\epsilon' := 8 (n+|A||B|)^{|A||B|} \epsilon^{1/4}$.
Let $\kappa > 0$.
Then,
\begin{align*}
\underset{U}{\prob}\bigg[
f(U) >
2^{
-\frac{n}{2} (H(A|R)_\rho - \delta(3 H(AR)_\rho + 7 H(R)_\rho))
-\frac{n}{2} (H(A'|B)_\omega - \delta(3 H(A'B)_\omega + 7 H(B)_\omega))
}\\ 
+ 28 (\epsilon')^{1/4} + 2 \kappa
\bigg] 
\leq 7 \cdot 2^{-a \kappa^2},~~~~~~~&
\\\hspace{0.9\textwidth}
\end{align*}
which can also be expressed as:

\begin{align*}
\underset{U}{\prob}\bigg[
f(U) >
2^{
-\frac{n}{2} (H(A|R)_\rho - \delta(3 H(AR)_\rho + 7 H(R)_\rho))
-\frac{n}{2} (H(A'|B)_\omega - \delta(3 H(A'B)_\omega + 7 H(B)_\omega))
}\\ 
+ 28 (\epsilon')^{1/4} + \sqrt{t/2} \, \beta
\bigg] 
\leq 7 \cdot 2^{-a \kappa^2},~~~~~~~&
\\\hspace{0.9\textwidth}
\end{align*}

where the unitary $U^{A^n}$ is chosen uniformly at random from a 
$(|A|^n, s, \lambda, t)$-qTPE, and the parameters $a$, $t$, $\alpha$, 
$\beta$ are defined as
\[
\begin{array}{c}
a := |A|^n \cdot
2^{
n (H(A|R)_\rho - \delta(3 H(AR)_\rho + 7 H(R)_\rho))
- n H(B)_\omega (1+7\delta) - 9
},\\
t :=
|A|^n \kappa^2
2^{
n(H(A|R)_\rho + 32 \sqrt{\epsilon'}) + \log (\epsilon')^{-1}
- n H(B)_\omega (1-5\delta) - 6
},\\
\beta := \sqrt{1/a}, \\
\end{array}
\] 
and $\mu$, $\lambda$ are defined in Theorem~\ref{thm:main} above.
\end{corollary} 
\begin{proof}
The proof follows by a direct application of 
Theorem~\ref{thm:main} and 
Propositions~\ref{prop:H_max_iid} and \ref{prop:H_2_iid}.
We just need to keep in mind that the `weighting state' in the definition
of $\mathbb{H}_2^{\epsilon', 5\delta}((A')^n|B^n)_{\omega^{\otimes n}}$ 
is $(\omega'''_{\epsilon,5\delta})^{B^n}$ as in the proof of
Proposition~\ref{prop:H_2_iid}, and the function $g(U)$ is defined
with respect to the perturbed Choi-Jamio{\l}kowski state 
$\eta^{(A')^n B^n}$ defined in Proposition~\ref{prop:H_2_iid}  
and a perturbed input state $(\rho')^{A^n R^n}$ which is
$4\sqrt{\epsilon'}$-close to the original input state
$\rho^{A^n R^n}$.  We get
\begin{eqnarray*}
\mu 
& \leq &
2^{
-\frac{1}{2} \mathbb{H}_2^{\epsilon', 5\delta}(
A^{\prime\; n}|B^n)_{\omega^{\otimes n}} -
\frac{1}{2} H_2^{4 \sqrt{\epsilon'}}(A^n|R^n)_{\rho^{\otimes n}}
} \\
& \leq &
2^{
-\frac{n}{2} (H(A|R)_\rho - \delta(3 H(AR)_\rho + 7 H(R)_\rho))
-\frac{n}{2} (H(A'|B)_\omega - \delta(3 H(A'B)_\omega + 7 H(B)_\omega))
}, \\ 
a 
& = & 
|A|^n 
2^{
H_2^{4 \sqrt{\epsilon'}}(A^n|R^n)_{\rho^{\otimes n}} - 
(1+\delta)\mathbb{H}_{\max}^{\epsilon}(B^n)_{\omega^{\otimes n}} - 9
}\\
& \geq & 
|A|^n 
2^{
n (H(A|R)_\rho - \delta(3 H(AR)_\rho + 7 H(R)_\rho)) - 
n H(B)_\omega (1+7\delta) - 9
}, \\ 
t 
& = & 
|A|^n \kappa^2
2^{
H_2^{4 \sqrt{\epsilon'}}(A^n|R^n)_{\rho^{\otimes n}} - 
\mathbb{H}_{\max}^{\epsilon}(B^n)_{\omega^{\otimes n}} - 6
}\\
& \leq & 
|A|^n \kappa^2
2^{
n (H(A|R)_\rho + 32 \sqrt{\epsilon'}) + \log (\epsilon')^{-1}  - 
n H(B)_\omega (1-5\delta) - 6
}.
\end{eqnarray*}
Substituting the above expressions in Theorem~\ref{thm:main} proves 
the desired corollary.
\end{proof}
\section{Decoupling under partial trace (also referred to as the 
Fully Quantum Slepian Wolf (FQSW) \cite{mother_protocol})}
\label{sec:fqsw}
In this section, we prove the concentration result for the 
Fully Quantum Slepian Wolf (FQSW) problem with respect to unitary 
designs.
\begin{theorem}
\label{thm:fqsw}
Consider the setting of Theorem~\ref{result:main}. Consider the FQSW
decoupling function
\[
f(U) = f_{FQSW}(U^{A_1 A_2}) := 
\lVert 
\Tr_{A_2} [(U^{A_1 A_2} \otimes I^R) \circ \rho^{A_1 A_2 R})] - 
\pi^{A_1} \otimes \rho^R 
\rVert_1.
\]
Suppose we are promised that 
\[
\begin{array}{c}
\lVert (\trho')^R \rVert_2^2 <
0.9 |A_1| |A_2| \lVert (\trho')^{AR} \rVert_2^2,
|A_1| \geq 2, |A_2| > |A_1|, \\
|A_2| 2^{H_2^\epsilon(A_1 A_2|R)_\rho - 8} - 4 > 
-H_2^\epsilon(A_1 A_2|R) +
\log |A_1| + 2 \log |A_2|.
\end{array}
\]
In other words $\rho^{A_1 A_2 R}$ is not too close to a tensor product 
state on 
$A_1 A_2 \otimes R$ or $\rho^{A_1 A_2}$ is not too close to the
completely mixed state $\pi^{A_1 A_2}$.
The following concentration inequality holds:
\[
\prob_{U \sim \mathrm{design}}\left[
f(U) > 
\sqrt{\frac{|A_1|}{|A_2|}} \cdot
2^{-\frac{1}{2} H_2^\epsilon(A|R)_\rho + 1 } 
+ 14 \sqrt{\epsilon} + 2 \kappa
\right] \leq 
7 \cdot 2^{-a \kappa^2},
\]
where the unitary $U^A$ is chosen uniformly at random from an 
$(|A_1||A_2|, s, \lambda, t)$-qTPE
$a := |A_2| 2^{H_2^\epsilon(A|R)_\rho - 9}$ and
$t := 8 a \kappa^2$.
The quantity $\lambda$, defined in 
Theorem~\ref{thm:main} satisfies the inequality
\[
\left(0.008 |A_2|^{-9} |A_1|^{-13} 
2^{-H_2^\epsilon(A_1 A_2|R)_\rho}\right)^t <
\lambda < 
\left(|A_2|^{-9} |A_1|^{-13} 2^{-H_2^\epsilon(A_1 A_2|R)_\rho}\right)^t.
\]
for values of $\kappa = O(\mu)$, $t = O(|A_1|)$.
Above concentration inequality can also be expressed as:

\begin{align}
\prob_{U \sim \mathrm{design}}[
f(U) > 
\sqrt{\frac{|A_1|}{|A_2|}} \cdot
2^{-\frac{1}{2} H_2^\epsilon(A|R)_\rho + 1 } 
+ 14 \sqrt{\epsilon} + \beta \, \sqrt{|A_1|/2} 
] \leq 
7 \cdot 2^{-\frac{|A_1|}{2}},
\end{align}
where $\beta := \sqrt{1/a}$.
Thus, if $|A_1| \leq \polylog(|A_2|)$ and $\kappa = O(\mu)$  
then efficient
constructions for such qTPEs exist.
\end{theorem}
\begin{proof}
In order to obtain the desired concentration result,
we apply Theorem~\ref{thm:main} with  the following parameters: 
\begin{itemize}

\item 
The input system $A$ to the superoperator
is $A := A_1 \otimes A_2$.
Output system $B := A_1$ and superoperator 
$\cT^{A \to A_1} := \Tr_{A_2}$;

\item 
The state 
$
\omega^{A'A_1} = 
(\cT^{A \to A_1} \otimes \I^{A'})(\Phi^{AA'}) = 
\Phi^{A_1 A'_1} \otimes  \pi^{A'_2};
$

\item 
Take $\delta = 0$.
We get $(\omega'''_{\epsilon, 0})^{A_1}=\pi^{A_1}$. Hence
$\mathbb{H}_{\max}^\epsilon(A_1)_\omega=\log |A_1|$ as the 
reduced state $\omega^{A_1}=\pi^{A_1}$;

\item 
The matrix
$
(\tomega')^{A'A_1} :=
(I^{A'_1} \otimes I^{A'_2} \otimes 
(\omega'''_{\epsilon, 0})^{A_1}
)^{-1/4} \circ 
(\Phi^{A_1 A'_1} \otimes \pi^{A'_2}) =
|A_1|^{1/2} (\Phi^{A_1 A'_1} \otimes \pi^{A'_2}).
$
Note that
$
(H'_2)^{\epsilon, 0}(A'|A_1)_\omega = 
-2 \log \lVert (\tomega')^{A' A_1} \rVert_2. 
$
We have,
\begin{align*}
\lVert 
(\tomega')^{A' A_1}
\rVert_2 
&= 
\lVert |A_1|^{1/2} \Phi^{A_1 A'_1} \otimes \pi^{A'_2} \rVert_2 
\;= \;
\sqrt{\frac{|A_1|}{|A_2|}}, \\
\lVert 
(\tomega')^{A_1}
\rVert_2 
& =
\lVert 
|A_1|^{1/2} \Phi^{A_1} 
\rVert_2 
\;=\;
1, \\
\eta 
& = 
\frac{
\lVert (\tomega')^{A' A_1} \rVert_2^2
}{
\lVert (\tomega')^{A_1} \rVert_2^2 
} 
\;=\;
\frac{|A_1|}{|A_2|}, \\
\delta_1 
& =
\lVert (\tomega')^{A_1} \rVert_2^2 
\frac{|A|^2 - |A| \eta}{|A|^2 - 1}
\;=\;
\frac{|A_1|^2 |A_2|^2 - |A_1|^2}{|A_1|^2 |A_2|^2 - 1}, \\
\delta_2
& =
\lVert (\tilde{\omega}')^{A'A_1}\rVert_2^2 
\frac{|A|^2 - |A| \eta^{-1}}{|A|^2 - 1}
\;=\;
\frac{|A_1|}{|A_2|} \cdot
\frac{|A_1|^2 |A_2|^2 - |A_2|^2}{|A_1|^2 |A_2|^2 - 1};
\end{align*} 

\item
The function
$
g(U) = 
|A_1|^{1/2}
\lVert
(\Tr_{A_2} \otimes \I^R)((U^A \otimes I^R) \circ (\trho')^{AR}) -
\pi^{A_1} \otimes (\trho')^R
\rVert_2.
$
Then,
\begin{align*}
\E_{U \sim \Haar}[(g(U))^2]
& =
\delta_1 \lVert (\trho')^R \rVert_2^2 + 
\delta_2 \lVert (\trho')^{AR} \rVert_2^2 -
\lVert (\tomega')^{A_1} \rVert_2^2 \lVert (\trho')^R \rVert_2^2 \\
& =
-\frac{|A_1|^2 - 1}{|A_1|^2 |A_2|^2 - 1}
\lVert (\trho')^R \rVert_2^2 + 
\frac{|A_1|}{|A_2|} \cdot
\frac{|A_1|^2 |A_2|^2 - |A_2|^2}{|A_1|^2 |A_2|^2 - 1}
\lVert (\trho')^{AR} \rVert_2^2 \\
& \leq 
\frac{|A_1|}{|A_2|} \cdot
\lVert (\trho')^{AR} \rVert_2^2, \\
\underset{{U \sim \Haar}}{\E}[(g(U))^2]
& =
-\frac{|A_1|^2 |A_2|^2 -  |A_2|^2}{(|A_1|^2 |A_2|^2 - 1) |A_2|^2}
\lVert (\trho')^R \rVert_2^2 + 
\frac{|A_1|}{|A_2|} \cdot
\frac{|A_1|^2 |A_2|^2 - |A_2|^2}{|A_1|^2 |A_2|^2 - 1}
\lVert (\trho')^{AR} \rVert_2^2 \\
& \geq 
\frac{0.1 |A_1|}{|A_2|} \cdot
\frac{|A_1|^2 |A_2|^2 - |A_2|^2}{|A_1|^2 |A_2|^2 - 1}
\lVert (\trho')^{AR} \rVert_2^2  \\
& \geq
\frac{0.1 |A_1|}{|A_2|} 
\left(1 - \frac{|A_2|^2}{|A_1|^2 |A_2|^2 }\right)
\lVert (\trho')^{AR} \rVert_2^2 \\
& \geq 
\frac{0.07 |A_1|}{|A_2|} \cdot
\lVert (\trho')^{AR} \rVert_2^2; 
\end{align*}

\item 
The tail probability exponent $a$ becomes
$
a = 
|A| |A_1|^{-1} 2^{H_2^\epsilon(A|R)_\rho - 9} = 
|A_2| 2^{H_2^\epsilon(A|R)_\rho - 9}.
$

\item
Define $\mu := \E_{\Haar}[g(U)]$. 
Since 
$
a \cdot \E_{U \sim \Haar}[(g(U))^2] +
\log \E_{U \sim \Haar}[(g(U))^2]  > \log |A_1| + \log |A_2| +
\mathbb{H}_{\max}^\epsilon(A_1)_\omega - H_2^\epsilon(A_1 A_2|R)_\rho,
$
we get $\mu^2 \leq \E[(g(U))^2] \leq 8\mu^2$;

\item
The qTPE parameter $t$ becomes
$
t = 
|A| \kappa^2 \cdot 2^{H_2^\epsilon(A|R)_\rho - 6}\cdot |A_1|^{-1} =
|A_2| \kappa^2 \cdot 2^{H_2^\epsilon(A|R)_\rho - 6}.
$
For values of $\kappa = O(\mu)$, $t = O(|A_1|)$;

\item
The qTPE parameter $\lambda$ becomes
$
\lambda =
(|A|^{-8} |A_1|^{-6} \mu^2)^t,
$
which satisfies
\[
(0.008 |A_2|^{-9} |A_1|^{-13} 2^{-H_2^\epsilon(A_1 A_2|R)_\rho})^t <
\lambda < 
(|A_2|^{-9} |A_1|^{-13} 2^{-H_2^\epsilon(A_1 A_2|R)_\rho})^t.
\]
Sequentially iterating 
$O(t (\log |A_1| + \log |A_2|))$ times an
$(|A_1||A_2|, s, O(1), t)$-qTPE gives us the desired
$(|A_1||A_2|, s, \lambda, t)$-qTPE for derandomisation.
Observe that if $t \leq \polylog(|A_2|)$, then efficient
constructions for such qTPEs exist \cite{brandao2012local, sen:zigzag}.

\end{itemize}

Now substituting these parameters in Theorem~\ref{thm:main} we get
\[
\prob_{U^{A_1 A_2} \sim \TPE}\left[
f(U) > 
2 \sqrt{\frac{|A_1|}{|A_2|}} 2^{-\frac{1}{2} H_2^\epsilon(A|R)_\rho} 
+ 14 \sqrt{\epsilon} + 2 \kappa
\right] \leq 
7 \cdot 2^{-a (\min\{\mu^2, \kappa^2\})},
\]
for $U$ chosen uniformly at random from a
$(|A_1| |A_2|, s, \lambda, t)$-qTPE.
This completes the proof.
\end{proof}

In the following section we give two applications of a simple form of 
our general decoupling, via Fully Quantum Slepian Wolf 
Theorem~\ref{thm:fqsw}. These applications are widely studied in 
theoretical physics and also gives an intuitive understanding of our 
theorem and a simple calculation of the various entropic quantities that 
govern the concentration of the decoupling theorem. These further aids 
the understanding that the nature is not as random as Haar measure 
rather performs efficient computation, at times the parameters 
oblivious to us makes it more random. 

\section{Applications of measure concentration of the decoupling 
theorem with unitary designs via FQSW}
\label{sec:Applications}
\subsection{Relative Thermalization}
An immediate application of measure concentration of FQSW is 
in the area of
quantum thermodynamics, in describing a  process called 
relative thermalization \cite{RelativeThermalization}.  One of 
the most fundamental questions in quantum thermodynamics is how a small
system starting out in a particular quantum state spontaneously 
thermalizes when brought in contact with a much larger environment e.g.
a bath. More
precisely when brought in contact with a bath,
the small system decouples from any other system, which we may
call as the reference system, it may be initially entangled with.
The formal definition of relative thermalization is as follows:
\begin{definition}
\label{def:relthermal}
Let system $S$, environment $E$ and reference $R$ be quantum systems and 
$\Omega \subseteq S \otimes E$ be a subspace corresponding to a physical 
constraint such as total energy. The global system is in a state 
$\rho^{\Omega R}$, supported in the Hilbert space 
$\Omega \otimes R$. The time evolution is described by a unitary
on $S \otimes E$. The state after time evolution is denoted by
$\sigma^{\Omega R}$. The system $S$ is said to be $\kappa$-thermalized 
relative to $R$ in state $\sigma^{\Omega R}$ if:
\[
\lVert \sigma^{SR} - \omega^S \otimes \sigma^R \rVert_1 \leq \kappa
\]
where $\sigma^{SR} := \Tr_E[\sigma^{\Omega R}]$ and 
$\omega^S \triangleq \Tr_E [ \frac{I^\Omega}{|\Omega|} ]$ is 
the so called local microcanonical state.
\end{definition}
Thus, relative thermalization requires that, after the environment 
$E$ is traced out, the system $S$
should be close to the state $\omega^S$ and should not have strong 
correlations with the reference $R$. If the time-evolution of 
$S \otimes E$ is modelled by a Haar random unitary on $\Omega$, then 
Fact~\ref{fact:dupuisexpectation} 
guarantees that relative thermalization occurs in expectation over 
the Haar measure. Furthermore, 
Fact~\ref{fact:dupuisconcentration} says that  
$
1- 
\exp\left( 
-\frac{|\Omega| \kappa^2}{2^{-H_{\mathrm{min}}^\epsilon(\Omega)_\rho + 6}}
\right)
$ fraction of Haar random unitaries achieve 
relative thermalization,
if 
$
2^{
-\frac{1}{2} H_2^\epsilon(\Omega|R)_\rho 
- \frac{1}{2}H_2^\epsilon(\Omega'|S)_\omega
} + 16 \epsilon \leq
\frac{\kappa}{2}.
$

Since Haar random unitaries are computationally inefficient, it is 
natural to wonder whether nature truly evolves via Haar random unitary. 
Hence, the work of Nakata et al.~\cite{Winter_decoupling} investigates 
what happens if 
the evolution of system plus environment is modelled by a unitary chosen 
from an efficiently implementable approximate unitary $2$-design. Their
unitary acts on the subspace $\Omega$ only. They 
show that relative thermalization indeed takes place for the same
parameter regime as Haar random unitaries, but for a much 
smaller fraction 
$
1-
\exp\left(
-\frac{\kappa^4}
{|\Omega|^3 2^{-4 H_{\mathrm{min}}^\epsilon(\Omega)_\rho + 22}}
\right)
$ of design unitaries. It is 
reasonable to expect that 
$
H_{\mathrm{min}}^\epsilon(\Omega)_\rho \leq
\log |\Omega| / 2,
$ 
i.e. $\rho^\Omega$ is not highly mixed over $\Omega$.
In this case the fraction of unitaries
achieving relative thermalization is only guaranteed to be at least
$1-\exp\left( -\frac{\kappa^4}{2^{22} |\Omega|} \right)$, which is
almost zero for large $|\Omega|$.

We now analyze relative thermalization using the lens of unitary designs.
We follow the proof technique of Theorem~\ref{thm:fqsw}.
\begin{itemize}

\item 
The superoperator is $\Tr_E$ with input system  $\Omega$ and output
system $S$;

\item 
The state 
$
\omega^{\Omega' S} = 
(\Tr_E \otimes \I^{\Omega'})(\Phi^{\Omega \Omega'})
$
and
$
{(\trho')}^{\Omega R}:=(\xi^R)^{-1/4} \circ \rho^{\Omega R};
$
\item 
Take $\epsilon = \frac{\kappa^2}{60}$ and $\delta = 0$.
Let 
$
2^{
-\frac{1}{2} H_2^\epsilon(\Omega|R)_\rho 
- \frac{1}{2}\mathbb{H}_2^{\epsilon,0}(\Omega'|S)_\omega
} \leq
\frac{\kappa}{4},
$
$
\mathbb{H}_2^{\epsilon,0}(\Omega'|S)_\omega \leq 
\log |\Omega| - \log |S|,
$
$|S| > 2$,
$\mathbb{H}_{\max}^\epsilon(S)_\omega = O\log (|S|)$, 
and
$
\lVert (\trho')^R \rVert_2^2 < 
0.9 |\Omega| \lVert (\trho')^{\Omega R} \rVert_2^2;
$

\item 
The matrix
$
(\tomega')^{\Omega' S} :=
(I^{\Omega'} \otimes (\omega'''_{\epsilon, \delta})^{S}
)^{-1/4} \circ 
\omega^{\Omega' S}.
$
Note that
$
(H'_2)^{\epsilon, 0}(\Omega'|S)_\omega = 
-2 \log \lVert (\tomega')^{\Omega' S} \rVert_2, 
$
$
H_2^{\epsilon}(\Omega|R)_\rho = 
-2 \log \lVert (\trho')^{\Omega R} \rVert_2, 
$
We have:
\begin{align*}
1 & \geq
\lVert 
(\tomega')^{S}
\rVert_2 
\geq 
\sqrt{1-\epsilon}, \\
\eta 
& = 
\frac{
\lVert (\tomega')^{\Omega' S} \rVert_2^2
}{
\lVert (\tomega')^{S} \rVert_2^2 
}, 
~~~
|\Omega|^{-1} \leq
2^{-\mathbb{H}_2^{\epsilon,0}(\Omega'|S)_\omega} \leq
\eta \leq
(1-\epsilon)^{-1} 2^{-\mathbb{H}_2^{\epsilon,0}(\Omega'|S)_\omega} \leq
1, \\
\delta_1 
& =
\lVert (\tomega')^{S} \rVert_2^2 
\frac{|\Omega|^2 - |\Omega| \eta}{|\Omega|^2 - 1}, ~~~
\lVert (\tomega')^{S} \rVert_2^2 (1 - |\Omega|^{-1}) \leq 
\alpha \leq 
\lVert (\tomega')^{S} \rVert_2^2, \\
\delta_2
& =
\lVert (\tilde{\omega}')^{\Omega' S}\rVert_2^2 
\frac{|\Omega|^2 - |\Omega| \eta^{-1}}{|\Omega|^2 - 1}, ~~~
\lVert (\tilde{\omega}')^{\Omega' S}\rVert_2^2
(1 - |S|^{-1}) \leq
\beta \leq
\lVert (\tilde{\omega}')^{\Omega' S}\rVert_2^2;
\end{align*} 

\item
The function\\
$
g(U) := 
\left\lVert
(I^{R} \otimes (\omega'''_{\epsilon, \delta})^{S}
)^{-1/4} \circ 
(
(\Tr_{E} \otimes \I^R)((U^\Omega \otimes I^R) \circ (\trho')^{\Omega R}) -
\omega^{S} \otimes (\trho')^R
)
\right\rVert_2.
$
Then,
\begin{align*}
\E_{U \sim \Haar}[(g(U))^2]
& =
\delta_1 \lVert (\trho')^R \rVert_2^2 + 
\delta_2 \lVert (\trho')^{\Omega R} \rVert_2^2 -
\lVert (\tomega')^{S} \rVert_2^2 \lVert (\trho')^R \rVert_2^2 \\
&\leq 
\lVert (\tomega')^{\Omega' S} \rVert_2^2 \lVert (\trho')^{\Omega R} 
\rVert_2^2, \\
\E_{U \sim \Haar}[(g(U))^2] 
& \geq 
\lVert (\tomega')^{\Omega' S} \rVert_2^2 \lVert (\trho')^{\Omega R} 
\rVert_2^2
(1 - |S|^{-1})
-
\lVert (\tomega')^{S} \rVert_2^2 
\lVert (\trho')^R \rVert_2^2 |\Omega|^{-1} \\
& \geq 
0.05 \lVert (\tomega')^{\Omega' S} \rVert_2^2 \lVert (\trho')^{\Omega R} 
\rVert_2^2
\end{align*}

\item 
The tail probability exponent $a$ becomes
$
a \geq 
|\Omega| |S|^{-1} 2^{H_2^\epsilon(\Omega|R)_\rho - 9};
$

\item
Define $\mu := \E_{\Haar}[g(U)]$. 
We ensure that\\
$
a \cdot \E_{U \sim \Haar}[(g(U))^2] +
\log \E_{U \sim \Haar}[(g(U))^2]  > \log |\Omega| +
\mathbb{H}_{\max}^\epsilon(S)_\omega - H_2^\epsilon(\Omega|R)_\rho,
$
so that we get $\mu^2 \leq \E[(g(U))^2] \leq 8\mu^2$;

\item
The qTPE parameter $t$ becomes
$t = 8 a \kappa^2.$   If 
$
\kappa^2 = 
\poly(|S|) |\Omega|^{-1} \cdot
2^{-H_2^\epsilon(\Omega|R)_\rho}
$
(which is the case when $\kappa$ is close to $\mu$ in FQSW),
then
$t \leq \poly(|S|)$;

\item
The qTPE parameter $\lambda$ becomes
$
\lambda =
(|\Omega|^{-8} |S|^{-6} \mu^2)^t.
$
Sequentially iterating 
$O(t \log |\Omega|)$ times an
$(|\Omega|, s, O(1), t)$-qTPE gives us the desired
$(|\Omega|, s, \lambda, t)$-qTPE for derandomisation.
Observe that if $t \leq \polylog(|\Omega|)$, then efficient
constructions for such qTPEs exist \cite{brandao2012local, sen:zigzag}.

\end{itemize}

We have thus proved the following theorem for relative thermalization.
\begin{theorem}
\label{thm:relthermal}
Consider the setting of Theorem~\ref{thm:main} and
Definition~\ref{def:relthermal}.
Suppose we are promised that 
\[
\begin{array}{c}
2^{
-\frac{1}{2} H_2^\epsilon(\Omega|R)_\rho 
- \frac{1}{2}\mathbb{H}_2^{\epsilon,0}(\Omega'|S)_\omega
} \leq
\frac{\kappa}{4},
\mathbb{H}_2^{\epsilon,0}(\Omega'|S)_\omega \leq 
\log |\Omega| - \log |S|, 
|S| > 2, 
\mathbb{H}_{\max}^\epsilon(S)_\omega = O\log (|S|), \\
\lVert (\trho')^R \rVert_2^2 < 
0.9 |\Omega| \lVert (\trho')^{\Omega R} \rVert_2^2,
2^{-14} |\Omega| |S|^{-1} \cdot
2^{H_2^\epsilon(\Omega|R)_\rho}
- \mathbb{H}_2^{\epsilon,0}(\Omega'|S)_\omega >
2 \log |\Omega.
\end{array}
\]
Then $S$ is $\kappa$-thermalized relative to $R$ in state
$\sigma^{\Omega R}$ for an
$
1 - 5 \cdot 2^{-a \kappa^2}
$
fraction of unitaries $U$ on $\Omega$ 
where $U$ is chosen uniformly at random from a 
$(|\Omega|, s, \lambda, t)$-qTPE
$
a := |\Omega| |S|^{-1} \cdot
	2^{H_2^\epsilon(A|R)_\rho - 9}
$, 
$t := 8 a \kappa^2$ and
$\lambda := (|\Omega|^{-8} |S|^{-6} \mu^2)^t$.
For 
$
\kappa^2 = \poly(|S|) |\Omega|^{-1} \cdot 2^{-H_2^\epsilon(A|R)_\rho},
$
$t \leq \poly(|S|)$.
For such $\kappa$, which includes the case where $\kappa = O(\mu)$, 
if $|S| \leq \polylog(|\Omega|)$, then efficient
constructions for such qTPEs exist.
\end{theorem}

Theorem~\ref{thm:relthermal} achieves better performance
than the result of Dupius (Fact~\ref{fact:dupuisconcentration}) 
and the result of Nakata et al.~\cite{Winter_decoupling} in the 
following sense:
\begin{enumerate}

\item
In our result, the system plus environment 
evolves according to a unitary chosen uniformly at random
from an approximate unitary $t$-design 
for moderate values of $t$. Our unitary acts on the subspace $\Omega$
only. For a wide range of parameters, our unitaries require 
$O(|S| \log |\Omega|)$ random bits for a precise description which
is less than the $\Omega(|\Omega|^2 \log |\Omega|)$ random bits required 
by the Haar random unitaries of Dupuis, as well as less than
the $\Omega(|\Omega| \log |\Omega|)$ random bits required by the
approach of Nakata et al. Moreover, the random unitaries used by
Dupuis and by  Nakata et al. are not efficiently implementable,
whereas our unitaries are efficiently implementable
when $|S| \leq \polylog(|\Omega|)$.

\item
Our Theorem~\ref{thm:relthermal} shows that 
relative thermalisation still takes place for the fraction 
$1 - 5 \cdot 2^{-a \kappa^2}$ of unitaries, where
$
a = 
|\Omega| |S|^{-1} \cdot 2^{H_2^\epsilon(\Omega|R)_\rho - 9}.
$
Note that $H_2^{\epsilon}(\Omega|R)_\rho \geq -\log |\Omega|$ for
any state $\rho^{\Omega R}$. The equality is achieved 
when $\rho^{\Omega R}$ is maximally entangled on $\Omega$.
Under the reasonable 
assumption that $H_2^{\epsilon}(\Omega|R)_\rho \geq -0.5 \log |\Omega|$,
i.e., $\rho^{\Omega R}$ is not highly entangled on $\Omega$,
the fraction of unitaries that 
achieve relative thermalisation is at least
$1 - 5 \cdot 2^{-2^{-9} \cdot |\Omega|^{1/2} |S|^{-1} \kappa^2}.$
As $\Omega$ is generally of a much larger 
dimension than the system $S$, it is reasonable to expect that
$|S| < |\Omega|^{1/4}$. In that case, the fraction of unitaries
that achieve relative thermalisation in our result is guaranteed to
be at least
$1 - 5 \cdot 2^{-2^{-9} \cdot |\Omega|^{1/4} \kappa^2}$
which is nearly one for large $|\Omega|$.
We also assume that 
$H_{\mathrm{min}}^\epsilon(\Omega) \leq \log |\Omega| / 2$,
i.e. the state $\rho^\Omega$ is not highly mixed on $\Omega$.
For this range of parameters,
our decoupling result is much better than  that of
Nakata et al. which can only guarantee that
$1-\exp\left(-\frac{\kappa^4}{2^{22} |\Omega|} \right) \approx 0$ 
fraction of unitaries achieve relative thermalisation. 
However, our result is worse than that of Dupuis which guarantees that
$1-\exp\left(-\frac{|\Omega|^{3/2} \kappa^2}{2^6}\right)$ fraction
of Haar random unitaries achieve relative thermalisation.
Note that the unitaries of Nakata et al. and Dupuis used for obtaining 
concentration of measure for  decoupling, and consequently for relative 
thermalisation, are not efficient to implement. 
Only the unitaries used to obtain decoupling in expectation by 
Nakata et al. can be implemented efficiently.
\end{enumerate}

We summarise our relative thermalization result 
in Table~\ref{table_thermalisation} for a clear 
comparison with the known results of Dupuis and Nakata et al.
\begin{center}
\begin{table}[!h]
\small
\begin{tabular}{|c|c|c|c|}
\hline
& Dupius\cite{decoupling} 
& Nakata et al. \cite{Winter_decoupling} 
& Theorem~\ref{thm:relthermal} \\
\hline
Fraction
& $ 1-\exp\left(-|\Omega|^{3/2} \kappa^2\right)$ 
& 
$ 
1-\exp\left(-|\Omega|^{-1} \kappa^2\right)
$ 
& 
$1-\exp\left(-|\Omega|^{1/2} \kappa^2\right)$ \\
\hline
Randomness 
& $|\Omega|^2 \log |\Omega|$
& $|\Omega| \log |\Omega|$
& $\poly(|S|) \log |\Omega|$ \\
\hline
Random unitary
& Haar 
& $X$ and $Z$-diagonal
& Approx. $\poly(|S|)$-design \\
\hline
Efficiency
& Always inefficient 
& Always inefficient  
& Efficient for $|S| = \polylog(|\Omega|)$ \\
\hline
\end{tabular}
\caption{
Achieving $\kappa$-relative thermalization, 
$
\kappa^2 = \poly(|S|) |\Omega|^{-1} \cdot 
	2^{-H_2^\epsilon(A|R)_\rho},
$
$H_{\mathrm{min}}^\epsilon(\Omega)_\rho = \log |\Omega| / 2$,
$H_2^{\epsilon}(\Omega|R)_\rho \geq -0.5 \log |\Omega|$,
constant factors ignored.
} 
\label{table_thermalisation}
\end{table}
\end{center} 

\begin{remark}
\label{rem:Thermal}
Nakata et al's work and other works on relative thermalisation shows
that for some well known Hamiltonians in physics relative
thermalisation eventually occurs. Concentration bounds, even if given,
could not answer our everyday observation in nature that
thermalisation occurs with overwhelmingly high probability. Prior to
this work, only Haar random unitaries where known to achieve
thermalisation with overwhelmingly high probability. Our work is the
first one to show that relative thermalisation can indeed take place
with overwhelmingly high probabillity using unitaries that are simpler
than Haar random in a precise sense. Moreovoer, for a particular range
of parameters, relative thermalisation can even be achieved by
computationally efficient unitaries. Physically this work lays out a
clear direction of further resarch,  that simple unitaries made up of
a few basic builiding blocks can achieve relative thermalisation with
overwhelmingly high probability.
\end{remark}

\subsection{Application to the black hole information paradox model 
by Hayden-Preskill \cite{Hayden_Preskill}}
A brief description of Hayden-Preskil toy model for the black hole 
information paradox \cite{Hayden_Preskill}:
Hayden-Preskill in their seminal paper \cite{Hayden_Preskill} described 
the evolution of black hole as a unitary operator. Further they assumed 
that an old black hole which has emitted half of its information via 
Hawking's radiation can be assumed to be in a maximally entangled 
state with its environment. Alice throws a quantum register of 
$\log |M|$ qubits in a maximally entangled state with its reference $N$, 
in such an old black hole. Subsequently, after the black hole evolves 
unitarily and emits Hawking radiation, a receiver say Bob, can 
intercept the radiation and can decode Alice's information by waiting 
for no longer than $O(\log |M|)$ qubits of radiation. Hayden-Preskill 
applied the decoupling theorem, by assuming black hole evolution as 
a Haar random unitary, to conclude that for a `typical' black hole 
Alice's information can be decoded by Bob provided Bob has access to 
the initial black hole's environment, and emitted Hawking radiation. 
This is because if a black hole's internal state, after interaction with 
Alice's register and evolution, is decoupled with Alice's purifying 
environment, then Uhlmann's theorem gives a decoder that maps black 
hole's initial and received radiations to purifying register for Alice's 
reference. This information retrieval property, similar to channel 
coding problem was first observed by Hayden-Preskill with the above 
mentioned assumptions in \cite{Hayden_Preskill}.
This phenomenon was termed as information mirror or scrambling.

Now, if we assume that black holes must scramble efficiently, is this 
information retrieval property still generic to a typical but 
efficiently evolving black hole? That is, even assuming efficiently 
implementable dynamics, is the fraction of black holes that do not 
behave like ``information mirrors'' still exponentially small? We provide 
an affirmative answer to the above question.

\vspace*{6mm}

\textbf{Hayden-Preskill \cite{Hayden_Preskill} toy model setup:} 
\begin{itemize}

\item 
Alice's initial state: A register of $\log |M|$ qubits.
It is assumed to be in a maximally entangled state with a purifying 
reference ($N$), denoted by: 
$\ket{\Phi}^{MN}:=\frac{1}{|M|} \Sigma_{i=1}^{|M|} \ket{i}^M \ket{i}^N$.

\item 
Black hole is denoted by system $B_1$ (before interacting with 
Alice's register). It is assumed to maximally entangled with an 
environment $E$ (initial radiations), denoted by: 
$\ket{\Phi}^{B_1E}:=\frac{1}{|B_1|} 
\Sigma_{j=1}^{|B_1|}\ket{j}^{B_1}\ket{j}^E$.

\item 
Evolution of black hole is modelled by a unitary operator 
$U^{MB_1 \to RB_2}$. 
Thus $\mathbf{|M| \times |B_1|=|R| \times |B_2|}$.

\item 
Emitted radiations from black hole after evolution are denoted by 
quantum system $R$.

\item 
Final state of black hole is denoted by quantum system $B_2$. $B_2$ is 
inaccessible to Bob.  

\item 
The decoupling superoperator $\mathbf{\cT^{RB_2E \to B_2} := \Tr_{RE} } $,
as the systems $E$ and $N$ are in product state from the beginning.

\item 
The input state to the ``decoupling theorem'' 
$\Phi^{MN} \otimes \Phi^{B_1E}$.

\item 
The decoupling step works by analyzing the distance of
\begin{align*}
&\left[\cT^{REB_2 \to B_2} \otimes \I^{N}\right]
\left\{ (U^{MB_1 \to RB_2} \otimes I^{NE} ) \circ 
(\Phi^{MN} \otimes \Phi^{B_1E}) \right\}\\
&= \left[\Tr_{RE} \otimes \I^{N}\right]
\left\{ (U^{MB_1 \to RB_2} \otimes I^{NE} ) \cdot  
(\Phi^{MN} \otimes \Phi^{B_1E}) \right\}\\
&= \left[\Tr_{R} \otimes \I^{N}\right]
\left\{ (U^{MB_1 \to RB_2} \otimes I^{N} ) \circ  
(\Phi^{MN} \otimes \Phi^{B_1E}) \right\}
\end{align*}
with its expected value equal to $\pi^{B_2} \otimes \pi^{N}$.  This shows 
that the `average' state on the joint system $N,\;E,\;R$ is a pure 
state and hence there exists a decoding isometry 
$V_{\mathcal{D}}^{RE \to M}$ can map this pure state to $\Phi^{MN}$, 
thereby giving Bob the access to Alice's initial quantum information.
\end{itemize}
For the rest of this section we work with the following decoupling 
function:
\begin{equation}
\label{eq:bh}
f(U):=\left\| \left[\Tr_R \otimes \I^{N}\right] 
\left\{ (U^{MB_1 \to RB_2} \otimes I^{N} ) \cdot  
(\Phi^{MN} \otimes \pi^{B_1}) \right\} - \pi^{B_2} \otimes \pi^{N} 
\right\|_1
\end{equation}
We follow the proof technique of Theorem~\ref{thm:fqsw}.
\begin{itemize}
\item
The input states in the notation of the Theorem~\ref{thm:fqsw} are
defined as $ \rho^{MB_1N}:= \Phi^{MN} \otimes \pi^N$ and 
$
{(\trho')}^{MB_1N}:=(\xi^N)^{-1/4} \circ \rho^{MB_1 N}
$. Note that $|M|=|N|$.\\
The choice of the weighting matrix $\xi^N=\pi^N$ gives 
$H_2^\epsilon(A|R)_\rho = -\log \lVert ({\trho')}^2\rVert_2^2 = 
\log \frac{|B_1|}{|M|} $.
\item 
The state 
$
\omega^{R'B_2' B_2} := 
\left[\Tr_R \otimes \I^{R'B'_2}\right](\Phi^{R R'} \otimes 
\Phi^{B_2 B'_2})=\pi^R \otimes \Phi^{B_2 B'_2},
$

\item 
Take $\delta = 0$.
We get $(\omega'''_{\epsilon, 0})^{B_2}=\pi^{B_2}$. Hence
$\mathbb{H}_{\max}^\epsilon(B_2)_\omega=\log |B_2|$ as the 
reduced state $\omega^{B_2}=\pi^{B_2}$.

\item 
The matrix
$
(\tomega')^{R'B'_2B_2} :=
(I^{B'_2} \otimes I^{R'} \otimes 
(\omega'''_{\epsilon, 0})^{B_2}
)^{-1/4} \circ 
(\Phi^{B_2 B'_2} \otimes \pi^{R'}) =
|B_2|^{1/2} (\Phi^{B_2 B'_2} \otimes \pi^{R'}).
$
Note that
$
\mathbb{H}_2^{\epsilon, 0}(R'B'_2|B_2)_\omega = 
-2 \log \lVert (\tomega')^{R'B'_2 B_2} \rVert_2. 
$
We have,
\begin{align*}
\lVert 
(\tomega')^{R'B'_2 B_2}
\rVert_2 
&= 
\lVert |B_2|^{1/2} \Phi^{B_2 B'_2} \otimes \pi^{R'} \rVert_2 
\;= \;
\sqrt{\frac{|B_2|}{|R|}}, \\
\lVert 
(\tomega')^{B_2}
\rVert_2 
& =
\lVert 
|B_2|^{1/2} \Phi^{B_2} 
\rVert_2 
\;=\;
1, \\
\eta 
& = 
\frac{
\lVert (\tomega')^{R'B'_2 B_2} \rVert_2^2
}{
\lVert (\tomega')^{B_2} \rVert_2^2 
} 
\;=\;
\frac{|B_2|}{|R|}, \\
\delta_1 
& =
\lVert (\tomega')^{B_2} \rVert_2^2 
\frac{|R|^2|B_2|^2 - |R||B_2| \eta}{|B_2|^2|R|^2 - 1}
\;=\;
\frac{|B_2|^2 |R|^2 - |B_2|^2}{|B_2|^2 |R|^2 - 1}, \\
\delta_2
& =
\lVert (\tilde{\omega}')^{B'_2R' B_2}\rVert_2^2 
\frac{|B_2|^2|R|^2 - |B_2||R| \eta^{-1}}{|B_2|^2|R|^2 - 1}
\;=\;
\frac{|B_2|}{|R|} \cdot
\frac{|B_2|^2 |R|^2 - |R|^2}{|B_2|^2 |R|^2 - 1};
\end{align*} 

\item
The function
$
g(U) = 
|B_2|^{1/2}
\left\lVert
\left[\Tr_{R} \otimes \I^{N}\right]\left((U^{MB_1 \to RB_2} \otimes I^N) 
\circ (\trho')^{MB_1N}\right) -
\pi^{B_2} \otimes (\trho')^N
\right\rVert_2.
$
Then,
\begin{align*}
\E_{U \sim \Haar}[(g(U))^2]
& =
\delta_1 \lVert (\trho')^N \rVert_2^2 + 
\delta_2 \lVert (\trho')^{MB_1N} \rVert_2^2 -
\lVert (\tomega')^{B_2} \rVert_2^2 \lVert (\trho')^N \rVert_2^2 \\
& =
-\frac{|B_2|^2 - 1}{|B_2|^2 |R|^2 - 1}
\lVert (\trho')^N \rVert_2^2 + 
\frac{|B_2|}{|R|} \cdot
\frac{|B_2|^2 |R|^2 - |R|^2}{|B_2|^2 |R|^2 - 1}
\lVert (\trho')^{MB_1N} \rVert_2^2 \\
& \leq 
\frac{|B_2|}{|R|} \cdot
\lVert (\trho')^{MB_1N} \rVert_2^2=\frac{|M|^2}{|R|^2}, \\
\underset{{U \sim \Haar}}{\E}[(g(U))^2]
& =
-\frac{|B_2|^2 |R|^2 -  |R|^2}{(|B_2|^2 |R|^2 - 1) |R|^2}
\lVert (\trho')^N \rVert_2^2 + 
\frac{|B_2|}{|R|} \cdot
\frac{|B_2|^2 |R|^2 - |R|^2}{|B_2|^2 |R|^2 - 1}
\lVert (\trho')^{MB_1N} \rVert_2^2 \\
& \geq 
\frac{0.1 |B_2|}{|R|} \cdot
\frac{|B_2|^2 |R|^2 - |R|^2}{|B_2|^2 |R|^2 - 1}
\lVert (\trho')^{MB_1N} \rVert_2^2  \\
& \geq
\frac{0.1 |B_2|}{|R|} 
\left(1 - \frac{|R|^2}{|B_2|^2 |R|^2 }\right)
\lVert (\trho')^{MB_1N} \rVert_2^2 \\
& \geq 
\frac{0.07 |B_2|}{|R|} \cdot
\lVert (\trho')^{MB_1N} \rVert_2^2=0.07 \frac{|M^2|}{|R|^2}; 
\end{align*}

\item 
The tail probability exponent $a$ becomes
$
a = 
|RB_2| |B_2|^{-1} 2^{H_2^\epsilon(MB_1|N)_\rho - 9} = 
2^{-9} \frac{|R||B_1|}{|M|}.
$

\item
Define $\mu := \E_{\Haar}[g(U)]$. 
By a direct calculation we have\\
$
a \cdot \E_{U \sim \Haar}[(g(U))^2] +
\log \E_{U \sim \Haar}[(g(U))^2]  > \log |B_2| + \log |R| +
\mathbb{H}_{\max}^\epsilon(B_2)_\omega - H_2^\epsilon(M B_1|N)_\rho,
$\\
and thus we get $\mu^2 \leq \E[(g(U))^2] \leq 8\mu^2$;

\item
The qTPE parameter $t$ becomes
$
t = 
|M||B_1| \kappa^2 \cdot 2^{H_2^\epsilon(MB_1|N)_\rho - 6}\cdot |B_2|^{-1}.
$
For values of $\kappa = O(\mu)$, $t = O(|B_2|)$;

\item
The qTPE parameter $\lambda$ becomes
$
\lambda =
(|MB_1|^{-8} |B_2|^{-6} \mu^2)^t,
$
which satisfies
\[
(0.008 |R|^{-9} |B_2|^{-13} 2^{-H_2^\epsilon(M B_1|N)_\rho})^t <
\lambda < 
(|R|^{-9} |B_2|^{-13} 2^{-H_2^\epsilon(MB_1|N)_\rho})^t.
\]
Sequentially iterating 
$O(t (\log |M| + \log |B_1|))$ times an
$(|M||B_1|, s, O(1), t)$-qTPE gives us the desired
$(|M||B_1|, s, \lambda, t)$-qTPE for efficient evolution.
Observe that if $t \leq \polylog(|R|)$, then efficient
constructions for such qTPEs exist \cite{brandao2012local, sen:zigzag}.

\end{itemize}

We summarize our conclusion as follows: 
\begin{theorem}[FQSW concentration and black hole as mirror] 
\label{thm:haydenpreskill}
If an old black hole, that has already radiated half of its matter, 
evolves according to an approximate unitary $O(|B_2|)$-design then 
it satisfies the following concentration inequality:
\[
\prob_{U \sim \mathrm{design}}\left[
f(U) > 
\frac{|M|}{|R|}
+ 14 \sqrt{\epsilon} + 2 \kappa
\right] \leq 
7 \cdot 2^{-a \kappa^2},
\]
where the unitary $U_{BH}^{MB_1 \to RB_2}$ is chosen uniformly at 
random from an approximate 
$t$-design, 
$a := 2^{-9} \frac{|R||B_1|}{|M|}$ and
$t := 8 a \kappa^2$.
The above concentration inequality can also be expressed directly in 
terms of the parameter $t$ of the evolution of the black hole 
according to a unitary chosen uniformly at random from an approximate 
unitary $t$-design as:
\[
\prob_{U \sim \mathrm{design}}\left[
f(U) > 
\frac{|M|}{|R|}
+ 14 \sqrt{\epsilon} + \sqrt{t/2} \, \frac{|M|}{|R||B_1|}
\right] \leq 
7 \cdot 2^{-\frac{t}{8}},
\]

\end{theorem}
This theorem in turn implies that an overwhelming fraction of such 
black holes act as an information mirror. Further, one has to wait no 
longer than the size of Alice's register in order to decode Alice's 
information scrambled by the such a black hole.
Such a unitary can be described by using $O(|B_2| \log (|M| |B_1|))$ 
random bits as opposed to a Haar random black hole that needs 
$O(|M|^2 |B_1|^2 \log(|M| |B_1|))$ for a reasonable approximation.

Moreover, if $|B_2| \leq O(\polylog(|R|))$ and 
$
\kappa = O\left({\frac{|M|}{|R|}}\right)
$
which further implies that $t = O(|B_2|)=O(poly \log |R|)$,
then a typical black hole not only acts as an information mirror but 
can also be described in a computationally efficient manner 
(or simulated efficiently on a quantum computer). 

\section{Conclusion}
\label{sec:conclusion}
In this work we obtain a novel concentration result for
one-shot
non-catalytic decoupling via approximate unitary $t$-designs for
moderate values of $t$. Our bounds are stated in terms of
one-shot smooth variants of R\'{e}nyi 2-entropies and max-entropies.
We then consider the asymptotic iid limit of our concentration result
and show that the bounds reduce to the standard Shannon entropies.
Finally, we apply our concentration result to a special case when the 
superoperator $\cT$ is just the partial trace. This case is also 
referred to as the Fully Quantum
Slepian Wolf theorem in the literature, for its application to lossless 
quantum source compression demonstrated in \cite{mother_protocol}. This 
leads to a new result on relative thermalisation
of quantum systems. In particular for systems that are much smaller
than their ambient spaces, we show that for a wide range of
parameters relative
thermalisation can be achieved with probability exponentially close to one
using efficiently implementable unitaries. This is the first result
of this kind.

We also apply our FQSW result to the Hayden-Preskill toy model for the 
black hole information paradox. We show that the information mirror
property observed by Hayden and Preskill for a Haar random unitary
black hole continues to hold with probability exponentially close to one
even when the black hole evolution is restricted to be a uniformly random
unitary picked from a small discrete set of `simple' i.e. 
computationally efficient unitaries.

For larger systems, it is unknown
whether suitable efficient approximate $t$-designs exist.
Hence the question of whether relative thermalisation can be
achieved by efficiently implementable unitaries with exponentially
high probability in the general case still remains open.

Several applications of the original decoupling theorem in expectation
are known in the literature. Our result can be applied
to many of them obtaining, for the first time, corresponding 
concentration results via approximate unitary $t$-designs.
Whether these concentration results have any operational significance
is a topic left for future research.

\section*{Acknowledgement}
AN would like to thank Prof. Francesco Buscemi for pointing out Hayden-Preskill toy model for black hole information paradox as a potential application to our main theorem. We acknowledge support of the Department of Atomic Energy, Government of 
India, under project no. 12-R\&D-TFR-5.01-0500, for carrying out this 
research work. AN acknowledges support from the European Research Council (ERC Grant Agreement No. 948139) and also from MEXT Quantum Leap Flagship Program (MEXT QLEAP) Grant No. JPMXS0120319794.

%\IEEEtriggeratref{19}

\bibliographystyle{IEEEtran}
\bibliography{IEEEabrv,ConcDecoupling}

\end{document}